\documentclass[11pt]{amsart}
\usepackage{mathtools,amssymb, xparse,amsmath,amsthm,tikz-cd,slashed}
\usepackage{enumitem}
\usepackage{breqn}
\usepackage{graphicx}
\usepackage{hyperref}
\usepackage{cleveref}

\newcommand{\Lx}{\left(}
\newcommand{\Rx}{\right)}

\newcommand{\R}{\mathbb R}

\newcommand{\Z}{\mathbb Z}

\newcommand{\abs}[1]{\left\lvert #1 \right\rvert}

\newcommand{\partv}[1]{\Lx \lambda^{-1}\partial_v\Rx^{#1}}
\newcommand{\partu}[1]{\Lx\nu^{-1}\partial_u\Rx^{#1}}

\newtheorem{theorem}{Theorem}[section]
\newtheorem{lemma}[theorem]{Lemma}

\newtheorem{proposition}[theorem]{Proposition}
\newtheorem{definition}[theorem]{Definition}
\newtheorem{corollary}[theorem]{Corollary}

\newtheorem*{theorem*}{Theorem}
\newtheorem*{lemma*}{Lemma}
\newtheorem*{hypothesis*}{Hypothesis}
\newtheorem*{condition*}{Condition}
\newtheorem*{proposition*}{Proposition}
\newtheorem*{definition*}{Definition}
\newtheorem*{corollary*}{Corollary}
\newtheorem*{conjecture*}{Conjecture}
\newtheorem*{claim*}{Claim}

\theoremstyle{remark}
\newtheorem{remark}[theorem]{Remark}

\numberwithin{equation}{section}

\DeclarePairedDelimiterX{\set}[1]{\{}{\}}{\setargs{#1}}
\NewDocumentCommand{\setargs}{>{\SplitArgument{1}{;}}m}
{\setargsaux#1}
\NewDocumentCommand{\setargsaux}{mm}
{\IfNoValueTF{#2}{#1} {#1\,\delimsize|\,\mathopen{}#2}}

\begin{document}
\title[Global Non-linearly Stable Solutions to ESF]{Global Non-linearly Stable Large-Data Solutions to the Einstein Scalar Field System}
\author{Eric Kilgore}
\date{}

\begin{abstract}
I study a class of global, causal geodesically complete solutions to the spherically symmetric Einstein scalar field (SSESF) system . Extending results of Luk-Oh (\textit{Quantitative Decay Rates For Dispersive Solutions to the Einstein-Scalar Field System in Spherical
Symmetry, 2015}), Luk-Oh-Yang (\textit{Solutions to the Einstein-Scalar-Field System in Spherical
Symmetry with Large Bounded Variation Norms, 2018}), I provide new bounds controlling higher derivatives of both the metric components of the solution and the scalar field itself for large data solutions to SSESF. Moreover, by constructing a particular set of generalized wave-coordinates, I show that, assuming sufficient regularity of the data, these solutions are globally non-linearly stable to non-spherically symmetric perturbations by recent results of Luk and Oh. In particular, I demonstrate the existence of a large collection of non-trivial examples of large data, globally nonlinearly stable, dispersive solutions to the Einstein scalar field system.
\end{abstract}

\maketitle

\section{Introduction}
I study the decay properties of a class of spherically symmetric solutions $(M,g,\phi)$ to the Einstein scalar field system, for $M$ a $3+1$ dimensional manifold, $g$ a Lorentzian metric, and $\phi:M \rightarrow \R$ a real valued scalar field:
\begin{equation*}
\label{ESF}
\tag{ESF}
\begin{cases} R_{\mu\nu} - \frac{1}{2}\textbf{g}_{\mu\nu}R = 2T_{\mu\nu}\\ \nabla^{\mu}\partial_{\mu}\phi = 0 \end{cases}.
\end{equation*}
Recently Luk and Oh in \cite{Luk20} proved a large data stability criterion for solutions to \eqref{ESF}, in this paper I prove the following
\begin{theorem}
There exist large-data solutions satisfying the stability criterion of \cite{Luk20}.
\end{theorem}
From this we obtain the immediate corollary:
\begin{corollary}
There exists an open set of large initial data for ESF which gives rise to dispersive solutions.
\end{corollary}
This is the first such result in the large data case, previously there is no known existence result for global dispersive solutions to \eqref{ESF} outside of spherical symmetry \cite{Luk2018}, or the small data regime \cite{ChrisKlain,LindRod} (note, however, the spectacular recent advances in the stability of the black hole problem, see \cite{DafHolzRodTay,KlainSzef,HintzVasy}).


I build off the results of Luk--Oh--Yang in \cite{Luk2018} in which a large class of large data spherically symmetric solutions to the Einstein-scalar field system are constructed to the future of a cone \footnote{In \cite{Luk2018}, Luk--Oh--Yang also construct global spacetimes by solving a scattering problem from past null infinity. For technical reasons I do not directly control such global spacetimes, but instead start from the spacetimes defined to the future of a cone. It thus remains open whether there are future and past complete spacetimes satisfying the assumptions of Luk--Oh's result \cite{Luk20} in both the future and past directions.}, with decay estimates both towards null and timelike infinity up to second derivatives in a spherically symmetric double null coordinate system:
\begin{itemize}
\item I begin by improving the estimates of Luk--Oh--Yang to all order of derivatives, since the stability result of Luk--Oh, \cite{Luk20}, requires $\geq 11$ derivatives. Given the basic control of the spacetime geometry that has been established in \cite{Luk2015}, these estimates follow from the methods of \cite{Luk2018}, with some additional care given to controlling terms near the axis of symmetry.
\item Next, I extend the Luk--Oh--Yang spacetime (originally only to the future of a cone) to a larger spacetime defined to the future of an asymptotically flat Cauchy hypersurface. This is realized by solving a ``sideways'' characteristic initial value problem towards spatial infinity. This builds upon the work of Dafermos in \cite{Dafermos2003}. As before, I prove decay estimates for all higher order derivatives.
\item In this spherically symmetric spacetime, I then introduce a future-normalized spherically symmetric double-null gauge in which one obtains some slightly stronger decay estimates. It is at this point when one identifies the logarithmic terms arising from the contribution of the mass, which play an important role in the generalized wave coordinates that are later introduced.
\item Finally, I introduce a generalized wave coordinate system and show that the constructed spacetime satisfies the estimates required by Luk--Oh in \cite{Luk20}, which are defined in terms of commuting vector fields in terms of the generalized wave coordinates. Here, one must take advantage of the null condition (manifested in the different decay estimates for the $\partial_v$ and $\partial_u$ derivatives) and also be careful about the regularity of the solution at the axis (since the function $r$ is not smooth at the axis).
\end{itemize}

\subsection{Prior Results}

In the early 1990's, Christodoulou-Klainerman in \cite{ChrisKlain} established the first global non-linear stability results in the asymptotically flat setting, showing that Minkowski spacetime is globally non-linearly stable under the Einstein vacuum equations; see also \cite{FriedStab}. In particular they showed that for asymptotically flat initial data which are sufficiently close to Minkowski, the maximal globally hyperbolic development is causally geodesically complete and approaches Minkowski at large times. Lindblad-Rodnianski \cite{LindRod} later simplified the proof, and extended the result to the Einstein scalar field system. These results both require explicit smallness of the data. Recently Luk-Oh further extended the results and techniques of \cite{LindRod} in \cite{Luk20}, to give a set of criteria for large data causally geodesically complete solutions to be globally non-linearly stable for the Einstein scalar field system. These criteria consist of decay conditions for high order derivatives of the geometry and scalar field, as well as the existence of a gauge in which the solution satisfies some specific asymptotic relations. This was the first such stability result for large data.


The spherically symmetric Einstein scalar field system has been studied extensively over the past four decades, being among the most accessible systems containing matter in the spherically symmetric class. Through the 1980's and 1990's Christodoulou established a complete picture of the singularity structure of spherically symmetric solutions (cf \cite{Chris86,Chris87,Chris91,Chris93,Chris94,Chris99}). In particular, he showed that generic (in BV) initial data gives rise to a solution which is either dispersive, or contains a black hole region and a spacelike curvature singularity. This work established both a complete understanding of the singularity structure of the spherically symmetric Einstein scalar field, and qualitative description of the long term dynamics of the system. Moreover \cite{Chris86} gives quantitative control for small data. It remained to establish quantitative bounds in the large data case.

There has since been significant progress in this direction. In the black hole case this was done by Dafermos-Rodnianski in \cite{DafRod}, in which they established polynomial decay rates conjectured by Price \cite{Price}. Corresponding lower bounds have more recently been established in \cite{luk2019strong}. We will focus on the dispersive case, which was studied by Luk-Oh in \cite{Luk2015}. They establish quantitative decay rates for $\phi$ as well as the geometry of the system for up to $C^2$ solutions, without (quantitative) restriction on the BV norm of the data. The new high order decay established in this paper is an extension of these results to $C^k$ data for arbitrary $k \geq 2$ (controlling up to $k$ derivatives of $\phi$ and geometric terms).

The above references provide only an incomplete picture of the full development of the current understanding of the Einstein scalar field system. For a more complete collection see \cite{Luk20}.

\subsection{Outline of the Paper}
In \cref{secprelim} we will lay out the problem, establish our notations and conventions, and restate some prior results (from \cite{Dafermos2003,Luk2015}) which will be referenced later. In \cref{secresults} we will provide precise statements of the primary results of the paper (cf. \cref{main1,main2}).

The remainder of the paper is devoted to proving the results stated in \cref{secresults}. In \cref{Qproof,oproof} we will \cref{main1}. In particular \cref{Qproof} provides estimates at arbitrary differential order in a region with compact curves of constant $u$ (respectively $v$). In \cref{oproof} I establish the estimates of \cref{main1} in a region away from the axis of symmetry. Together these complete the proof of \cref{main1}.

Finally, \cref{axisregsec,stabproof} are devoted to proving global nonlinear stability of the spacetimes considered in \cref{main1} (\cref{main2}), from the results of \cref{main1}. In \cref{axisregsec} we show that our estimates in spherical symmetry extend nicely to the full, un-reduced spacetime, and in \cref{stabproof} we check that the class of solutions considered in \cref{main2} in fact satisfies the conditions of the main theorem of \cite{Luk20} and are thus exhibit global non-linear stability to (an open class of) non-spherically symmetric perturbations.

\section{Preliminaries}
\label{secprelim}
In this section we will go over the set up of the problem at hand, the form of the equations and coordinates we will use, some important terms and useful machinery, and prior results we will use throughout the paper.

\subsection{The Einstein Scalar Field System}
We begin with an overview of the deriving of the spherically symmetric Einstein scalar field system. We begin with the full Einstein scalar field system in 3+1 dimensions.

Solutions are described by a triple $(\mathcal{M},\textbf{g}_{\mu\nu},\phi)$ where $(\mathcal{M},\textbf{g}_{\mu\nu})$ is a $(3+1)$ dimensional Lorentzian manifold, and $\phi$ is a real-valued function on $\mathcal{M}$. The metric and scalar field satisfy the Einstein scalar field system:
\begin{equation*}
\tag{ESF}
\begin{cases} R_{\mu\nu} - \frac{1}{2}\textbf{g}_{\mu\nu}R = 2T_{\mu\nu}\\ \nabla^{\mu}\partial_{\mu}\phi = 0 \end{cases}
\end{equation*}
where $R_{\mu\nu}$ is the Ricci curvature of $\textbf{g}$, $R$ is the scalar curvature, $\nabla_{\mu}$ is the covariant derivative corresponding to the Levi-Civita connection on $(\mathcal{M},\textbf{g})$, and $T_{\mu\nu}$ is the energy-momentum tensor given by $\phi$:
\[ T_{\mu\nu} = \partial_\mu\phi\partial_{\nu}\phi - \frac{1}{2}\textbf{g}_{\mu\nu}\partial^{\lambda}\phi\partial_{\lambda}\phi. \]
Assume that $(\mathcal{M},\textbf{g},\phi)$ admits a smooth action of $SO(3)$ by isometries on $(\mathcal{M},\textbf{g})$ such that each orbit is either a point, or isometric to $S^2$ with a round metric, and $\phi$ is constant on each orbit. Such a solution is called \textit{spherically symmetric}. These properties are all propagated by \eqref{ESF}, thus if $(\mathcal{M},\textbf{g},\phi)$ is a Cauchy development of some initial data, it suffices to assume that the initial data is spherically symmetric to ensure spherical symmetry of the solution.

Under this assumption we can take the quotient $\mathcal{M}/SO(3)$ which yields a $(1+1)$ dimensional Lorentzian manifold with boundary which we will denote by $(M,g)$. The boundary $\Gamma$ is the set of fixed points of the $SO(3)$ action.

In this setup we can define the \textit{area radius function} $r$ on $M$ by
\[ r := \sqrt{\frac{\text{area of symmetry sphere}}{4\pi}} \]
with $r = 0$ on $\Gamma$. Note that each connected component of $\Gamma$ is a timelike geodesic.

Henceforth we will assume that $\Gamma$ is non-empty and connected, and that there exists a system of global double null coordinates $(u,v)$, in which the metric takes the form
\[ g_{ab}dx^adx^b = -\Omega^2 dudv \]
for some $\Omega > 0$. Both of these assumptions are certainly justified so long as $(\mathcal{M},\textbf{g})$ is a Cauchy development of a spacelike hypersurface homeomorphic to $\R^3$.

We can recover the metric $\textbf{g}$ on $\mathcal{M}$ from $\Omega$ and $r$:
\[ \textbf{g}_{\mu\nu}dx^{\mu}dx^\nu = -\Omega^2 dudv + r^2 ds^2_{S^2} \]
where $ds^2_{S^2}$ is the line element for the unit sphere $S^2 \subset \R^3$. From this we can reformulate our inherited equations on $(M,g)$ as the spherically symmetric Einstein scalar field system \eqref{SSESF1} in terms of the triple $(\phi,r,\Omega)$ as
\begin{equation*}
\tag{SSESF}
\label{SSESF1}
\begin{cases}
\partial_u\partial_vr = -\partial_ur\partial_vr - \frac{1}{4}\Omega^2\\
r^2\partial_u\partial_v\log\Omega = \partial_ur\partial_vr + \frac{1}{4}\Omega^2 - r^2\partial_u\phi\partial_v\phi\\
r\partial_u\partial_v\phi = -\partial_ur\partial_v\phi - \partial_vr\partial_u\phi\\
2\Omega^{-1}\partial_ur\partial_u\Omega = \partial_u^2r + r(\partial_u\phi)^2\\
2\Omega^{-1}\partial_v\partial_v\Omega = \partial_v^2r + r(\partial_v\phi)^2
\end{cases}
\end{equation*}
with the boundary condition $r = 0$ on $\Gamma$.

We can reformulate this problem once more in terms of the \textit{Hawking mass} $m$ defined by
\begin{equation} 1 - \frac{2m}{r} = \partial^ar\partial_ar = 4\Omega^{-2}\partial_ur\partial_vr. \end{equation}
We also define the \textit{mass ratio}:
\[ \mu := \frac{2m}{r}, \]
and introduce some shorthand for important derivatives of $r$:
\[ \lambda := \partial_vr \qquad \nu := \partial_ur. \]
With this in mind we see that we can reformulate \eqref{SSESF1} in terms of the triple $(\phi,r,m)$ in the following way:

We say that $(\phi,r,m)$ is a solution to \eqref{SSESF1} if the following relations hold:
\begin{equation}
\tag{SSESF'}
\label{SSESF}
\begin{cases}
\partial_u\partial_vr = \frac{2m\lambda\nu}{(1 - \mu)r^2}\\
\partial_u\partial_v(r\phi) = \frac{2m\lambda\nu}{(1-\mu)r^2}\phi\\
\nu^{-1}\partial_um = \frac{1}{2}(1-\mu)r^2(\nu^{-1} \partial_u \phi)^2\\
\lambda^{-1}r\partial_vm = \frac{1}{2}(1-\mu)r^2(\lambda^{-1} \partial_v\phi)^2
\end{cases}
\end{equation}
and moreover, $r = m = 0$ on $\Gamma$.

Observe that we can equivalently write our equation for $r\phi$ as:
\begin{equation}
\begin{cases}
\partu{}\partv{}(r\phi) = \frac{2m}{(1-\mu)r^2}\phi - \frac{2m}{(1-\mu)r^2}(\lambda^{-1}\partial_v\phi)\\
\partv{}\partu{}(r\phi) = \frac{2m}{(1-\mu)r^2}\phi - \frac{2m}{(1-\mu)r^2}(\nu^{-1}\partial_u\phi)
\end{cases}.
\end{equation}

\subsection{Notation and Conventions}
Here we will write down some notation and assumptions that will carry throughout the remainder of the paper.

We begin with a more concrete definition of the reduced space we work with.

Let $\R^{1+1}$ denote the $(1+1)$ dimensional Minkowski space with standard double null coordinates $(u,v)$. Let $M$ be a $(1+1)$ dimensional Lorentzian manifold conformally embedded in $\R^{1+1}$ with $ds^2_M = -\Omega^2dudv$. From $r$ a non-negative function on $M$ define the set
\[ \Gamma := \set*{(u,v) \in M; r(u,v) = 0}. \]
Define $\mathcal{M} = M \times S^2/\sim$ where $\sim$ is the equivalence
\[ (u,v,s) \sim (u,v,s') \]
if and only if $(u,v) \in \Gamma, s,s' \in S^2$. This is the full $3+1$ dimensional space above $M$.

We assume that $\Gamma$ is connected, the image of a future-directed timelike curve emanating from $(1,1)$. We also assume that $C_R,\underline{C}_{R} \subset M$ for all $R > 0$. Where
\[ C_{R} = \set*{(u,v) \in \R^{1+1}; u = R, R \leq v < \infty} \qquad \underline{C}_R = \set*{(u,v) \in \R^{1+1};v = R, -\infty < u \leq v}. \]

Finally, define past and future null infinity (denoted by $\mathcal{I}^{-},\mathcal{I}^+$ respectively) to be the sets of ``points'' $(-\infty,v)$, and $(u,\infty)$ respectively s.t. $\sup_{C_u}r = \sup_{\underline{C}_v}r = \infty$.

Combining this with the Hawking mass above, we define the \textit{Bondi mass} as $M_u = \lim_{v \rightarrow \infty}m(u,v)$. The \textit{final Bondi mass} $M_f = \lim_{u \rightarrow \infty} M_u$ and the \textit{initial Bondi mass} $M_i = \lim_{u \rightarrow -\infty} M_u$.

We also outline our convention for integrating over curves in $M$. When integrating over $C_u$ or $\underline{C}_v$ we write
\[ \int_{C_u \cap \set*{v_1 \leq v \leq v_2}}f = \int_{v_1}^{v_2}f(u,v')dv' \]
\[ \int_{\underline{C}_v \cap \set*{u_1 \leq u \leq u_2}}f = \int_{u_1}^{u_2}f(u',v)du'. \]

We define the \textit{domain of dependence} of a line segment $C_{u_0} \cap \set*{v \leq v_0}$ (which we denote $\mathcal{D}(u_0,v_0)$) to be the set of points $p \in M$ such that all past-directed causal curves through $p$ intersect $\Gamma \cup (C_{u_0} \cap \set*{v \leq v_0})$ along with the segment $C_{u_0} \cap \set*{v \leq v_0}$.

\subsection{Gauge Conditions}
Observe that, to this point, the coordinates $u,v$ are free to be reparametrized by transformations of the form
\[ u \mapsto \tilde{u}(u) \quad v \mapsto \tilde{v}(v) \quad \tilde{u}(1) = \tilde{v}(1) = 1 \]
for any monotone increasing $\tilde{u}, \tilde{v}$. To fix these coordinates we must prescribe some \textit{gauge condition}.

In what follows we will consider the following three different conditions:
\begin{enumerate}[label=(G\arabic*)]
\item \label{gauge1} $\lambda \equiv \frac{1}{2}$ on $C_{1}$, $\Gamma = \set*{(u,v);u = v}$. 
\item \label{gauge2} $\lim_{v \rightarrow \infty}\nu(u,v) = -\frac{1}{2}$ for all $u$, and $\Gamma = \set*{(u,v);u = v}$.
\item \label{gauge3} $\lambda \equiv \frac{1}{2}$ on $C_1$ and $\nu \equiv -\frac{1}{2}$ on $\underline{C}_R$ for some $R \gg 1$
\end{enumerate}
\begin{remark}
The gauge \ref{gauge2} can be obtained from \ref{gauge1} by the following transformations:
\begin{gather}
	\tilde{u}_2(u_1,v_1) = -2\int_1^{u_1} \bar{\nu}(u_1')du_1' \qquad \tilde{v}_2(u_1,v_1) = -2\int_1^{v_1} \bar{\nu}(v_1')dv_1'
\end{gather}
where $\bar{\nu}(u) := \lim_{v \rightarrow \infty}\nu(u,v)$.

\ref{gauge3} can be obtained via a similar transformation fixing $v = R$ rather than taking this limit.
\end{remark}

\begin{remark}
In both the gauges \ref{gauge1} and \ref{gauge2} for $r$, $\phi$ $C^k$-smooth on $M$ the following hold
\[ \lim_{v \rightarrow u^+}(\partial_u + \partial_v)^lr(u,v) = \lim_{u \rightarrow v^-}(\partial_u + \partial_v)^lr(u,v) = 0 \]
\[ \lim_{v \rightarrow u^+}(\partial_u + \partial_v)^l(r\phi)(u,v) = \lim_{u \rightarrow v^-}(\partial_u + \partial_v)^l(r\phi)(u,v) = 0 \]
for all $l \leq k$.
\end{remark}
\begin{remark}
Also under either of \ref{gauge1} or \ref{gauge2} the domain of dependence $\mathcal{D}(u_0,v_0)$ has the form
\[ \mathcal{D}(u_0,v_0) = \set*{(u,v) \in M; u \in [u_0,v_0],v \in [u,v_0]}. \]
\end{remark}

\subsection{The Characteristic Initial Value Problem}
We are now ready to pose the problem on which our analysis will focus throughout the first half of the paper. Of course the equation we must satisfy is given by \eqref{SSESF} but it remains to specify the precise notions of solution, and specify the constraints on initial data that we will consider.

This initial data is quite constrained by \eqref{SSESF}. In fact, to obtain a solution in all of $\R^{1+1}_{x \geq 0}$ it suffices to pose data for $\partial_v(r\phi)$ and $\partial_u(r\phi)$ on the characteristic curves $C_1$ and $\underline{C}_R$ respectively, as well as the values $\phi(1,1),m(1,1)$ (which we will take to be 0), and a choice of gauge. The data on $C_1$ then completely determines the solution in the region $\mathcal{D}_{C_1}$, becuase our gauge, in combination with the equation determines data for $m$ and $r$ as well. Of course on $\underline{C}_R$ we must check that our data in the region $u \geq 1$ is compatible with that posed for $\underline{C}_{R}$. This is only a local constraint, in that we can prescribe whatever data we want for $u < 1-\epsilon$ for any $\epsilon > 0$, and ($C^k$)-smoothly interpolate in the region $[1-\epsilon,1]$.

\begin{definition}[$C^k$ solutions to \eqref{SSESF} (in \ref{gauge1})]
\label{charinit}
A solution $(\phi,r,m)$ to \eqref{SSESF} is called a $C^k$ solution on $M$ if the following holds on every domain of dependence $\mathcal{D}(u_0,v_0)$:
\begin{enumerate}
\item $\sup_{\mathcal{D}(u_0,v_0)}(-\nu), \sup_{\mathcal{D}(u_0,v_0)}\lambda^{-1} < \infty$.
\item $\lambda, \nu$ are $C^k$ on $\mathcal{D}(u_0,v_0)$.
\item For each $(u,u) \in \Gamma$
\[ \lim_{\epsilon \rightarrow 0^+}(\nu + \lambda)(u,u+\epsilon) = \lim_{\epsilon \rightarrow 0^+}(\nu + \lambda)(u-\epsilon,u) = 0. \]
\item $\partial_v(r\phi), \partial_u(r\phi)$ are $C^k$ on $\mathcal{D}(u_0,v_0)$.
\item For each $(u,u) \in \Gamma$
\[ \lim_{\epsilon \rightarrow 0^+}(\partial_v(r\phi) + \partial_u(r\phi))(u,u+\epsilon) = \lim_{\epsilon \rightarrow 0^+}(\partial_v(r\phi) + \partial_u(r\phi))(u-\epsilon,u) = 0. \]
\end{enumerate}
\end{definition}
Note that data which is $C^k$ on $C_1, \underline{C}_R$ will give rise to a $C^k$ solution.

\subsection{Definitions and Prior Results}
Here we state some longer definitions of terms we will use throughout the paper.

It is essential to obtaining the desired decay that the prescribed initial data already verifies such an estimate. As such we make the following definition.
\begin{definition}[Asymptotic Flatness of Order $\omega' \geq 0$ in $C^k$]
\label{asymflat}
An initial data set is said to be asymptotically flat of order $\omega'$ in $C^k$ towards $\mathcal{I}^+$ (resp. $\mathcal{I}^-$) if $\partial_v(r\phi)(1,\cdot) \in C^k[1,\infty)$ ($\partial_u(r\phi)(\cdot,R) \in C^k(-\infty,1]$) and there exists $A>0$ ($A'$) such that
\begin{gather}
\sup_{C_1}v^{\omega' + l}\abs{\partial^{l+1}_v(r\phi)} \leq A < \infty\\
\sup_{\underline{C}_R}(1+\abs{u})^{\omega'+l}\abs{\partial^{l+1}_v(r\phi)} \leq A' < \infty
\end{gather}
for all $l \leq k$.
\end{definition}

The solutions we consider will also be required to satisfy an additional constraint:
\begin{definition}[Local Scattering]
\label{localbvscat}
A $C^k$ solution $(\phi,r,m)$ on $M$ is said to be \textit{locally scattering} if the following holds:
\begin{enumerate}
\item The full $(3+1)$ dimensional solution $(\mathcal{M},\textbf{g},\phi)$ is future causally geodesically complete.

\item There exists $r_0 > 0$ such that
\[ \int_{C_u \cap \set*{r \leq r_0}}\abs{\partial_v^2(r\phi)} \rightarrow 0, \qquad \int_{C_u \cap \set*{r \leq r_0}} \abs{\lambda^{-1}\partial_v\lambda} \rightarrow 0 \]
as $u \rightarrow \infty$.
\end{enumerate}
\end{definition}

Throughout what follows we will frequently be writing down bounds for various quantities. We will not often write out explicit constants, as all will be taken to be the same in the end and depend on the same quantities. To this end we will write $\partial^{\alpha}A \lesssim B$ to mean that $A$ is bounded by $B$ up to a constant depending on the constants $A,A'$ above, the order of derivative $\abs{\alpha}$, the initial Bondi mass $M_{-\infty}$, and some more complicaated quantities depending on low order behavior of solutions \footnote{For details of these dependencies see \cite{Luk2015}}. The important point is that our constants are global, in particular they have no dependence on coordinates.

Here we will give some names to various spacetime regions we will be interested in. Let
\begin{gather*}
\mathcal{Q} := \set*{(u,v) \in M; u \geq 1, v \geq u}\\
\mathcal{O}_R := \set*{(u,v) \in M; u \leq 1, v \geq R}\\
\mathbb{I} := \set*{(u,v) \in M; v \geq \abs{u}}.
\end{gather*}
\begin{remark}
The region $\mathbb{I}$ corresponds to essentially to the future of the curve $u+v = 0$ in this spacetime.
\end{remark}

We will also be interested in decay of solutions along certain vector fields, in particular those generating the symmetries of Minkowski space. We introduce some shorthand for these (vector fields on $\mathcal{M}$). In what follows in this section let $(t,x^1,x^2,x^3)$ be coordinates on the manifold(-with-boundary) $[0,\infty) \times \R^3$, $\tilde{r} = \sqrt{\sum_{i = 1}^3(x^i)^2}$. We will write latin indices $i,j \in \set*{1,2,3}$ for spatial coordinates, and greek indices $\alpha,\beta \in \set*{0,1,2,3}$ for spacetime coordiantes (in what follows $t = x^0$).
\begin{definition}[Minkowski Commuting Vector Fields]
\label{mcvec}
The Minkowskian commuting vectors fields are the set of vector fields
\[ \set*{\partial_{\mu},x^i\partial_j - x^j\partial_i,t\partial_i + x^i\partial_t, S := t\partial_t + \sum_{i  =1}^3x^i\partial_i}. \]
\end{definition}
We will henceforth use $\Gamma$ to denote a general Minkowskian commuting vector field. For a multi-index $I = (i_1,i_2,\dotsc,i_{\abs{I}})$, $\Gamma^I$ denotes a product of $\abs{I}$ Minkowksian vector fields.

We will also write, in these coordinates $\tilde{r} = \sqrt{\sum (x^i)^2}$.

We also define the vector fields $L,\underline{L},E^1,E^2,E^3$:
\begin{definition}
Let $\partial_{\tilde{r}} = \sum_{i = 1}^3 \frac{x^i}{\tilde{r}}\partial_i$. Define
\[ L = \partial_t + \partial_{\tilde r}, \quad \underline{L} = \partial_t - \partial_{\tilde r}. \]
Define the vector fields
\[ \set*{E_1,E_2,E_3} := \set*{\frac{x^2}{\tilde r}\partial_3 - \frac{x^3}{\tilde r}\partial_2, \frac{x^1}{\tilde r}\partial_3 - \frac{x^3}{\tilde r}\partial_1, \frac{x^1}{\tilde r}\partial_2 - \frac{x^2}{\tilde r}\partial_1} \]
tangent to the coordinate 2-spheres of constant $\tilde r,t$.
\end{definition}

Later, we will use coordinates $(s,q,\theta,\phi)$ where $(\theta,\phi)$ the standard spherical coordinates, and $(s,q)$ are given by
\[ s = t+\tilde r, \quad q = \tilde r-t. \]
This leads us immediately to the vector fields
\[ \partial_s = \frac{1}{2}(\partial_t + \partial_{\tilde r}), \quad \partial_q = \frac{1}{2}(\partial_{\tilde r} - \partial_t). \]

We introduce a little more notation before we are ready for our final definition. For a scalar function $\xi$ we write
\[ \abs{\partial\xi} := \sum_{\mu = 0}^3 \abs{\partial_{\mu}\xi}^2, \quad \abs{\bar\partial \xi}^2 := \abs{\partial_s \xi}^2 + \frac{1}{2}\sum_{i,j = 1}^3 \Lx \frac{x^i}{\tilde r}\partial_j \xi - \frac{x^j}{\tilde r}\partial_i \xi \Rx^2 =: \abs{\partial_s\xi^2} + \abs{\slashed{\nabla}\xi}^2. \]

Next, for a 2-tensor $p$ define
\[ \abs{p}^2 = \sum_{0 \leq \mu,\nu \leq 3} \abs{p_{\mu\nu}}^2 \]

Define also for a first order differential operator $D$, $Dp$ to be $D$ applied component-wise to $p$: $(Dp)_{\mu\nu} = D(p_{\mu\nu})$.

With this in mind we make the following definition:
\begin{definition}
Let $\mathcal{T} = \set*{L,E^1,E^2,E^3}$, $\mathcal{U} = \set*{L,\underline{L},E^1,E^2,E^3}$, $\mathcal{L} = \set*{L}$. Then for any two of these families $\mathcal{V},\mathcal{W}$ (allowing repeats), and a 2-tensor $p$ define
\[ \abs{p}^2_{\mathcal{V}\mathcal{W}} = \sum_{V \in \mathcal{V},W \in \mathcal{W}}\abs{p_{\alpha\beta}V^{\alpha}W^{\beta}}^2. \]
\end{definition}

Now we make the following definition as in \cite{Luk20}.
\begin{definition}[Dispersive Spacetime Solution of Size $(C,\gamma_0,N)$]
\label{dispsoln}
Let $\gamma_0 > 0$ be a real number and $N \geq 11$ be an integer. A spacetime $(\mathcal{M} = [0,\infty)\times \R^3,g)$ with scalar field $\phi:\mathcal{M} \rightarrow \R$ is a dispersive spacetime solution of size $(C,\gamma_0,N)$ if
\begin{enumerate}[label=(D\arabic*)]
	\item \label{disp1} The triple $(\mathcal{M},g,\phi)$ is a solution to the Einstein scalar field system.
	\item \label{disp2} There exists a global system of coordinates $(t,x^1,x^2,x^3)$ such that with respect to this coordinate system the metric takes the form
	\[ g - m = h \]
	where
	\[ m = -dt^2 + \sum_{i = 1}^3 (dx^i)^2 \]
	is the Minkowski metric and $h$ satisfies the bound
	\[ \abs{\Gamma^Ih} \leq \frac{C\log(2+s)}{1+s} \]
	for $\abs{I} \leq N+1$ where $I$ a multi-index, $\Gamma$'s are the Minkowski commuting vector fields defined above.
	\item \label{disp3} For $\abs{I} \leq N+1$ we have
	\[ \abs{\partial \Gamma^Ih} \leq \frac{C}{(1+s)(1+\abs{q})^{\gamma_0}} \]
	for any combination of Minkowski commuting vector fields $\Gamma$.
	\item \label{disp4} For $\abs{I} \leq N+1$ we have
	\[ \abs{\bar\partial \Gamma^I h} \leq \frac{C}{(1+s)^{1+\gamma_0}} \]
	for any combination of Minkowskian commuting vector fields $\Gamma$.
	\item \label{disp5} For $\abs{I} \leq 1$ the following components satisfy bounds:
	\[ \sum_{\abs{I} \leq 1} \abs{\Gamma^Ih}_{LL} + \abs{h}_{L\mathcal{T}} \leq \frac{C}{(1+s)^{1+\gamma_0}}. \]
	\item \label{disp6} For $\abs{I} \leq N+1$ we have
	\[ \abs{\partial\Gamma^I\phi} \leq \frac{C}{(1+s)(1+\abs{q})^{\gamma_0}},\quad \abs{\bar\partial\Gamma^I\phi} \leq \frac{C}{(1+s)^{1+\gamma_0}}. \]
	\item \label{disp7} The metric $g$ is everywhere Lorentzian with uniformly bouned inverse
	\[ \abs{g^{-1}} \leq C. \]
	Let $(\hat{g})_{ij}$ be the restriction of the metric $g$ on the tangent space to the constant $t$-hypersurfaces ($i,j = 1,2,3$). $(\hat{g})_{ij}$ satisfies the condition that for any $\xi_i$,
	\[ C^{-1}\abs{\xi}^2 \leq \sum_{i,j = 1}^3 (\hat{g}^{-1})^{ij}\xi_i\xi_j \leq C \abs{\xi}^2 \]
	where
	\[ \abs{\xi}^2 = (\xi_1)^2 + (\xi_2)^2 + (\xi_3)^2. \]
	Moreover the spacetime gradient of $t$ is timelike and satisfies
	\[ (g^{-1})^{00} = (g^{-1})^{\alpha\beta}\partial_\alpha t\partial_\beta t \leq -C^{-1} < 0. \]
	\item \label{disp8} For $\abs{I} \leq N+1$ the global coordinate functions satisfy the estimate
	\[ \abs{\Gamma^I(\Box_{g}x^{\mu})} \leq \frac{C\log(2+s)}{(1+s)^2} \]
	where $\Box_{\textbf{g}}$ is the Laplace-Beltrami operator associated to $g$:
	\[ \Box_{\textbf{g}} = \frac{1}{\sqrt{-\det g}}\partial_{\alpha}((g^{-1})^{\alpha\beta}\sqrt{-\det g}\partial_\beta). \]
\end{enumerate}
\end{definition}

Finally, we are ready to state the main theorem of \cite{Luk20}
\begin{theorem*}[Large Data Stability, Luk \& Oh]\label{stabtheorem}
Let $N \geq 11$ and $0 < \gamma,\gamma_0 \leq \frac{1}{8}$. For every dispersive spacetime solution $\mathcal{M},g,\phi)$ of size $(C,\gamma_0,N)$ there exists $\epsilon = \epsilon(C,\gamma,\gamma_0,N) > 0$ such taht for all $(\epsilon,\gamma,N)$-admissible perturbations \footnote{For a definition of this term see \cite{Luk20} section 3} of $(\mathcal{M},g,\phi)$ the maximal globally hyperbolic future development is future causally geodesically complete and the spacetime remains close to $(\mathcal{M},g,\phi)$ \footnote{In a sense defined precisely in Section 4 of \cite{Luk20}}.
\end{theorem*}
This paper is devoted to finding a class of spherically symmetric solutions which satisfy the hypotheses of this theorem.

\subsection{Averaging Operators}

Some of the estimates near the axis will require some extra machinery to obtain.
In particular, it will sometimes be useful to consider the averages of certain quantities over regions near the axis in order to obtain better control over them.
To this end, following \cite{Luk2018} we introduce the operators:
\begin{equation}
\label{vavg}
I_v^s[f](u,v) = \frac{1}{r^{s}(u,v)}\int_u^v f(v')r^{s-1}\lambda(u,v')dv'
\end{equation}
\begin{equation}
\label{uavg}
I_v^s[f](u,v) = \frac{1}{r^s(u,v)}\int_v^u f(u')r^{s-1}\nu(u',v)du'
\end{equation}
\begin{equation}
\label{ravg}
I_{\tilde{r}}^s[f](u,v) = \frac{1}{\tilde{r}^{s}(u,v)}\int_0^R f(\tilde{r}')(\tilde{r}')^{s-1}d\tilde{r}'
\end{equation}
where as above $\tilde{r} := v-u$

Similar to the $u,v$ cases demonstrated in \cite{Luk2018} we have for this $r$-averaging operator a differential formula:
\begin{lemma}
\label{avglemma}
For $s \geq 1$,
\begin{dmath*} \partial_r I_r^s[f](r) = I_r^{s+1}[\partial_r f](r) \end{dmath*}
\end{lemma}

\begin{proof} The procedure is identical to that presented in \cite{Luk2018}, but we will repeat it here in the new case for completeness. So let $\rho = r^{s}$. Then $s\rho^{1 - \frac{1}{s}}dr = d\rho$ and $s\rho^{1-\frac{1}{s}}\partial_{\rho} = \partial_r$. Then the LHS is
\begin{dmath*} (\partial_{r}I_r^s[f])(r) = \rho^{1-\frac{1}{s}} \partial_\rho \Lx \frac{1}{\rho} \int_0^\rho f(\rho')d\rho' \Rx. \end{dmath*}
One then checks that, letting $\sigma' = \frac{\rho'}{\rho}$
\begin{align*}
\partial_{\rho}\Lx \frac{1}{\rho}\int_0^\rho f(\rho')d\rho' \Rx &= \partial_{\rho} \Lx \int_0^1 f(\rho\sigma')d\sigma' \Rx\\
&= \int_0^1 (\partial_{\rho}f)(\rho\sigma')\sigma' d\sigma'\\
&= \frac{1}{\rho^2}\int_0^{\rho}(\partial_{\rho}f)(\rho') \rho' d\rho'.
\end{align*}
Thus
\begin{dmath*} \rho^{1 - \frac{1}{s}} \partial_\rho \Lx \frac{1}{\rho} \int_0^{\rho} f(\rho')d\rho' \Rx = \frac{1}{\rho^{1 + \frac{1}{s}}} \int_0^{\rho} f(\rho')\rho'd\rho' = \frac{1}{r^{s+1}} \int_0^r f(r') (r')^{s} dr'. \end{dmath*}
So substituting we have
\begin{dmath*} (\partial_r I_s[f])(r) = I_r^{s+1}[\partial_rf](r) \end{dmath*}
as claimed.
\end{proof}

\section{Main Results}
\label{secresults}
In this section I give precise formulations of the main results of this paper. The first two sections will focus on establishing decay for $\phi$ and the geometry of solutions in the $1+1$ dimensional, spherically symmetric setting:
\begin{theorem}\label{main1}
Let $(\phi,r,m)$ be a locally scattering solution to \eqref{SSESF} with initial data asymptotically flat of order $\omega' \geq 2$ in $C^k$ ($k \geq 1$) towards $\mathcal{I}^+$ and $\mathcal{I}^-$. Then the following estimates hold in $\mathbb{I}$ for all multi-indices $\alpha,\beta$ with $\abs{\alpha} \leq k, \abs{\beta} \leq k+1$:
\begin{gather}
\label{main11}\abs{\partv{\abs{\alpha}}\lambda} \lesssim (1+v)^{-(\abs{\alpha}+1)}\\
\label{main12}\abs{\partu{\alpha_u}\partv{\alpha_v}\lambda} \lesssim \min\set*{(1+\abs{u})^{-(\alpha_u+1)}v^{-\alpha_v},(1+\abs{u})^{-\alpha_u}(1+v)^{-(\alpha_v+1)}}\\
\label{main13}\abs{\partu{\abs{\alpha}}\nu} \lesssim (1+\abs{u})^{-(\abs{\alpha}+1)}\\
\label{main14}\abs{\partu{\alpha_u}\partv{\alpha_v}\nu} \lesssim \min\set*{(1+\abs{u})^{-(\alpha_u+1)}(1+v)^{-\alpha_v}, (1+\abs{u})^{-\alpha_u}(1+v)^{-(\alpha_v + 1)}}\\
\label{main15}\abs{\partu{\alpha_u}\partv{\alpha_v}\phi} \lesssim \min\set*{(1+\abs{u})^{-(\alpha_u+1}(1+v)^{-\alpha_v},(1+\abs{u})^{-\alpha_u}(1+v)^{-(\alpha_v+1)}}\\
\label{main16}\abs{\partu{\alpha_u}\partv{\alpha_v}m} \lesssim \min\set*{(1+\abs{u})^{-(\alpha_u+1)}(1+v)^{-\alpha_v},(1+\abs{u})^{-\alpha_u}(1+v)^{-(\alpha_v+1)}}\\
\label{main17}\abs{\partu{\alpha_u}\partv{\alpha_v}\frac{m}{r^l}} \lesssim \min\set*{(1+\abs{u})^{-(\alpha_u+1+l)}(1+v)^{-\alpha_v},(1+\abs{u})^{-\alpha_u}(1+v)^{-(\alpha_v+1+l)}}\\
\label{main18}\abs{\partu{\beta_u}\partv{\beta_v}(r\phi)} \lesssim \min\set*{(1+\abs{u})^{-(\alpha_u+1}(1+v)^{-\alpha_v}, (1+\abs{u})^{-\alpha_u}(1+v)^{-(\alpha_v+1)}}\\
\label{main19}\abs{\partv{\abs{\beta}}(r\phi)} \lesssim (1+v)^{-(\abs{\beta}+1)}\\
\label{main110}\abs{\partu{\abs{\beta}}(r\phi)} \lesssim (1+\abs{u})^{-(\abs{\beta}+1)}
\end{gather}
for $l \leq 2$.
\end{theorem}
Here, and throughout the remainder of this paper we will use the symbol $\lesssim$ to indicate an inequality up to an unspecified constant depending only on the size of initial data (Bondi mass and $\abs{\phi}$), $\abs{\alpha}$ and $\abs{\beta}$.

\begin{remark}
In the proof of this theorem we will not explicitly check that data of the type above gives rise to a global solution (i.e. one defined in the entire right half plane) of the form we wish to consider. This is, however true, and we will assume as much. For discussion of this see \cite{Dafermos2003}.
\end{remark}
In fact we will prove slightly stronger decay in $\mathcal{Q}$ en-route to this theorem, however the above bounds suffice to prove our second result:
\begin{theorem}\label{main2}
Let $(\phi,r,m)$ be as above, and suppose $k \geq 11$. Let $0 < \gamma \leq \frac{1}{8}$. Then the corresponding $(3+1)$ dimensional solution to ESF, $(\mathcal{M}, \textbf{g}, \tilde\phi)$ is a dispersive spacetime solution of size $(C,\gamma_0,k)$ (for some $C > 0$ depending on the constants of \cref{main1}), and is thus globally non-linearly stable to good non-spherically symmetric perturbations by \cref{stabtheorem}.
\end{theorem}

\section{Inductive Estimates in $u,v$ in the Region $\mathcal{Q}$}
\label{Qproof}
In this section we prove that, in $\mathcal{Q}$ all $f\in \set*{\lambda,\nu,\mu,\phi}$ satisfy estimates of the form
\begin{dmath*} \abs{\partial^{\alpha}f}(u,v) \leq \min\set*{\frac{C}{u^{\abs{\alpha}}}, \frac{C}{u^{\alpha_u}r^{\alpha_v}}} \end{dmath*}
to arbitrary differential order. We will work in the gauge \ref{gauge1}. Moreover, in this section we will only consider data posed on $C_1$, rather than the full setting eventually required. This is all that's required to address the region $\mathcal{Q}$, and is consistent by our compatibility assumption on the initial data.

We will, throughout this section, make use of several of the estimates in $\mathcal{Q}$ proved in \cite{Luk2015}
\begin{theorem*}
Let $(\phi,r,m)$ be a locally scattering solution to \eqref{SSESF} on $\mathcal{Q}$ with initial data asymptotically flat of order $\omega'$ in $C^k$ towards $\mathcal{I}^+$ for $k \geq 1$. Then the following bounds hold:
\begin{gather}
\label{unifapriori} \frac{1}{3} \leq \lambda \leq \frac{1}{2} \quad \frac{1}{3} \leq -\nu \leq \frac{2}{3} \quad \frac{2}{3} \leq (1-\mu) \leq 1\\
\label{phiapriori} \abs{\phi} \lesssim \min\set*{r^{-1}u^{-1},u^{-2}}\\
\label{dvlambdaapriori} \abs{\partial_v\lambda} \lesssim \min\set*{r^{-3},u^{-3}}\\
\label{dunuapriori} \abs{\partial_u \nu} \lesssim u^{-3}\\
\label{dvrphiapriori} \abs{\partial_v(r\phi)} \lesssim \min\set*{v^{-2},u^{-2}} \quad \abs{\partial_v^2(r\phi)} \lesssim \min\set*{v^{-3},u^{-3}}\\
\label{durphiapriori} \abs{\partial_u(r\phi)} \lesssim u^{-2} \quad \abs{\partial_u^2(r\phi)} \lesssim u^{-3}\\
\label{dvphiapriori} \abs{\partial_v\phi} \lesssim \min\set*{r^{-2}u^{-1},u^{-3}}\\
\label{duphiapriori} \abs{\partial_u\phi} \lesssim \min\set*{r^{-1}u^{-2},u^{-3}}\\
\label{mapriori} \abs{m} \lesssim \min\set*{u^{-3}, r^3u^{-6}}.
\end{gather}
\end{theorem*}

\subsection{Inductive Framework}
The goal in this section is to obtain bounds to arbitrary differential order in $u,v$ via an inductive process. Our goal is the following inductive step:
\begin{theorem}
\label{uvind}
Let $(\phi,r,m)$ be a locally scattering, asymptotically flat of order $\omega' \geq 2$ in $C^q$ towards $\mathcal{I}^+$ solution to \eqref{SSESF} in $\mathcal{Q}$, with data prescribed on some $C_u$. Suppose that for all $\alpha = (\alpha_u,\alpha_v),\beta \in \Z^2$ multindices with $1 \leq |\alpha| \leq n < q$, $|\beta| \leq n+1$ (excluding $\alpha,\beta = (0,0)$) we have the estimates
\begin{gather}
\label{ind1hypvlambda} \abs{\partv{\abs{\alpha}}\lambda}\lesssim \min\set*{r^{-(\abs{\alpha}+2)},u^{-(\abs{\alpha}+2)}}\\
\label{ind1hypuvlambda} \abs{\partu{\alpha_u}\partv{\alpha_v}\lambda} \lesssim \min\set*{r^{-(\alpha_v+2)}u^{-(\alpha_u+1)},u^{-(\abs{\alpha}+4)}}\\
\label{ind1hypunu} \abs{\partu{\abs{\alpha}}\nu} \lesssim u^{-(\abs{\alpha}+2)}\\
\label{ind1hypuvnu} \abs{\partu{\alpha_u}\partv{\alpha_v}\nu} \lesssim \min\set*{r^{-(\alpha_v+1)}u^{-(\alpha_u+2)},u^{-(\abs{\alpha}+5)}}\\
\label{ind1hypphi} \abs{\partu{\alpha_u}\partv{\alpha_v}\phi} \lesssim \min\set*{r^{-(\alpha_v + 1)}u^{-(\alpha_u+1)},u^{-(\abs{\alpha}+2)}}\\
\label{ind1hypvrphi} \abs{\partv{\abs{\beta}}(r\phi)} \lesssim \min\set*{r^{-(\abs{\beta}+1)},u^{-(\abs{\beta}+1)}}\\
\label{ind1hypurphi} \abs{\partu{\abs{\beta}}(r\phi)} \lesssim \min\set*{u^{-(\abs{\beta}+1)}}\\
\label{ind1hypuvrphi} \abs{\partu{\beta_u}\partv{\beta_v}(r\phi)} \lesssim \min\set*{r^{-(\beta_v+1)}u^{-(\beta_u+2)},u^{-(\abs{\beta}+4)}}\\
\label{ind1hypm} \abs{\partu{\alpha_u}\partv{\alpha_v}m} \lesssim \min\set*{r^{-(\alpha_v+k)}u^{-(\alpha_u+2)},u^{-(\abs{\alpha}+3)}}\\
\label{ind1hypmweight} \abs{\partu{\alpha_u}\partv{\alpha_v}\frac{m}{r^k}} \lesssim \min\set*{r^{-(\alpha_v+k}u^{-(\alpha_u+2)},u^{-(\abs{\alpha}+k+3)}}
\end{gather}
for $k \leq 2$, where $\vartheta(l) = \begin{cases}1 & n \geq 1\\0 & \text{otherwise}\end{cases}$, in the region $\mathcal{Q}$, and we take $\alpha_u,\beta_u > 0$ in \eqref{ind1hypuvlambda}, and \eqref{ind1hypuvrphi},
and $\alpha_v, \beta_v > 0$ in \eqref{ind1hypuvnu} and \eqref{ind1hypuvrphi} respectively. Then in fact these bounds hold for $\abs{\alpha} = n+1, \abs{\beta} = n+2$.
\end{theorem}

The remainder of this section will be devoted to the proof of this result,
and thus the closing of our first inductive bounds. 

Before proceeding any further, it will be useful at times to exchange the order of (gauge invariant) derivatives without worry about changing the form of the resulting bounds. Moreover, we would like to know that the above bounds apply to more than just the specific orderings of derivatives written. To this end we prove the following lemma:
\begin{lemma}\label{uvexclemma}
Suppose the bounds of \cref{uvind} hold for $\abs{\alpha} \leq n, \abs{\beta} \leq n+1$. Then if any of the bounds \eqref{ind1hypuvlambda}, \eqref{ind1hypuvnu}, \eqref{ind1hypphi}, or \eqref{ind1hypm} hold for some ordering of derivatives $\partv{},\partu{}$ at differential order $n+1$, then in fact the same estimates hold for arbitrary reorderings of $\partu{},\partv{}$.

Moreover, if \cref{ind1hypuvlambda,ind1hypuvnu} hold at order $n+1$, then if \eqref{ind1hypuvrphi} holds at order $n+2$ for any ordering of $\partv{},\partu{}$, it holds for all orderings.
\end{lemma}

\begin{proof}
It suffices to check that the difference between two adjacent (i.e. differing by a single exchange of $u$ and $v$ derivatives) orderings always satisfies at least the same bound as the initial ordering.

We denote an ordering as a $k$-tuple of integers $\vec{l}$ where the first entry represents the number of $u$ derivatives acting at the far left of our differential expression, and subsequent entries give the number of derivative of the type different from that preceding it that occur before the next change. We will write an ordered gauge normalized differential with multi-index $\gamma = (\gamma_u,\gamma_v)$ as $\bar{\partial}^{\gamma,\vec{l}}$. For example, $\bar{\partial}^{(2,1),(1,1,1)}f = \partu{}\partv{}\partu{}f$.

We begin with the first part of our result. We also assume that \cref{uvexclemma} holds at all orders $\leq n$.

Then let $\abs{\gamma} = n+1$, $\vec{l},\vec{l}'$ be two orderings corresponding to $\gamma$ that differ by one exchange. Then for a function $f = f(u,v)$ we can write
\[ \abs{\Lx\bar{\partial}^{\gamma,\vec{l}} - \bar{\partial}^{\gamma,\vec{l}'}\Rx f} = \abs{\bar{\partial}^{\gamma_1,\vec{l}_1}\Lx \partv{}\nu\partu{} - \partu{}\lambda\partv{} \Rx \bar{\partial}^{\gamma_2,\vec{l}_2}f} \]
for $\gamma_1,\gamma_2,\vec{l}_1,\vec{l}_2$ splitting $\gamma,\vec{l}$ around the two derivatives which are exchanged. Thus we see that if
\[ \abs{\bar{\partial}^{\gamma_3,\vec{s}}\partv{}\nu \bar{\partial}^{\gamma_3',\vec{s}'}\partu{}\bar{\partial}^{\gamma_2,\vec{l}'}f - \bar{\partial}^{\gamma_3,\vec{s}}\partu{}\lambda \bar{\partial}^{\gamma_3',\vec{s}'}\partv{}\bar{\partial}^{\gamma_2,\vec{l}'}f} \lesssim \abs{\bar{\partial}^{\gamma,\vec{l}}f} \]
for all $\gamma_3 + \gamma_3' = \gamma_1$ and appropriate sub-orderings given by the Leibniz rule, then the $\vec{l}'$ ordering must satisfy the same estimate as the $\vec{l}$ ordering. So in our first case, since all our bounds increase by no more than one power of $r$ or $u$ as a derivative is added, we need only check that
\[ \abs{\bar{\partial}^{\gamma_3,\vec{s}}\partv{}\nu} \lesssim \min\set*{r^{-(\gamma_{3,v+1})}u^{-\gamma_{3,u}},u^{-(\abs{\gamma_3}+1)}} \]
and the analagous statement for $\lambda$. But each of these hold immediately since $\abs{\gamma}+1 \leq n$ by construction, so by our hypotheses, these terms verify \eqref{ind1hypuvnu}, \eqref{ind1hypuvlambda} respectively. (There is a single exception to this, which is the case of a derivative of $\lambda$ or $\nu$ with only a single $u$ or $v$ derivative respectively. In this case the change in hypothesized order of decay from the mixed derivative to the non-mixed is slightly more, however one checks that in fact the change is still covered, since it is given by exactly the extra term we obtain from the exchange). Thus our first case is complete.

The second case is completely identical, requiring this extra order of derivative only because we work at one order higher (again there is some subtlety making the change from a mixed derivative of $r\phi$ to a non-mixed derivative. Again one easily checks that this in fact a non-issue).

It remains to check that our assumption that \cref{uvexclemma} holds at all lower orders is justified. By induction it suffices to check this for the case $n = 0$. The only non-trivial case here is
\begin{align*} \abs{\partu{}\partv{}(r\phi) - \partv{}\partu{}(r\phi)} &= {\partv{}\nu \partu{}(r\phi) - \partu{}\lambda \partv{}(r\phi)}\\
	&\lesssim \abs{\partv{}\nu \partu{}(r\phi)}\\
	&\lesssim \min\set*{r^{-2}u^{-5},u^{-7}}
\end{align*}
which is better decay than that hypothesized for $\partu{}\partv{}(r\phi)$, so the exchange holds (note that this also checks the problematic portion of the second case above).

Thus the proof is complete.
\end{proof}

\subsection{Decay for Mixed Derivatives of $r$}

We begin with mixed derivativs of $r$, i.e. $\partial^\alpha r$ such that neither $\alpha_u$ nor $\alpha_v$ is 0.
\begin{lemma}
\label{ind1mixr}
Under the hypothesis of \cref{uvind}, let $\abs{\alpha},\abs{\beta} = n+1$ and $\alpha_u,\beta_v > 0$. Then the following hold:
\begin{dmath}\label{ind1uvlambda} \abs{\partv{\alpha_v}\partu{\alpha_u}\lambda} \lesssim \min\set*{r^{-(\alpha_v+2)}u^{-(\alpha_u+2)},u^{-(\abs{\alpha}+4)}} \end{dmath}
\begin{dmath}\label{ind1uvnu} \abs{\partu{\beta_u} \partv{\beta_v}\nu} \lesssim \min\set*{r^{-(\alpha_v+1)}u^{-(\alpha_u+3)},u^{-(\abs{\alpha}+5)}}. \end{dmath}
\end{lemma}

\begin{proof} In each case this is essentially completely computational. We start with \eqref{ind1uvlambda}. By our hypothesis and \eqref{SSESF}, we can rewrite the LHS as
\begin{dmath*} \partv{\alpha_v} \partu{\alpha_u}\lambda = \partv{\alpha_v} \partu{\alpha_u-1}\Lx \frac{m\lambda}{r^2(1-\mu)} \Rx. \end{dmath*}
Expanding this expression via the Leibniz rule we have a general term:
\begin{dmath*} \frac{1}{(1-\mu)^{\abs{\alpha^3}+1}}\partv{\alpha^1_v}\partu{\alpha^1_u}\Lx \frac{m}{r^2} \Rx \partv{\alpha^2_v}\partv{\alpha^2_u}(\lambda)\partv{\alpha^3_v}\partu{\alpha^3_u}(\mu) \end{dmath*}
where $\alpha^1+\alpha^2+\alpha^3 = (\alpha_u-1,\alpha_v)$. We may ignore the leading term as $\mu < \frac{1}{2}$ by \eqref{unifapriori}, so we only consider the differential terms.
There are, in principle, many combinations to check, however observe that by the form of our hypothesized
estimates (in particular the uniform gains on certain terms) we only need to address the following cases:
\begin{enumerate}
	\item[(1)] $\alpha^1 = \alpha^2 = 0$.
	\item[(2)] $\alpha^1 = \alpha^3 = 0$.
	\item[(3)] $\alpha^2 = \alpha^3 = 0$.
	\item[(4)] $\alpha^2 = 0$, $\alpha^1,\alpha^4 \neq 0$.
\end{enumerate}
Since all other terms are strictly better by our inductive hypothesis, since they gain more than one power of decay per derivative.
In the first case we have
\begin{align*}
	\abs{\frac{m}{r^2}\lambda\nu \partv{\alpha_v}\partu{\alpha_u-1}\mu} &\lesssim \frac{m}{r^2}\min\set*{r^{-(\alpha_v+1}u^{-(\alpha_u+1)},u^{-(\abs{\alpha}+3)}}\\
	&\lesssim \min\set*{r^{-(\alpha_v+3)}u^{-(\alpha_u+4)},u^{-(\abs{\alpha}+8)}}
\end{align*}
by \eqref{mapriori}, \eqref{ind1hypm}. The second case is bounded (up to a constant) by
\begin{align*}
	\abs{\frac{m}{r^2}\partv{\alpha_v}\partu{\alpha_u-1}\lambda} &\lesssim \frac{m}{r^2}\min\set*{r^{-(\alpha_v+2)}u^{-(\alpha_u-2)},u^{-(|\alpha|+1)}}\\
	&\lesssim \min\set*{r^{-(\alpha_v+4)}u^{-(\alpha_u+1)},u^{-(\abs{\alpha}+6)}}.
\end{align*}
The third case has
\begin{dmath*} \abs{\partv{\alpha_v}\partu{\alpha_u-1}\Lx \frac{m}{r^2} \Rx} \lesssim \min\set*{r^{-(\alpha_v+2)}u^{-(\alpha_u+1)},u^{-(\abs{\alpha}+4)}}. \end{dmath*}
Finally in the fourth case we have
\begin{equation*}\begin{split} &\abs{\partv{\alpha^1_v}\partu{\alpha^1_u}\Lx\frac{m}{r^2}\Rx \partv{\alpha^3_v}\partu{\alpha^3_u}\mu}\\&\lesssim \min\set*{r^{-(\alpha^1_v+2)}u^{-(\alpha^1_u+2)},u^{-(\abs{\alpha^1}+5)}}\cdot\min\set*{r^{-(\alpha^3_v+1)}u^{-(\alpha^4_u+2)},u^{-(\abs{\alpha^3}+4)}}\\ &\leq \min\set*{r^{-(\alpha_v+3)}u^{-(\alpha_u+3)},u^{-(\abs{\alpha}+8)}} \end{split}.\end{equation*}
Thus all our terms satisfy the desired decay.

The $\nu$ case then amounts to the same bounds, except we lose a power of $v$ via \eqref{SSESF} instead of $u$. Thus we will not repeat the details of this proof.
\end{proof}

This proposition leaves out exactly two cases, $\partial^{n+1}_u\nu,\partial^{n+1}_v\lambda$, which require slightly more careful treatment. We will return to these after closing some of the other bounds at order $n+1$.

\subsection{Decay for Mixed Derivatives of $r\phi$}
We are now ready to work on derivatives of $r\phi$ at order $n+2$. The goal is as above:
\begin{lemma}
\label{ind1mixrphi}
Under the hypotheses of \cref{uvind} for $\beta$
a multi-index with $\abs{\beta} = n+2$ and $\beta_u,\beta_v \neq 0$, we have
\begin{dmath*} \abs{\partv{\beta_v-1} \partu{\beta_u}\partv{}(r\phi)} \lesssim \min\set*{r^{-(\beta_v+1)}u^{-(\beta_u+2)},u^{-(\abs{\beta}+4)}}. \end{dmath*}
\end{lemma}

\begin{proof} The procedure here is roughly the same as that for mixed derivatives of $\lambda$ and $\nu$ above.
We have by the wave equation for $\phi$ that
\begin{dmath*} \partu{}\partv{}(r\phi) = \frac{2m}{(1-\mu)r^2}\Lx \phi - \partv{}(r\phi) \Rx. \end{dmath*}
Thus we can rewrite our first expression as
\begin{equation*}\begin{split} &\partv{\beta_v-1}\partu{\beta_u}\partv{}(r\phi)\\&= \partv{\beta_v-1}\partu{\beta_u-1}\Lx \frac{m}{(1-\mu)r^2}\Lx \phi - \partv{}(r\phi) \Rx \Rx. \end{split}\end{equation*}
Then we can expand this term by term via the Leibniz rule with the general term being of the form (omitting terms of order 1)
\begin{dmath*} \partv{\gamma^1_v}\partu{\gamma^1_u}\Lx \frac{m}{r^2} \Rx \partv{\gamma^2_v}\partu{\gamma^2_u}\mu \partv{\gamma^3_v}\partu{\gamma^3_u}\Lx \phi - \partv{}(r\phi)\Rx \end{dmath*}
with $\sum_{i}\gamma_i = \alpha'$.
As before, the terms with lowest order of decay occur when exactly one term is acted on by all derivatives. Thus we have three cases to check:
\begin{enumerate}
	\item[(1)] $\gamma^1 = \gamma^2 = 0$.
	\item[(2)] $\gamma^1 = \gamma^3 = 0$.
	\item[(3)] $\gamma^2 = \gamma^3 = 0$.
\end{enumerate}
All other mixtures only improve by a constant order in $r$ or $u$, so we need not consider these.

In the first case we have
\begin{equation*}\begin{split}
&\abs{\frac{m}{r^2} \partv{\beta_v-1}\partu{\beta_u-1}\Lx \phi - \partv{}(r\phi) \Rx}\\&\lesssim \abs{\frac{m}{r^2}}\Lx \min\set*{r^{-(\beta_v)}u^{-(\beta_u)},u^{-(\abs{\beta})}} - \min\set*{r^{-(\beta_v-1)}u^{-(\beta_u-1)},u^{-(\abs{\beta}-1)}} \Rx\\
&\lesssim \min\set*{r^{-(\beta_v+1)}u^{-(\beta_u+2)},u^{-(\abs{\beta}+4)}}
\end{split}\end{equation*}
by \eqref{mapriori}, \eqref{ind1hypphi}, \eqref{ind1hypuvrphi}, \eqref{ind1hypvrphi}.
In the second case we have
\begin{equation*}\begin{split}
&\abs{\frac{m}{r^2}\Lx \phi - \partv{}(r\phi)\Rx \partv{\beta_v-1}\partu{\beta_u-1}\mu}\\&\lesssim \abs{\frac{m}{r^2}\Lx \phi - \partv{}(r\phi)\Rx} \min\set*{r^{-(\beta_v)}u^{-(\beta_u+1)},u^{-(\abs{\beta}+2)}}\\
&\lesssim \min\set*{r^{-(\beta_v+4)}u^{-(\beta_u+4)}, u^{-(\abs{\beta} + 8)}}
\end{split}\end{equation*}
by \eqref{mapriori}, \eqref{dvrphiapriori}, \eqref{phiapriori}.
The third case similarly gives us
\begin{align*}
&\abs{\Lx\phi - \partv{}(r\phi)\Rx \partv{\beta_v-1}\partu{\beta_u-1}\Lx\frac{m}{r^2}\Rx}\\&\lesssim \abs{\Lx\phi - \partv{}(r\phi)\Rx} \min\set*{r^{-(\beta_v+1)}u^{-(\beta_u+1)},u^{-(\abs{\beta}+3)}}\\
&\lesssim \min\set*{r^{-(\beta_v+1)}u^{-(\beta_u+3)}, u^{-(\abs{\beta} + 5)}}
\end{align*}
by \eqref{phiapriori}, \eqref{dvrphiapriori}, \eqref{ind1hypm}.
So in each case the the desired bound holds, and we conclude that the first relation holds at order $n+1$.
\end{proof}

As above, it remains to deal with the cases $\partial^{n+1}_v(r\phi),\partial^{n+1}_u(r\phi)$.

\subsection{Preliminary Estimates for $\phi$, $m$ and $\frac{m}{r^k}$}
In this section we will use our upgraded bounds for derivatives of $r$ above to obtain some initial estimates for $\phi$, $m$ weightings at order $n+1$. We will begin with an easy bound for $\phi$.
\begin{lemma}
\label{phiprelim}
Suppose our inductive hypotheses hold at order $n$. Then for $\abs{\alpha} = n+1$ we have the bounds
\begin{dmath*} \abs{\partu{\alpha_u}\partv{\alpha_v}\phi} \lesssim \min\set*{r^{-(\alpha_v+1)}u^{-(\alpha_u+1)}, r^{-1}u^{-(\abs{\alpha}+1)}}. \end{dmath*}
\end{lemma}
\begin{proof}
This follows immediately by expanding $\abs{\partu{\alpha_u}\partv{\alpha_v}(r\phi)}$ to obtain
\begin{equation*}\begin{split} &\abs{r (\partu{\alpha_u}\partv{\alpha_v}\phi) + \alpha_u\alpha_v\Lx \partu{\alpha_u-1}\partv{\alpha_v} + \partu{\alpha_u}\partv{\alpha_v-1}\Rx \phi}\\& \lesssim \min\set*{r^{-(\alpha_v+1)\vartheta(\alpha_v)}u^{-\alpha_u},u^{-(\abs{\alpha}+1)}} \end{split}\end{equation*}
by \eqref{ind1hypvrphi},\eqref{ind1hypurphi}, \eqref{ind1hypuvrphi}.
So then via the triangle inequality we can write
\begin{equation*}\begin{split} &\abs{r \partu{\alpha_u}\partv{\alpha_v}\phi}\\
&\leq C\min\set*{r^{-(\alpha_v+1)\vartheta(\alpha_v)}u^{-\alpha_u},u^{-(\abs{\alpha}+1)}}\\
&- \abs{\alpha_u\alpha_v\Lx \partu{\alpha_u-1}\partv{\alpha_v} + \partu{\alpha_u}\partv{\alpha_v-1}\Rx \phi} \end{split}\end{equation*}
and thus
\begin{dmath*} \abs{\partial^{\alpha}\phi} \lesssim \min\set*{r^{-(\alpha_v+1)}u^{-(\alpha_u+1)}, r^{-1}u^{-\abs{\alpha}+1}} \end{dmath*}
dividing through by the $r$ and taking our most weakly decaying terms (those which lose a $v$ derivative acting on $\phi$) on the RHS as our minimal order of decay to obtain our desired bound.
\end{proof}

Observe that this bound then allows us to immediately obtain optimal next order control over $m$ as well:
\begin{corollary}
\label{mprelim}
With the estimate \cref{phiprelim} we have, for $\abs{\alpha} = n+1, \alpha_v > 0$, we have
\begin{dmath*} \abs{\partu{\alpha_u}\partv{\alpha_v}m} \lesssim \min\set*{r^{-(\alpha_v+1)}u^{-(\alpha_u+2)},u^{-(\abs{\alpha}+3)}} \end{dmath*}
\begin{dmath*} \abs{\partu{n+1}m} \lesssim u^{-(n+3)}. \end{dmath*}
\end{corollary}
\begin{proof}
We begin with our first relation. Recall from \eqref{SSESF} that we have
\begin{equation*}
\lambda^{-1}\partial_vm = \frac{1}{2}(1-\mu)r^2(\lambda^{-1}\partial_v\phi)^2.
\end{equation*}
Then the cases we need to consider here are:
\begin{enumerate}
\item[i.] All derivatives act on $r^2$ and $(\partial_v \phi)^2$.
\item[ii.] All derivatives act on $(1-\mu)$ and $r^2$.
\end{enumerate}
In the first case we have a general term:
\begin{dmath*} \abs{\frac{1}{2}(1-\mu)r^{2-\abs{\beta}}\partu{\alpha_u-\beta_u}\partv{\alpha_v - \beta_v - 1}(\lambda^{-1}\partial_v\phi)^2)} \lesssim \min\set*{r^{-(\alpha_v+1)}u^{-(\alpha_u + 2)},u^{-(\abs{\alpha} + 3)}} \end{dmath*}
for $\abs{\beta} \leq 2$, where we have used the inductive hypothesis and the result of \cref{phiprelim}, as well as the fact that $r^{-k}u^{-l} \lesssim \min\set*{r^{-(k+s)}u^{-(l-s)},u^{-(k+l)}}$ to control each case, as we gain the same overall power of decay by taking a derivative of $r$, as we do by differentiating a copy of $\partv{}\phi$. This gives us the desired order of decay.

In the second case we have the term term
\begin{dmath*} \abs{\frac{1}{2} r^{2-\abs{\beta}}\partu{\alpha_u-\beta_u}\partv{\alpha_v-\beta_v-1}(1-\mu)(\lambda^{-1}\partial_v\phi)^2} \lesssim \min\set*{r^{-(\alpha_v+2)}u^{-(\alpha_u+4)},u^{-(\abs{\alpha}+7)}} \end{dmath*}
using \eqref{dvphiapriori}, \eqref{ind1hypm} and the same strategy as above. Thus we have our bound in the case $\alpha_v > 0$.

It remains to address the case where $\alpha_v = 0$. Now we must use our other equation for $m$:
\[ \partu{}m = \frac{1}{2}(1-\mu) r^2(\partu{}\phi)^2. \]
As above we have two distinct cases. The overall analysis is the same, differing only in the fact that we now lose our extra power of $r$ decay from the additional $\partv{}\phi$ term. As such we will arrive at a bound of the same overall form, with one less power of $r$ decay than we would expect if we had a $v$ derivative, and so arrive at a bound:
\[ \abs{\partu{n+1}m} \lesssim u^{-(n+3)} \]
by the same computation as above.
\end{proof}

With this control of $m$ established, we now must concern ourselves with the weighted versions $\mu,\frac{m}{r^2}$. The following lemma will take care of this:
\begin{lemma}
\label{mweight}
Given the result of \cref{mprelim}, we have the bounds
\begin{dmath*} \abs{\partu{\alpha_u}\partv{\alpha_v}\Lx\frac{m}{r^{k}}\Rx} \lesssim \min\set*{r^{-(\alpha_v + k)}u^{-(\alpha_u+2)},r^{-k}u^{-(\abs{\alpha}+3)}} \end{dmath*}
for $k = 1,2$, $\abs{\alpha} = n+1$.
\end{lemma}
\begin{proof}
We will employ a similar technique as was used to gain our initial bound on $\partial^{\alpha}\phi$ above. In particular, observe that we have
\begin{dmath*} \partu{\alpha_u}\partv{\alpha_v}m = \partu{\alpha_u}\partv{\alpha_v}\Lx \frac{m}{r^k}r^k \Rx = \sum_{\beta + \gamma = \alpha} \binom{\alpha_v}{\beta_v}\binom{\alpha_u}{\beta_u}\partu{\beta_u}\partv{\beta_v}\Lx \frac{m}{r^k} \Rx \partu{\gamma_u}\partv{\gamma_v}(r^k). \end{dmath*}
So rearranging we find
\begin{dmath*} \abs{\partu{\alpha_u}\partv{\alpha_v}\frac{m}{r^k}} \leq \frac{1}{r^k} \Lx \abs{\partu{\alpha_u}\partv{\alpha_v}m} + \sum_{\substack{\beta + \gamma = \alpha\\\gamma \neq 0}}\abs{\binom{\alpha_v}{\beta_v}\binom{\alpha_u}{\beta_u}\partu{\beta_u}\partv{\beta_v}\Lx \frac{m}{r^k} \Rx \partu{\gamma_u}\partv{\gamma_v}(r^k)}\Rx. \end{dmath*}

This gives rise to the bound
\begin{equation*}\begin{split} &\abs{\partu{\alpha_u}\partv{\alpha_v}\frac{m}{r^k}}\\&\lesssim \frac{1}{r^k}\Lx\min\set*{r^{-(\alpha_v+\vartheta{\alpha_v})}u^{-(\alpha_u+2)},u^{-(\abs{\alpha}+3)}} + \min\set*{r^{-(\alpha_v)}u^{-(\alpha_u+2)},u^{-(\abs{\alpha}+3)}} \Rx, \end{split}\end{equation*}
where we bound the general term in the sum using that $r^{-s}u^{-t} \lesssim \min\set*{r^{-(s+l)}u^{-(t-l)},u^{-(s+t)}}$. Thus the overall order of decay this obtains is
\begin{dmath*} \abs{\partial^{\alpha}\frac{m}{r^k}} \lesssim \min\set*{r^{-(\alpha_v + k)}u^{-(\alpha_u+2)},r^{-k}u^{-(\abs{\alpha}+3)}}. \end{dmath*}
\end{proof}

\begin{remark}
The $r$ decay obtained in this manner for $\mu,\frac{m}{r^2}$ is already as strong as we hope for.
\end{remark}

\subsection{Full Decay for Derivatives of $r\phi$}

With this initial control, we find ourselves in a difficult position. The remaining terms to bound all suffer from a heavy $r^{-1}$ weighting, which will prevent us from closing any further bounds through the standard manner of computation we have pursued thus far (and also excludes the use of averaging operators).

In order to proceed, it will be necessary to obtain some next order control near the axis as well. To do this we will employ a bootstrapping approach centered around next order control of $r\phi$. Before beginning this however, we will need a few preliminary estimates in order to check that the botstrap closes:
\begin{proposition}
\label{vbootprelim}
Suppose, under the hypotheses of \cref{phiprelim}, with associated constants at order $\alpha = (\alpha_u,\alpha_v)$, $C_{\alpha_u,\alpha_v}$, we, in addition, have the bound
\begin{equation}\label{rphiassum} \abs{\partv{n+2}(r\phi)} \leq Cv^{-(n+3)} \end{equation}
For $1 \leq u < U$.

Then the following hold on $1 \leq u < U$, for some $C'$ depending only on $n$ (in particular we require only one constant on all of $\mathcal{Q}$)
\begin{equation}\label{phivboot} \abs{\partv{n+1}\phi} \leq CC'\min\set*{r^{-(n+2)}u^{-1},u^{-(n+3)}} \end{equation}
\begin{equation}\label{mvboot} \abs{\partv{n+1}\frac{m}{r^2}} \leq CC'\min\set*{r^{-(n+3)}u^{-2},u^{-(n+6)}} \end{equation}
\begin{equation}\label{muvboot} \abs{\partv{n+1}\mu} \leq CC' \min\set*{r^{-(n+2)}u^{-2},u^{-(n+5)}}. \end{equation}
\end{proposition}

\begin{proof}
To obtain \eqref{phivboot} we simply employ our averaging operator:
\begin{dmath*} \partial_v^{n+1}\phi(u,v) = \frac{1}{r^{n+2}(u,v)}\int_u^v \partial_v^{n+2}(r\phi)(u,v')r^{n+1}(u,v')dv' \end{dmath*}
then applying our hypothesis \eqref{rphiassum} the RHS is bounded in absolute value by
\begin{dmath*} \frac{C}{r^{n+2}(u,v)}\int_u^v (v')^{-(n+3)}r^{n+1}(u,v')dv' \leq \frac{3}{2}C\min\set*{r^{-(n+2)}u^{-1},u^{-(n+3)}}, \end{dmath*}
where the first bound is obtained by dividing through by $v^{-(n+2)}$ and pulling out the remaining copy of $v^{-1}$, to yield a bounded integral with an additional $u$ weight, and the second bound is obtained by using the simple supremum estimate to remove the $v$ term, and then integrating directly. The additional factor of $\frac{3}{2}$ results from evaluating the integral in $r$ thanks to our bound on $\abs{\lambda}$, \eqref{unifapriori}. Since the derivatives acting on extra factors of $\lambda$ are all well controlled, we ignore these terms without restriction.

Next, turning our attention to \eqref{mvboot} we observe:
\begin{dmath*} \partv{n+1}\frac{m}{r^2} = \partv{n}\Lx \frac{\partial_v m}{\lambda r^2} \Rx - \partv{n} \frac{2m}{r^{3}}. \end{dmath*}
Note that by our hypotheses we can then simply ignore the terms with derivatives acting on the $\lambda$'s, as these cannote be worse than those in which only $\frac{m}{r^3}$ is differentiated. We do not have enough a-priori control of either of these terms. Observe that, employing averaging operators, we have
\begin{dmath*} \partial_v^n\frac{m}{r^3}(u,v) = \frac{1}{r^{n+3}}\int_u^v \partial_v^n\Lx \frac{1-\mu}{2\lambda}(\partial_v \phi)^2 \Rx r^{n+2}dv'. \end{dmath*}
In particular we need only to bound the term $\partial_v^n \frac{\partial_v m}{r^2}$ in order to control the derivative of $\frac{m}{r^3}$. But this is exactly our first term, thus estimating this is sufficient. By our equation \eqref{SSESF}, the first term is
\begin{dmath*} \partv{n}(\frac{1}{2}(1-\mu)(\partv{}\phi)^2).  \end{dmath*}
Observing that, by our hypotheses and the bound found above, all these terms gain either a power of $r$ or of $u$ decay for each derivative applied, we immediately have:
\begin{dmath*} \abs{\partv{n}\Lx\frac{1}{2}(1-\mu)\partv{}\phi \partv{}\phi\Rx} \leq C_1C\min\set*{r^{-(n+4)}u^{-2},u^{-(n+6)}} \end{dmath*}
by our hypotheses and the above estimates (Thus $C_1$ depends only on constants controlling lower order terms, and $C$ for the highest derivatives of $\phi$). This controls the first term we are interested in directly, and we use this to obtain
\begin{dmath*} \abs{\partial_v^n\frac{m}{r^3}} \leq \frac{CC_1}{r^{n+3}}\int_u^v \min\set*{r^{-(n+4)}u^{-2},u^{-(n+6)}}r^{n+2}dv' \lesssim CC_2\min\set*{r^{-(n+3)},u^{-(n+6)}}. \end{dmath*}
Splitting the integral into regions $r^{n+4}(u,v') < u^{n+6},r^{n+4}(u,v') > u^{n+6}$ and applying the appropriate bounds to obtain the first bound, and simply using the uniform $u$ bound to obtain the second, modifying the constant to $C_2$ to account for the various constants (all independent of $C$ by our above bounds) that enter in this integration. The desired estimate \eqref{mvboot} follows.

Finally we will check \eqref{muvboot}. This is quite similar to \eqref{mvboot} above, but now our term is:
\begin{dmath*} \partv{n+1}\mu = \partv{n}\Lx \frac{\partial_vm}{\lambda r} - \partv{n}\frac{m}{r^2} \Rx. \end{dmath*}
The second term is already optimally controlled by our hypotheses, so we need only address the first. This is given by
\begin{dmath*} \partv{n}\Lx\frac{1}{2}(1-\mu)r(\partv{}\phi)^2\Rx. \end{dmath*}
Since we lose at most one derivative to removing this extra power of $r$ this is bounded by
\begin{dmath*} C\min\set*{r^{-(n+2)}u^{-2},u^{-(n+5)}}. \end{dmath*}
Combining this with our bound for the second term we obtain the desired bound (with an extra leading factor to observe the constant independent of $C$).

\end{proof}

In order to do our bootstrapping, we must guarantee some smallness of the term dependent on the highest order constant $C$ above. To this end, we prove the following proposition:
\begin{proposition}
Let $\epsilon > 0$. There exists some $v_0 > 1$ such that for all $v > v_0$,
\begin{equation}\label{small1} \int_1^v \abs{\frac{m}{r^2}(u',v)}du' < \epsilon v^{-1}, \end{equation}
\begin{equation}\label{small2} \int_1^v \abs{\phi(u',v)}du' < \epsilon, \end{equation}
and
\begin{equation}\label{small3} \int_1^v \abs{\partial_v(r\phi)}(u',v)du' < \epsilon. \end{equation}
\end{proposition}

\begin{proof}
Since we need only find some $v_0 > 1$, and $r$ is monotone in $v$ and unbounded, we may assume without restriction that $r(1,v) > \frac{v}{6} > 1$. We begin with \eqref{small1}. We have by \cref{mapriori} that
\begin{dmath*} \frac{m}{r^2} \lesssim \min\set*{ru^{-6},r^{-2}u^{-3}}. \end{dmath*}
By monotonicity of $r$ in $u$ we can thus bound \eqref{small1} by
\begin{dmath*} \tilde{C}\Lx\int_{0\leq r(u',v) \leq 1}(u')^{-6}du' + \int_{1 \leq r(u',v) \leq \frac{v}{6}}(u')^{-3}du' + \int_{\frac{v}{6} \leq r(u',v) \leq r(1,v)}r^{-2}(u',v)(u')^{-3}du'\Rx. \end{dmath*}
Now recall that $\frac{1}{3}(v-u) \leq r(u,v) \leq \frac{1}{2}(v-u)$; thus we conclude that
\begin{dmath*} u \geq v - 3r(u,v). \end{dmath*}
The first two terms can simply be evaluated directly. The last term can be rewritten
\begin{dmath*} \frac{36}{v^2}\int_{\frac{v}{6} \leq r(u',v)\leq r(1,v)} (u')^{-3}du' \leq \frac{36}{v^2}. \end{dmath*}
Then evaluating each term above using these bounds, we conclude that we can bound \eqref{small1} by
\begin{dmath*} \tilde{C}\Lx C_1(v-3)^{-6} + 4C_2v^{-2} + 36v^{-2}\Rx \end{dmath*}
with $C_1,C_2$ accounting for the necessary change of variable factors. Thus we see that we can take $v_0 > \frac{11\max\set*{C_1,C_2,36}\tilde{C} + 1}{\epsilon}$ and this satisfies the required conditions.

We will now address \eqref{small2}. In this case we have by \eqref{phiapriori} that
\begin{dmath*} \abs{\phi(u,v)} \leq \tilde{C}'\min\set*{u^{-2},r^{-1}(u,v)u^{-1}}. \end{dmath*}
Thus \eqref{small2} is bounded by
\begin{dmath*} \tilde{C}'\Lx \int_{0 \leq r(u',v) \leq \frac{v}{6}} \abs{(u')^{-2}}du' + \frac{2}{v}\int_{\frac{v}{6}\leq r(u',v) \leq r(1,v)}\abs{r^{-1}} \Rx. \end{dmath*}
Integrating, and absorbing constants from changing variables into an overall factor $\tilde{C}''$ we can bound \eqref{small2} by
\begin{dmath*} \tilde{C}''\Lx \frac{6}{v} + \abs{\frac{2}{v}(\log\abs{v} + C_4)} \Rx, \end{dmath*}
with $C_4$ absorbing any multiplicative factors in the logarithm (again from integration). Thus we can find a $v_0$ satisfying our requirements.

Finally, to obtain \eqref{small3}, we have
\begin{dmath*} \int_1^v\abs{\partial_v(r\phi)}(u',v)du' \lesssim \int_1^v \min\set*{r^{-2},(u')^{-2}}du' \lesssim v^{-1} \end{dmath*}
by simply integrating to $u' = \frac{v}{2}$ in $r^{-2}$, and the rest of the way using the $u$ bound. We thus, again, find a suitable $v_0$. Taking the largest of the three choices, we have our result.
\end{proof}

\begin{lemma}\label{vrphiboot}
In fact we have:
\begin{equation}\label{rphivboot} \abs{\Lx\frac{\partial_v}{\lambda}\Rx^{n+2}(r\phi)} \lesssim v^{-(n+3)}. \end{equation}
\end{lemma}

\begin{proof}
Recall from the asymptotic flatness and gauge conditions that we have \eqref{asymflat}:
\begin{dmath*} \sup_{C_1} v^{n+3}\abs{\Lx\frac{\partial_v}{\lambda}\Rx^{n+2}(r\phi)} \leq \mathcal{I}_{n+2}. \end{dmath*}
Moreover, observe that we have the following bound for $\frac{\partial_u}{\nu}\Lx\frac{\partial_v}{\lambda}\Rx^{n+2}(r\phi)$:
\begin{dmath*} \abs{\frac{\partial_u}{\nu}\Lx\frac{\partial_v}{\lambda}\Rx^{n+2}(r\phi)} \leq \abs{\partv{n+1}\Lx-\frac{m}{(1-\mu)r^2}(\nu^{-1}\partial_u(r\phi)) + \frac{m}{(1-\mu)r^2}\phi\Rx} + C'v^{-(n+3)} \leq C''\min\set{r^{-(n+3)},r^{-2}u^{-(n+1)}}, \end{dmath*}
using the results of \cref{mprelim,phiprelim}, with our remainder term contributed by those terms in which some derivative acts on the $\lambda^{-1},\nu^{-1}$ weights. In particular, for any $C > 0$ there is an $\epsilon_{C}$ such that
\begin{dmath*} \abs{\partial_v^{n+2}(r\phi)}(1+\delta,v) \leq (\mathcal{I}_{n+2} + C)v^{-(n+3)} \end{dmath*}
for all $\delta \leq \epsilon_{C}$. Thus we are in the scenario of \cref{vbootprelim} above, and we can conclude the results of \cref{vbootprelim} for $u \leq 1+\epsilon_{C}$. Our goal is to bootstrap this bound for $\partv{n+2}(r\phi)$, so in particular we would like to improve our control of the integral
\begin{dmath*} \int_1^u \partv{n+1}\Lx -\frac{m}{(1-\mu)r^2}(\lambda^{-1}\partial_v(r\phi)) + \frac{m}{(1-\mu)r^2}\phi\Rx du'. \end{dmath*}
Observe that by taking $C$ large enough, we can ignore all the integral terms that do not themselves include a factor of $C$. By \cref{vbootprelim} and our previous bounds, all terms immediately verify estimates with the proper $u$ and $v$ weights except those with all derivatives acting on one of $\phi,\mu$ and $\frac{m}{r^2}$. We will check these explicitly. In the $\phi$ case we have
\begin{dmath*} \int_1^u \abs{\frac{m}{(1-\mu)r^2}\partv{n+1}\phi} du' \leq CC'\min\set*{r^{-n+2}u^{-1},u^{-(n+3)}}\int_1^u \frac{m}{(1-\mu)r^2} du' \leq \epsilon CC' v^{-1}\min\set*{r^{-(n+2)}u^{-1},u^{-(n+3)}}, \end{dmath*}
using \cref{vbootprelim} and \eqref{small1}, so long as $v > v_0$ corresponding to the required $\epsilon$ (which depends only on $n$ by construction). By \eqref{muvboot} and our bounds on $\phi$ the $\mu$ case satisfies the same bound (in fact better, but the procedure is the same). It remains to check the $\frac{m}{r^2}$ case. Here we have:
\begin{dmath*} \int_1^u \abs{\frac{\lambda\nu\phi}{(1-\mu)}\partv{n+1}\frac{m}{r^2}}du' \leq CC'\min\set*{r^{-(n+3)}u^{-2},u^{-(n+6)}}\int_1^u \abs{\phi} \leq \epsilon CC'\min\set*{r^{-(n+3)}u^{-2},u^{-(n+6)}} \end{dmath*}
using \cref{vbootprelim}. We obtain the same result for the other $\partv{n+1}\frac{m}{r^2}$ term, as a consequence of \eqref{small3}. Since we can suppress our leading constants as much as we like for $v > v_0$, and the region $\set*{(u,v);v \leq v_0}$ is compact, if we take our $C$ large enough we can improve our estimates to:
\begin{dmath*} \abs{\partial_v^{n+2}(r\phi)}(1+\delta,v) \leq (\mathcal{I}_{n+2} + \frac{1}{2}C)v^{-(n+3)} + \int_u^v \frac{m}{(1-\mu)r^2}\partv{n+2}(r\phi) du'. \end{dmath*}
Using \eqref{small1}, and Gronwall's lemma, we can deal with this last term to obtain
\begin{dmath*} \partv{n+2}(r\phi)(1+\delta,v) \leq e^{\epsilon}(\mathcal{I}_{n+2} + \frac{1}{2}C)v^{-(n+3)}. \end{dmath*}
Taking $\epsilon$ small enough this represents a strict improvement on our initial bound.

Thus we conclude by continuity, and our preliminary estimates, that the region on which $\partial_v^{n+2}(r\phi)$ satisfies our bound is, in fact, strictly larger along each $\underline{C_v}$. In particular, we conclude that the region on which \eqref{rphivboot} is closed (by continuity), open, and non-empty (by the above). Thus it is our entire domain, so we have our required estimate.
\end{proof}

With this established we can also address the $u$ derivative case.

\begin{lemma} \label{ubootprelim} Suppose the following bound holds in a neighborhood $Op(\Gamma)$ of the axis:
\begin{equation} \abs{\partial_u^{n+2}(r\phi)} \leq Cu^{-(n+3)}. \end{equation}
Then we have the following bounds on a neighborhood of $\Gamma$:
\begin{equation}\label{phiuboot} \abs{\partial_u^{n+1}\phi} \leq CC'\min\set*{r^{-1}u^{-(n+2)},u^{-(n+3)}}, \end{equation}
\begin{equation}\label{muboot} \abs{\partial_u^{n+1}\frac{m}{r^2}} \leq CC'\min\set*{u^{-(n+4)}r^{-2},u^{-(n+6)}}, \end{equation}
\begin{equation}\label{muuboot} \abs{\partial_u^{n+1}\mu} \leq CC'\min\set*{r^{-2}u^{-(n+3)},u^{-(n+5)}}. \end{equation}
\end{lemma}

\begin{proof}
The approach for each term is essentially the same as above. In order to obtain \eqref{phiuboot} we again employ our averaging operators, now integrating in $u$ rather than $v$. We have
\begin{dmath*} \partial_u^{n+1}\phi(u,v) = \frac{1}{r^{n+2}(u,v)}\int_v^u\partial_u^{n+2}(r\phi)r^{n+1}du'. \end{dmath*}
So using the sup bound for our $u$ term, and integrating in $r$ we find that
\begin{dmath*} \abs{\partial_u^{n+1}\phi}(u,v) \leq Cu^{-(n+3)}. \end{dmath*}
If we instead use the sup bound for $r$, and integrate in $u$, we also obtain
\begin{dmath*} \abs{\partial_u^{n+1}\phi}(u,v) \lesssim r^{-1}u^{-(n+2)}. \end{dmath*}
Together these give our bound.

Now we can move to $\frac{m}{r^2}$. We can employ the same technique as above, noting that
\begin{dmath*} \partu{n+1}\frac{m}{r^2} = \partu{n}\Lx \frac{\partial_u m}{\nu r^2} \Rx - \partu{n}\frac{2m}{r^3}. \end{dmath*}
As above we will ignore terms in which $\nu$ is differentiated, as these are no worse than thosee in which $\frac{m}{r^3}$ is by hypothesis. Similar to the above, it suffices to control the expression
\begin{dmath*} \abs{\partial_u^{n}\Lx \frac{1}{2\nu}(1-\mu)(\partial_u\phi)^2 \Rx}. \end{dmath*}
This verifies the bound
\begin{dmath*} \abs{\partial_u^{n}\Lx \frac{1}{2\nu}(1-\mu)(\partial_u\phi)^2 \Rx} \leq CC' \min\set*{r^{-2}u^{-(n+4)},u^{-(n+6)}}. \end{dmath*}
This controls our first term directly, so we need only address our $\frac{m}{r^3}$ term. This time we obtain control by averaging in $u$:
\begin{dmath*} \partial_u^{n}\frac{m}{r^3}(u,v) = \frac{1}{r^{n+3}}\int_v^u \Lx\partial_u^{n}\Lx \frac{1}{2\nu}(1-\mu)(\partial_u\phi)^2 \Rx\Rx r^{n+2}du'. \end{dmath*}
Observe that, substituting our bounds in, this integral is the same as that used to bound \eqref{mvboot} but with the roles of $u$ and $r$ interchanged. Thus we obtain the following bound:
\begin{dmath*} \partu{n+1}\frac{m}{r^2} \leq CC_1\min\set*{u^{-(n+4)}r^{-2},u^{-(n+6)}}. \end{dmath*}

Finally we can address \eqref{muuboot}. As above we need only consider the term
\begin{dmath*} \partu{n}\frac{\partial_u m}{r\nu} = \partu{n} \Lx \frac{1}{2\nu}(1-\mu)r(\partial_u\phi)^2 \Rx. \end{dmath*}
Applying the same analysis as above, we conclude that
\begin{dmath*} \abs{\partu{n+1}\mu} \leq CC_2\min\set*{r^{-2}u^{-(n+3)},u^{-(n+5)}}, \end{dmath*}
so taking $C'$ sufficiently large we have our result.
\end{proof}

Again in parallel to the case above, we must obtain some small quantities in order to close this second bootstrap.
\begin{proposition}
For any $\epsilon > 0$ there is $U > 1$ such that for all $u > U$,
\begin{equation}\label{smallu1} \int_u^\infty \frac{m}{r^2}(u,v')dv' < \epsilon. \end{equation}
\end{proposition}

\begin{proof}
We begin with \eqref{smallu1}. By \eqref{mapriori} we have the bound
\begin{dmath*} \abs{\frac{m}{r^2}} \leq C_1\min\set*{r^{-2}u^{-3},ru^{-6}}. \end{dmath*}
So, integrating the second bound from $r = 0$ to 1, and the first the rest of the way, we obtain
\begin{dmath*} \int_u^\infty \abs{\frac{m}{r^2}}(u,v')dv' \leq C_1\Lx u^{-6}\int_{0 \leq r \leq 1} rdv' + u^{-3}\int_{r > 1} r^{-2}dv' \Rx \leq C_1C'(u^{-6} + u^{-3}) \leq (C_1C' + 1)u^{-3}, \end{dmath*}
where $C'$ absorbs any integration constants. Thus we have our result, taking $U > \frac{C_1C' + 1}{\epsilon}$.
\end{proof}

\begin{lemma}
Similar to \cref{vrphiboot} we in fact have:
\begin{equation}\label{rphiuboot} \abs{\partu{n+2}(r\phi)} \lesssim u^{-(n+3)}. \end{equation}
\end{lemma}

\begin{proof}
As above we will ignore the extraneous factors of $\lambda,\nu$ introduced by our weighted derivatives, since these cannot contribute anything worse than our desired decay.

Recall that, by our initial data, we have that, on $\Gamma$:
\begin{dmath*} \lim_{\epsilon,\delta\rightarrow 0^+}(\partial_u + \partial_v)^k(r\phi)(u-\epsilon,u+\delta) \equiv 0 \end{dmath*}
for all $k$. In particular, we obtain the bound:
\begin{dmath*} \abs{\partu{n+2}(r\phi)}(u,u) \leq Cu^{-(n+3)} \end{dmath*}
from our bounds \eqref{rphivboot}, and \cref{ind1mixrphi}, for some $C > 0$. Then by continuity, there is a neighborhood $N \supset \Gamma$ such that
\begin{dmath*} \abs{\partu{n+2}(r\phi)}(u,v) \leq 2Cu^{-(n+3)} \end{dmath*}
for all $(u,v) \in \overline{N}$. We can also estimate $\partu{n+2}(r\phi)$ by
\begin{dmath*} \abs{\partu{n+2}(r\phi)}(u,v) = \abs{\partu{n+2}(r\phi)(u,u)} + 3\abs{\int_u^v \partv{}\partu{n+2}(r\phi)(u,v')dv'}. \end{dmath*}
In turn we can bound this by:
\begin{dmath*} (C + C_1)u^{-(n+3)} + \tilde{C}\abs{\int_u^v \partu{n+1}\Lx -\frac{2m}{(1-\mu)r^2}(\nu^{-1}\partial_u(r\phi)) + \frac{2m}{(1-\mu)r^2}\phi \Rx(u,v')dv'} \end{dmath*}
for $C_1,\tilde{C}$ independent of our bootstrap constant $C$, in particular depending only on $n$, and bounds on low order derivatives of $\lambda,\nu$.

The only terms which can contribute constants proportional to $C$ are those in which all derivatives act on $\phi,\mu$ or $\frac{m}{r^2}$. In the first two cases, we find similar to the above, that this constant can be suppressed as much as we like outside some region with bounded $u$, by \eqref{phiuboot}, \eqref{muuboot}, \eqref{smallu1}. Thus we are left only to address the $\frac{m}{r^2}$ terms. First, we have
\begin{dmath*} \int_{u}^v \frac{\partu{}(r\phi)}{(1-\mu)}\partu{n+1}\frac{m}{r^2}dv' \leq CC'\int_u^v \min\set*{r^{-2}u^{-(n+6)},u^{-(n+8)}}dv' \end{dmath*}
for $C'$ dependent only on $n$. Splitting the integral in a manner similar to the above, we conclude that we can bound this term by
\begin{dmath*} \int_{u}^v \frac{\partv{}(r\phi)}{(1-\mu)}\partu{n+1}\frac{m}{r^2}dv' \leq 4CC'u^{-(n+7)}. \end{dmath*}
Thus, as above, this constant can be suppressed arbitrarily outside of some finite $u$ region.

Finally we have
\begin{dmath*} \int_u^v \frac{\phi}{(1-\mu)}\partv{n+1}\frac{m}{r^2}dv' \leq CC''\int_u^v\min\set*{r^{-3}u^{-(n+5)},u^{-(n+8)}}. \end{dmath*}
So once again integrating we obtain a bound
\begin{dmath*} \int_u^v \frac{\phi}{(1-\mu)}\partv{n+1}\frac{m}{r^2}dv' \leq 3CC''u^{-(n+7)}. \end{dmath*}

Thus we conclude that there is some universal $U > 1$ satisfying the conditions for \eqref{smallu1} to hold with small enough $\epsilon$, and moreover satisfying
\begin{dmath*} U^4 > 10\tilde{C}\max\set*{3C',4C''}. \end{dmath*}
In particular this $U$ does not depend on $C$ or $N$. Then taking $C$ large enough to absorb lower order constants, we conclude that we stricly improve our estimate in the region $u > U$. For $u \leq U$, we can simply take $C$ large enough to make our bound hold everywhere in this region if it does not already, since for any finite $u$ region we require only some constant bound, which holds immediately by continuity. Thus we conclude that we can improve our estimate to
\begin{dmath*} \abs{\partu{n+2}(r\phi)}(u,v) \leq \frac{3}{2}Cu^{-(n+3)} \end{dmath*}
on $N$. Then our initial bound is satisfied on some stricly larger neighborhood, and as above we conclude by connectedness and continuity that there is $C$ such that
\begin{dmath*} \abs{\partu{n+2}(r\phi)(u,v)} \leq 2Cu^{-(n+3)} \end{dmath*}
on all of $\mathcal{Q}$.
\end{proof}

\subsection{Full Decay for Derivatives of $\lambda,\nu$}
We are now ready to control non-mixed derivatives of $r$.

\begin{lemma}
\begin{equation}\label{lambdavfull} \abs{\partv{n+1}\lambda} \lesssim \min\set*{r^{-(n+3)},u^{-(n+3)}}, \end{equation}
\begin{equation}\label{nuufull} \abs{\partu{n+1}\nu} \lesssim \min\set*{u^{-(n+3)}}. \end{equation}
\end{lemma}

\begin{proof}
As usual, we can safely ignore terms where some derivatives act on different copies of $\lambda$ or $\nu$, as there immediately satisfy our bounds as a conequence of our inductive hypothesis.

We'll start with \eqref{lambdavfull}. To begin, we use our gauge condition to obtain:
\begin{dmath*} \partv{n+1}\lambda(u,v) = \int_1^u \nu^{-1}\partial_u\partv{n+1}\lambda(u',v)du'. \end{dmath*}
The only term we must control on the RHS is
\begin{dmath*} \int_1^u \abs{\partv{n+1}\Lx \frac{m\lambda\nu}{(1-\mu)r^2} \Rx} du' \leq C\min\set*{r^{-(n+3)},u^{-(n+3)}} + \int_1^u\frac{m}{(1-\mu)r^2}\partv{n+1}\lambda du', \end{dmath*}
since all the other terms satisfy the required decay up to some constant $C$ by our hypotheses, \eqref{mvboot}, \eqref{muvboot}, \cref{ind1mixr} and the split integration employed above. This can then be controlled by Gronwall's lemma, and \eqref{small1}, thus we conclude \eqref{lambdavfull}.

For \eqref{nuufull}, we gain initial control by the second half of our gauge condition:
\begin{dmath*} (\partial_u + \partial_v)^kr = 0 \end{dmath*}
on $\Gamma$. As a result of this, \eqref{lambdavfull} and \cref{ind1mixr}, we conclude that
\begin{dmath*} \abs{\partu{n+1}\nu}(u,u) \leq Cu^{-(n+3)}. \end{dmath*}
Thus we have
\begin{dmath*} \abs{\partu{n+1}\nu}(u,v) \leq Cu^{-(n+3)} + \int_u^v \partial_v\partu{n+1}\nu(u,v')dv'. \end{dmath*}
As above, we can commute the $v$ derivative past all the others, and then we only need to control the term
\begin{dmath*} \int_u^v \partu{n+1}\frac{m\lambda\nu}{(1-\mu)r^2}dv'. \end{dmath*}
By \eqref{muuboot}, \eqref{muboot}, \cref{ind1mixr}, and splitting the integral into the regions $r<1$, $r > 1$ we control all terms adequately except that in which all derivatives act on $\nu$. In this case, we once again employ Gronwall and \eqref{smallu1}, and conclude the desired bound.
\end{proof}

\subsection{Full Decay for Derivatives of $\phi$,$\frac{m}{r^k}$}

Following the results of \cref{vbootprelim,rphivboot,rphiuboot,ubootprelim}, we can conclude optimal bounds for all the $n+1$st derivatives of $\phi,\mu$ and $\frac{m}{r^2}$. To aid in this, we will need some auxiliary estimates:
\begin{proposition}
\label{excprop1}
For $\phi$ a solution to ~\eqref{SSESF}
\begin{dmath*} 2\lambda\abs{\partial_u\partial_v \phi}(u,v) = -\partial_v(\nu\partial_v\phi) - \partial_v\lambda\partial_u\phi - r\partial_v^2\partial_u\phi. \end{dmath*}
\end{proposition}
\begin{proof}
Observe that we have
\begin{dmath*} \partial_u\partial_v \phi = \partial_u\partial_v(r\phi) - \frac{\lambda}{r}\partial_u\phi - \frac{\nu}{r}\partial_v\phi - \phi\partial_u\partial_vr. \end{dmath*}
By \eqref{SSESF} $\partial_u\partial_v(r\phi) = \phi\partial_u\partial_v r$, so we are left with
\begin{dmath*} \partial_u\partial_v\phi = -\frac{\lambda}{r}\partial_u\phi - \frac{\nu}{r}\partial_v\phi. \end{dmath*}
Now consider
\begin{dmath*} \partial^2_v(r\partial_u\phi) = \partial_v\lambda\partial_u\phi + 2\lambda\partial_u\partial_v\phi + r\partial_v^2\partial_u\phi. \end{dmath*}
Expanding this last term we see
\begin{align*}
r\partial_v^2\partial_u\phi &= -r\partial_v\Lx \frac{\lambda}{r}\partial_u\phi + \frac{\nu}{r}\partial_v\phi \Rx\\
&= -\partial_v\lambda\partial_u\phi - \lambda\partial_u\partial_v\phi - \nu\partial_v^2\phi +  - \partial_v\nu\partial_v\phi + \frac{\lambda^2}{r}\partial_u\phi + \frac{\lambda\nu}{r}\partial_v\phi\\
&= -\partial_v(\nu\partial_v\phi) - 2\lambda\partial_u\partial_v\phi - \partial_v\lambda\partial_u\phi.
\end{align*}
Returning to our original equation we find that
\begin{dmath*} \partial_v^2(r\partial_u\phi) = -\partial_v(\nu\partial_v\phi), \end{dmath*}
thus, substituting this in and rearranging we have
\begin{dmath*} 2\lambda\partial_u\partial_v\phi = -\Lx\partial_v(\nu\partial_v\phi) + \partial_v\lambda\partial_u\phi + r\partial_v^2\partial_u\phi\Rx. \end{dmath*}
\end{proof}

We also note the following relation:
\begin{equation}
\label{uvmexc}
\nu^{-1}\partial_u \Lx\frac{m}{r^k}\Rx - \lambda^{-1}\partial_v \Lx\frac{m}{r^k}\Rx = \frac{r^{2-k}}{2}(1-\mu)\Lx (\nu^{-1}\partial_u \phi)^2 - (\lambda^{-1}\partial_v\phi)^2 \Rx
\end{equation}
for $k = 1,2$.

We begin with $\phi$.
\begin{lemma}
\label{ind1phi}
The following holds for $\abs{\alpha} = n+1$.
\begin{equation}\label{phimixbound} \abs{\partu{\alpha_u}\partv{\alpha_v}\phi(u,v)} \lesssim \min\set*{r^{-(\alpha_v + 1)}u^{-(\alpha_u+1)},u^{-(\abs{\alpha} + 2)}}. \end{equation}
\end{lemma}

\begin{proof}
The non-mixed cases have already been done, and moreover, the $r$ weighted bounds follow immediately from \eqref{ind1mixrphi}, \eqref{rphivboot}, \eqref{rphiuboot} and our inductive hypothesis by the same method as \cref{phiprelim}. Thus it remains only to address the non-$r$-weighted bounds. In order to do this we employ \cref{excprop1}, noting that, as a consequence, we can write for $\abs{\alpha} = n+1, \alpha_u,\alpha_v > 0$:
\begin{dmath*} \partial_u^{\alpha_u-1}\partial_v^{\alpha_v-1}(2\lambda\partial_u\partial_v\phi) = -\partial_u^{\alpha_u-1}\partial_v^{\alpha_v-1}\Lx\partial_v(\nu\partial_v\phi) + \partial_v\lambda\partial_u\phi + r\partial_v^2\partial_u\phi\Rx. \end{dmath*}
The left hand side can be rewritten
\begin{dmath*} \partial_u^{\alpha_u-1}\partial_v^{\alpha_v-1}(2\lambda\partial_u\partial_v\phi) = \nu^{\alpha_u}\lambda^{\alpha_v}(2\lambda \partu{\alpha_u}\partv{\alpha_v}\phi + R_0) \end{dmath*}
where $R_0$ equal to the sum of remaining terms (those with some product of derivatives of at least two of $\lambda,\nu$ and $\partial_u\partial_v\phi$) all of which are controlled by hypothesis as each individual term is of lower total differential order. In particular, we have that
\begin{dmath*} \abs{R} \lesssim u^{-(\abs{\alpha} + 4)}. \end{dmath*}
On the RHS, we will work term by term:

First we have
\begin{dmath*} \partial_u^{\alpha_u-1}\partial_v{\alpha_v}(\nu\partial_v\phi)) = \nu^{\alpha_u-1}\lambda^{\alpha_v+1}(\nu \partu{\alpha_u-1}\partv{\alpha_v+1}\phi + R_1) \end{dmath*}
where as above $R_1$ collects terms with derivatives acting on at least two of $\lambda,\nu,\phi$. This is similarly controlled by the hypotheses \eqref{ind1hypphi}, \eqref{ind1hypuvlambda}, \eqref{ind1hypvlambda}, \eqref{ind1hypunu}, and \eqref{ind1hypuvnu} as each derivative is of total order smaller than $n+1$, so we have
\begin{dmath*} \abs{R_1} \lesssim u^{-(\abs{\alpha} + 4)}. \end{dmath*}
The next term is
\begin{dmath*} \partial_u^{\alpha_u-1}\partial_v^{\alpha_v-1}(\partial_v\lambda\partial_u\phi) = \nu^{\alpha_u}\lambda^{\alpha_v}\partu{\alpha_u-1}\partv{\alpha_v-1}(\partv{}\lambda\partu{}\phi) + R_2, \end{dmath*}
with $R_2$ defined similarly to $R_0,R_1$ above. But all of these terms are immediately bounded (up to overall constant) by $u^{-(\abs{\alpha}+4)}$ by \eqref{ind1hypphi}, \eqref{ind1hypuvlambda}, \eqref{ind1hypvlambda}, \eqref{ind1hypunu}, and \eqref{ind1hypuvnu} since no single term receives more than $n$ derivatives.

Finally, we consider
\begin{dmath*} \partial_u^{\alpha_u-1}\partial_v^{\alpha_v-1}(r\partial_v^2\partial_u\phi) = (\alpha_v-1)\lambda\partial_u^{\alpha_u}\partial_v^{\alpha_v}\phi + (\alpha_u-1)\nu\partial_u^{\alpha_u-1}\partial_v^{\alpha_v+1}\phi + r\partial_u^{\alpha_u}\partial_v^{\alpha_v+1} + R_3, \end{dmath*}
where $R_3$ absorbs all terms with at least two derivatives acting on $r$. We can rewrite this further in terms of our gauge invariant derivatives as
\begin{dmath}
\partial_u^{\alpha_u-1}\partial_v^{\alpha_v-1}(r\partial_v^2\partial_u\phi) = \nu^{\alpha_u}\lambda^{\alpha_v}(\alpha_v-1)\lambda\partu{\alpha_u}\partv{\alpha_v}\phi + \nu^{\alpha_u-1}\lambda^{\alpha_v+1}(\alpha_u-1)\nu\partu{\alpha_u-1}\partv{\alpha_v+1}\phi + r\nu^{\alpha_u}\lambda^{\alpha_v+1}\partu{\alpha_u}\partv{\alpha_v+1}\phi - r\nu^{\alpha_u}\lambda^{\alpha_v+1}\left(\partu{}\lambda + \partu{}\nu\right)\partu{\alpha_u-1}\partv{\alpha_v+1}\phi - r\nu^{\alpha_u}\lambda^{\alpha_v+1}\partv{}\lambda \partu{\alpha_u}\partv{\alpha_v}\phi + R_3',
\end{dmath}
where $R_3'$ additionally absorbs the additional mixed derivative remainder terms from converting to gauge invariant derivatives excluding those from $r\partial_u^{\alpha_u}\partial_v^{\alpha_v+1}\phi$ already included above. This term is controlled by \eqref{ind1hypphi}, \eqref{ind1mixr}, \eqref{ind1hypuvlambda}, \eqref{ind1hypvlambda}, \eqref{nd1hypunu}, and \eqref{ind1hypuvnu} as no derivative of order greater than $n$ acts on $\phi$, and no term of order beyond $n+1$ acts on $\lambda,\nu$. As such this term satisfies
\begin{dmath*} \abs{R_3'}\lesssim u^{-(\abs{\alpha}+4)}. \end{dmath*}
Combining all of this we arrive at the following equation:

\begin{equation}\begin{split}
&\nu^{\alpha_u}\lambda^{\alpha_v}(2\lambda \partu{\alpha_u}\partv{\alpha_v}\phi + R_0) =\\ &-\Lx \nu^{\alpha_u-1}\lambda^{\alpha_v+1}\nu \partu{\alpha_u-1}\partv{\alpha_v+1}\phi\right.\\ &+ \nu^{\alpha_u}\lambda^{\alpha_v}\partu{\alpha_u-1}\partv{\alpha_v-1}(\partv{}\lambda\partu{}\phi)\\ &+  \nu^{\alpha_u}\lambda^{\alpha_v}(\alpha_v-1)\lambda\partu{\alpha_u}\partv{\alpha_v}\phi\\ &+ \nu^{\alpha_u-1}\lambda^{\alpha_v+1}(\alpha_u-1)\nu\partu{\alpha_u-1}\partv{\alpha_v+1}\phi\\ &+ r\nu^{\alpha_u}\lambda^{\alpha_v+1}\partu{\alpha_u}\partv{\alpha_v+1}\phi\\ &- r\nu^{\alpha_u}\lambda^{\alpha_v+1}(\partu{}\lambda + \partu{}\nu)\partu{\alpha_u-1}\partv{\alpha_v+1}\phi\\ &\left.- r\nu^{\alpha_u}\lambda^{\alpha_v+1}\partv{}\lambda \partu{\alpha_u}\partv{\alpha_v}\phi + R_1 + R_2 + R_3' \Rx.
\end{split}\end{equation}
Observe that we have two terms proportional to $\partu{\alpha_u}\partv{\alpha_v}\phi$ on the RHS, so we can rearrange to obtain:
\begin{equation}\label{ind1phiexc}\begin{split}
&\nu^{\alpha_u}\lambda^{\alpha_v+1}(\alpha_v+1 - \partv{}\lambda) \partu{\alpha_u}\partv{\alpha_v}\phi =\\ &-\Lx \nu^{\alpha_u-1}\lambda^{\alpha_v+1}\nu \partu{\alpha_u-1}\partv{\alpha_v+1}\phi\right.\\&+ \nu^{\alpha_u}\lambda^{\alpha_v}\partu{\alpha_u-1}\partv{\alpha_v-1}(\partv{}\lambda\partu{}\phi)\\ &+ \nu^{\alpha_u-1}\lambda^{\alpha_v+1}(\alpha_u-1)\nu\partu{\alpha_u-1}\partv{\alpha_v+1}\phi\\&+ r\nu^{\alpha_u}\lambda^{\alpha_v+1}\partu{\alpha_u}\partv{\alpha_v+1}\phi\\& - r\nu^{\alpha_u}\lambda^{\alpha_v+1}\left(\partu{}\lambda + \partu{}\nu\right)\partu{\alpha_u-1}\partv{\alpha_v+1}\phi\\
&\left.+ R_0 + R_1 + R_2 + R_3' \Rx.
\end{split}
\end{equation}
Then recall that $\partv{}\lambda \leq C\min\set{r^{-3},u^{-3}}$, for some $C > 0$, thus outside of some compact region we have $(\alpha_v+1-\partv{}\lambda) > 1$ for all $\alpha_v \geq 1$. Since the region is compact we need not be concerned with the behavior inside as this can be absorbed by a constant.

Thus it suffices to bound the RHS of \eqref{ind1phiexc} in order to control $\partu{\alpha_u}\partv{\alpha_v}\phi$. We will do this inductively, inducting on $\alpha_u$. The base case $\alpha_u = 0$ is already covered by our bootstrap, so suppose the bound \eqref{phimixbound} holds for $\alpha_u < k < n$. So treating each term individually, we have by our inductive hypothesis (on $\alpha_u$)
\begin{dmath*} \abs{\nu\partu{\alpha_u-1}\partv{\alpha_v+1}\phi} \lesssim \min\set*{r^{-(\alpha_v + 2)}u^{-(\alpha_u)},u^{-(\abs{\alpha} + 2)}}, \end{dmath*}
which is strictly better than required (the estimates agree near the axis, and the $r$ estimate is strictly better at large $r$). Thus we can move to our next term:
\begin{dmath*} \partu{\alpha_u-1}\partv{\alpha_v-1}(\partv{}\lambda\partu{}\phi), \end{dmath*}
but there is nothing to do here, as all these terms are controlled by hypothesis, and thus immediately verify the necessary bounds.

Next, we consider
\begin{dmath*} \abs{\nu^{\alpha_u-1}\lambda^{\alpha_v+1}(\alpha_u-1)\nu\partu{\alpha_u-1}\partv{\alpha_v+1}\phi}. \end{dmath*}
This now is controlled by our inductive hypothesis, as we take one fewer derivatives in $u$ and overall order $n+1$. This leaves us
\begin{dmath*} \abs{r\nu^{\alpha_u}\lambda^{\alpha_v+1}\partu{\alpha_u}\partv{\alpha_v+1}\phi} \lesssim rr^{-1}u^{-(\abs{\alpha}+2)} = u^{-(\abs{\alpha+2})} \end{dmath*}
where we obtain this bound by the same mechanism as \cref{phiprelim} at the next order.
Finally, we have
\begin{dmath*} \abs{r\nu^{\alpha_u}\lambda^{\alpha_v+1}\left(\partu{}\lambda + \partu{}\nu\right)\partu{\alpha_u-1}\partv{\alpha_v+1}\phi} \lesssim ru^{-3} r^{-1}u^{-(\abs{\alpha} + 1)}, \end{dmath*}
again obtaining an $r^{-1}$ weighted bound in the manner of \cref{phiprelim}, and using the $u^{-3}$ bound for $\partu{}\lambda,\partu{}\nu$ of \cref{dunuapriori}, \eqref{ind1hypuvlambda}.

Putting this all together, we obtain
\[ \abs{\partu{\alpha_u}\partv{\alpha_v}\phi} \lesssim u^{-(\abs{\alpha}+2)}. \]
So by induction on $\alpha_u$ we have our result.
\end{proof}

It remains only to check that $\mu,\frac{m}{r^2}$ also verify the required estimates to complete our induction:
\begin{lemma}\label{ind1uvmr}
For $\abs{\alpha} = n+1, \alpha_u,\alpha_v \neq 0$ we have
\begin{dmath}\label{uvmrbounds} \abs{\partu{\alpha_u}\partv{\alpha_v}\frac{m}{r^k}} \lesssim \min\set*{u^{-(n+4+k)}} \end{dmath}
for $k = 1,2$.
\end{lemma}

\begin{proof}
Recall that for $\alpha_v = n+1,\alpha_u = n+1$ we already have our optimal bounds via the bootstrap. Moreover optimal $r$-weighted bounds are achieved in general in \cref{mweight}, so we need only concern ourselves with bounds in terms of $u$ only. We will proceed via induction in $\alpha_u$ with base case 0 already done. To induct, suppose \eqref{uvmrbounds} already holds for all $\alpha_u < l < n$. By \cref{uvexclemma} it suffices to control $\partu{l-1}\partv{n+1-l}\partu{}\frac{m}{r^k}$ in order to proceed to the next level. Then by \cref{uvmexc} we have
\begin{equation*}\begin{split} &\partu{l-1}\partv{n+1-l}\partu{} \frac{m}{r^k} =\\&\partu{l-1} \partv{n+1-l}\Lx \partv{}\frac{m}{r^k} + \frac{r^{2-k}}{2}(1-\mu)\Lx (\nu^{-1}\partial_u \phi)^2 - (\lambda^{-1}\partial_v\phi)^2 \Rx \Rx.\end{split}\end{equation*}
Splitting this up term by term we have
\[ \abs{\partu{l-1}\partv{n+2-l}\frac{m}{r^k}} \lesssim \min\set*{u^{-(n+4+k)}} \]
by our inductive hypothesis. Next we have
\[ \partu{l-1}\partv{n+1-l} \frac{r^{k-2}}{2}(\partu{}\phi)^2, \]
and
\[ \partu{l-1}\partv{n+1-l} \frac{r^{k-2}}{2}(\partv{}\phi)^2. \]
Observe that regardless of the derivative that acts on $\phi$ these terms will satisfy the same order of $u$ decay by our above results, so we will only prove this for the first of these terms.

When $k = 2$, this is simply
\[ \partu{l-1}\partv{n+1-l}(\partu{}\phi)^2 \lesssim u^{-(n+6)} \]
by \eqref{ind1hypphi}. When $k=1$ we have
\[ \partu{l-1}\partv{n-l}r(\partu{}\phi)^2 \lesssim u^{-(n+5)} \]
by \eqref{ind1hypphi}, obtaining a term with an $r^{-1}$ (in the case all derivatives act on the $\partu{}\phi$'s) by the methods of \cref{phiprelim}.

Putting this together, we have our result by induction.
\end{proof}

\subsection{Closing Induction}
Finally we can close our induction checking that our bounds hold at order $\abs{\alpha} = 1$.

\section{Extension to the First Quadrant}
\label{oproof}
To prove \cref{main1} it remains to control solutions in the region $\mathbb{I} \setminus \mathcal{Q}$. In order to do this we will have to turn to our full data prescribed on both $C_1$ and $\underline{C}_R$. In fact, we will establish estimates on all of the region $\mathcal{O}_R$, completing the proof of \cref{main1}\footnote{In fact this leaves a compact region of $\mathbb{I}$ without explicit control. By standard persistence of regularity, the solution is still $C^k$ smooth in this region (for some discussion cf. \cite{Dafermos2003}), and thus the decay can be realized simply by adjusting our constants.}. The initial data considered in \cref{Qproof} is thus extended by any compatible prescription of $\partial_u(r\phi),\nu$ on $\underline{C}_{2R}$ (we modify our notation here slightly to ease things later). Note that this is a well posed problem since we work in the spherically symmetrically reduced setting. This extension is studied at first order in \cite{Dafermos2003}, and one can check by a standard iteration argument that this is a well posed initial value problem on the region between $C_1$ and $\underline{C}_{2R}$ for sufficiently regular data, so long as one guarantees that the data agree at the intersection point $C_1\cap \underline{C}_{2R}$.

In this section we will consider data which is asymptotically flat of order $\omega' \geq 2$ in $C^k$ towards both $\mathcal{I}^+$ and $\mathcal{I}^-$ and the gauge \ref{gauge3}.

As in \cref{Qproof} we will inductively establish the control that we need. The goal will be the following
\begin{theorem}\label{Obounds}
Let $(r,\phi,m)$ a solution to \eqref{SSESF} in the region $\mathcal{O}_{2R}$ with data asymptotically flat of order $\omega' \geq 2$ in $C^k$ towards both $\mathcal{I}^+$ and $\mathcal{I}^-$. Then the following bounds hold for all multi-indices $\alpha,\beta$ with $\abs{\alpha} \leq k$, $\abs{\beta} \leq k+1$:
\begin{gather}
\abs{\partv{\abs{\alpha}}\lambda} \lesssim v^{-(\abs{\alpha}+1)},\\
\label{oulambdabound}\abs{\partu{\alpha_u}\partv{\alpha_v}\lambda} \lesssim r^{-1}(1+\abs{u})^{-\alpha_u}v^{-\alpha_v},\\
\abs{\partu{\abs{\alpha}}\nu} \lesssim (\abs{u}+1)^{-(\abs{\alpha}+1)},\\
\label{ovnubound}\abs{\partu{\alpha_u}\partv{\alpha_v}\nu} \lesssim r^{-1}(1+\abs{u})^{-\alpha_u}v^{-\alpha_v},\\
\abs{\partu{\alpha_u}\partv{\alpha_v}\phi} \lesssim r^{-1}(1+\abs{u})^{-\alpha_u}v^{-\alpha_v},\\
\abs{\partu{\alpha_u}\partv{\alpha_v}m} \lesssim \min\set*{(1+\abs{u})^{-(\alpha_u+1)}v^{-\alpha_v},(1+\abs{u})^{-\alpha_u}v^{-(\alpha_v+1)}},\\
\abs{\partu{\alpha_u}\partv{\alpha_v}\frac{m}{r^k}} \lesssim r^{-k}\min\set*{(1+\abs{u})^{-(\alpha_u+1)}v^{-\alpha_v},(1+\abs{u})^{-\alpha_u}v^{-(\alpha_v+1)}},\\
\label{ouvrphibound}\abs{\partu{\beta_u}\partv{\beta_v}(r\phi)} \lesssim r^{-1}(1+\abs{u})^{-\alpha_u}v^{-\alpha_v},\\
\abs{\partv{\abs{\beta}}(r\phi)} \lesssim v^{-(\abs{\beta}+1)},\\
\abs{\partu{\abs{\beta}}(r\phi)} \lesssim (1+\abs{u})^{-(\abs{\beta}+1)},
\end{gather}
where we take $\alpha_u,\beta_u$ to be non-zero in \cref{oulambdabound,ouvrphibound} respectively and $\alpha_v,\beta_v$ non-zero in \cref{ovnubound,ouvrphibound} respectively.
\end{theorem}

\subsection{First Estimates for $\lambda,\nu,r\phi$ and $m$}
We will make use of the following result of \cite{Dafermos2003} (cf. Proposition 5):
\begin{proposition}\label{dafinit}
There exists $R_0 > 1$ (depending on the size of our data) such that for $R > R_0$ the domain of the solution to \eqref{SSESF} with data posed on $\underline{C}_{2R}$ (in the $u,v$ coordinates) is $\mathcal{O} = [-\infty,1] \times [2R,\infty]$, and moreover the following estimates hold on $\mathcal{O}$:
\begin{gather}\label{orphiinit} \abs{r\phi} \lesssim 1,\\
\label{olambdainit} \frac{1}{2} \leq \lambda \leq 2,\\
\label{onuinit} \frac{1}{2} \leq -\nu \leq 2,\\
\label{omuinit} \frac{1}{4} \leq 1-\mu \leq 1,\\
\label{ourphiinit} \abs{\partu{}(r\phi)} \lesssim (\abs{u}+1)^{-2},\\
\label{ovrphiinit} \abs{\partv{}(r\phi)} \lesssim v^{-2},\\
\label{mextinit} \abs{m} \lesssim 1.\end{gather}
\end{proposition}

We'll now obtain some additional low order bounds necessary to begin our bounding by induction on the order of derivatives.
\begin{proposition}\label{oorder1r}
The following bounds hold on $\mathcal{O}$:
\begin{gather}
\label{ovlambda} \abs{\partv{}\lambda} \lesssim v^{-2},\\
\label{oulambda} \abs{\partu{}\lambda} \lesssim r^{-2},\\
\label{ounu} \abs{\partu{}\nu} \lesssim (\abs{u}+1)^{-2},\\
\label{ovnu} \abs{\partv{}\nu} \lesssim r^{-2},\\
\label{ovphi} \abs{\partv{}\phi} \lesssim r^{-1}v^{-1},\\
\label{ouphi} \abs{\partu{}\phi} \lesssim r^{-1}(\abs{u}+1)^{-1},\\
\label{ovm} \abs{\partv{}m} \lesssim v^{-2},\\
\label{oum} \abs{\partu{}m} \lesssim (\abs{u}+1)^{-2},\\
\label{ovmr} \abs{\partv{}\frac{m}{r^k}} \lesssim r^{-k}v^{-1},\\
\label{oumr} \abs{\partu{}\frac{m}{r^k}} \lesssim r^{-k}(\abs{u}+1)^{-1},\\
\label{ouvrphi} \abs{\partu{}\partv{}(r\phi)} \lesssim r^{-2}v^{-1},\\
\label{ovurphi} \abs{\partv{}\partu{}(r\phi)} \lesssim r^{-2}(\abs{u}+1)^{-1},\\
\label{ov2rphi} \abs{\partv{2}(r\phi)} \lesssim v^{-3},\\
\label{ou2rphi} \abs{\partu{2}(r\phi)} \lesssim (\abs{u}+1)^{-3}.
\end{gather}
\end{proposition}

\begin{proof}
\eqref{oulambda}, \eqref{ovnu}, \eqref{ouvrphi}, and \eqref{ovurphi} can be read off directly from \eqref{SSESF} and \cref{dafinit}.

From here we will begin with \eqref{ovphi}, and \eqref{ouphi}. For \eqref{ovphi} we have
\begin{dmath*} \partv{}(r\phi) = \abs{\phi + r\partv{}\phi}. \end{dmath*}
So rearranging we have by \eqref{ovrphiinit}, and \eqref{orphiinit}
\begin{dmath*} \abs{\partv{}\phi} \leq r^{-1}\Lx\abs{\partv{}(r\phi)} + \abs{\phi}\Rx \lesssim r^{-2} + r^{-1}v^{-2} \leq r^{-1}v^{-1}. \end{dmath*}
The same can be done for $\partu{}\phi$ using \eqref{ourphiinit},and \eqref{orphiinit}, and we obtain \eqref{ouphi} as well.

Now we can move to our bounds for $m$ \eqref{ovm}, and \eqref{oum}. For \eqref{ovm} we have
\begin{dmath*} \abs{\partv{}m} = \abs{\frac{1}{2}(1-\mu)r^2(\partv{}\phi)^2} \lesssim v^{-2} \end{dmath*}
using \eqref{ovphi}. We similarly obtain \eqref{oum} using \eqref{ouphi}.

From this \eqref{ovmr}, and \eqref{oumr} follow immediately via the Leibniz rule and \eqref{olambdainit}, \eqref{onuinit}.

Next, for \eqref{ovlambda}:
\begin{dmath*} \abs{\partv{}\lambda}(u,v) = \abs{\int_1^u \partial_u\partv{}\lambda(u',v)du'} = \abs{-\int_1^u \lambda^{-2}\partial_u\lambda \partial_v\lambda du' + \int_1^u \lambda^{-1}\partial_v\partial_u\lambda du'}. \end{dmath*}
By our assumptions on $R$, and \eqref{oulambda} $\int_1^u \abs{\partial_u\lambda(u',v)} du' \lesssim r^{-1}(1,v)$ for all $v \geq 2R$, so we can apply Gronwall's inequality to deal with our first term, and our decay will be determined by the second term. Applying \eqref{SSESF} this can be bounded by
\begin{dmath*} \int_1^{-\infty}\lambda^{-1} \partial_v\Lx \frac{2m\lambda\nu}{(1-\mu)r^2} \Rx du'. \end{dmath*}
The $\partial_v\lambda$ can be grouped with our other term containing this in our application of Gronwall (the coefficient decays like $\frac{m}{r^2}$ and thus is integrable by \eqref{oum}), so our decay is determined by the remaining terms, thus bounded by $r^{-2}v^{-1}$ (with the term differentiating $r^{-2}$ having the lowest order of decay), using \eqref{ovnu}, \eqref{ovm}. Thus integrating we obtain \eqref{ovlambda}. Again repeating the same procedure for $\partv{}\nu$ obtains the symmetric $u$ bound \eqref{ounu}.

Now, for \eqref{ov2rphi} we have
\begin{dmath*} \partv{2}(r\phi)(u,v) = \partv{2}(r\phi)(1,v) + \int_1^{u} \partial_u\partv{2}(r\phi)(u',v)du'. \end{dmath*}
This leading term is bounded by $v^{-3}$, and the latter term can be written
\begin{dmath*} \int_1^u \partv{}(\partial_u\partv{}(r\phi))(u',v) - \partial_u\lambda \lambda^{-1} \partv{2}(r\phi)(u',v)du'. \end{dmath*}
The second term here can be controlled via Gronwall's inequality since $\partial_u \lambda$ is integrable by \eqref{oulambda}. The first term can be expanded:
\begin{dmath*} \int_1^u \partv{} \Lx\frac{m\nu}{(1-\mu)r^2}\Lx \phi - \partv{}(r\phi) \Rx\Rx(u',v)du'. \end{dmath*}
As usual, we can absorb the $\partv{2}(r\phi)$ term via Gronwall, so we are left with only the other terms. By the above results, we have that these are all bounded by $r^{-2}v^{-2}$, so we obtain as an overall bound:
\begin{dmath*} \abs{\partv{2}(r\phi)}(u,v) \lesssim v^{-3} \end{dmath*}
since our integral contributes a $r^{-1}v^{-2}$ which is strictly smaller than $v^{-3}$ on $\mathcal{O}$. One sees directly that \eqref{ou2rphi} can be obtained in the same manner by exchanging the roles of $u$ and $v$ at each step, and substituting the correct bounds from above.
\end{proof}

\subsection{Higher Order Derivatives}
The remainder of this section will be devoted to closing the following induction which completes the proof of \cref{Obounds}:
\begin{lemma}\label{ind2}
Let $r,\phi,m$ solve \eqref{SSESF} in the region $\mathcal{O}$ with initial data smooth, asymptotically flat to order $N$, satisfying the assumptions of \cref{dafinit}. Suppose the following bounds hold for multi-indices $\abs{\alpha} \leq n, \abs{\beta} \leq n+1 < N$:
\begin{gather}
\label{ind2hypvlambda} \abs{\partv{\abs{\alpha}}\lambda} \lesssim v^{-(\abs{\alpha}+1)},\\
\label{ind2hypulambda} \abs{\partu{\alpha_u}\partv{\alpha_v}\lambda} \lesssim r^{-1}(1+\abs{u})^{-\alpha_u}v^{-\alpha_v},\\
\label{ind2hypunu} \abs{\partu{\abs{\alpha}}\nu} \lesssim (\abs{u}+1)^{-(\abs{\alpha}+1)},\\
\label{ind2hypvnu} \abs{\partu{\alpha_u}\partv{\alpha_v}\nu} \lesssim r^{-1}(1+\abs{u})^{-\alpha_u}v^{-\alpha_v},\\
\label{ind2hypuvphi} \abs{\partu{\alpha_u}\partv{\alpha_v}\phi} \lesssim r^{-1}(1+\abs{u})^{-\alpha_u}v^{-\alpha_v},\\
\label{ind2hypm} \abs{\partu{\alpha_u}\partv{\alpha_v}m} \lesssim \min\set*{(1+\abs{u})^{-(\alpha_u+1)}v^{-\alpha_v},(1+\abs{u})^{-\alpha_u}v^{-(\alpha_v+1)}},\\
\label{ind2hypmr} \abs{\partu{\alpha_u}\partv{\alpha_v}\frac{m}{r^k}} \lesssim r^{-k}\min\set*{(1+\abs{u})^{-(\alpha_u+1)}v^{-\alpha_v},(1+\abs{u})^{-\alpha_u}v^{-(\alpha_v+1)}},\\
\label{ind2hypuvrphi} \abs{\partu{\beta_u}\partv{\beta_v}(r\phi)} \lesssim r^{-1}(1+\abs{u})^{-\beta_u}v^{-\beta_v},\\
\label{ind2hypvrphi} \abs{\partv{\abs{\beta}}(r\phi)} \lesssim v^{-(\abs{\beta}+1)},\\
\label{ind2hypurphi} \abs{\partu{\abs{\beta}}(r\phi)} \lesssim (1+\abs{u})^{-(\abs{\beta}+1)},
\end{gather}
where we take $\alpha_u$ to be non-zero in \eqref{ind2hypulambda} respectively and $\alpha_v$ non-zero in \eqref{ind2hypvnu} and $\beta_u,\beta_v$ both non-zero in \eqref{ind2hypuvrphi}.

Then in fact these estimates hold for $\abs{\alpha} \leq n+1, \abs{\beta} \leq n+2$.
\end{lemma}

Before proving this we must check, in the same vein as \cref{uvexclemma}, that in fact suffices to control only one ordering of the above derivatives in order to obtain the listed order of decay for any rearrangement of them:
\begin{proposition}\label{reorder1}
Suppose the hypotheses of \cref{ind2} hold for $\abs{\alpha} \leq n, \abs{\beta} \leq n+1$, for some ordering of derivatives $\partv{},\partu{}$, then in fact the same estimates hold for arbitrary reorderings of $\partu{},\partv{}$.
\end{proposition}

\begin{proof}
The proof is identical to that given for \cref{uvexclemma} above and is not repeated.
\end{proof}

Observe that the hypothesis above is immediately satisfied by \cref{oorder1r}. Thus proving \cref{ind2} immediately gives us \cref{Obounds}.

\begin{proof}[Proof of ~\cref{ind2}]
We begin with \eqref{ind2hypuvphi}. In this case, we write
\begin{equation*}\begin{split} &\abs{\partu{\alpha_u}\partv{\alpha_v}(r\phi)} = \abs{r\partu{\alpha_u}\partv{\alpha_v}\phi}\\&+\abs{\vartheta(\alpha_v-1)\alpha_v\partu{\alpha_u}\partv{\alpha_v-1}\phi} + \abs{\vartheta(\alpha_u-1)\alpha_u\partu{\alpha_u-1}\partv{\alpha_v}\phi}, \end{split}\end{equation*}
where $\vartheta(x)$ is a step function, $0$ for $x < 0$, $1$ for $x \geq 0$. Thus we obtain the desired bound,
\begin{dmath*} \abs{\partu{\alpha_u}\partv{\alpha_v}\phi} \lesssim r^{-1}(1+\abs{u})^{-\alpha_u}v^{-\alpha_v}, \end{dmath*}
by simply multiplying our hypothesized bounds at these orders, and using the fact that $u,v \lesssim r$ in $\mathcal{O}$.

Next, we move to \eqref{ind2hypuvrphi}. By \cref{reorder1} it suffices to check the following case:
\begin{dmath*} \abs{\partu{\beta_u-1}\partv{\beta_v-1}\partu{}\partv{}(r\phi)} = \abs{\partu{\beta_u-1}\partv{\beta_v-1}\Lx \frac{m}{(1-\mu)r^2} \Lx \phi - \partv{}(r\phi) \Rx \Rx}. \end{dmath*}
By hypothesis this is bounded by
\begin{dmath*} \abs{\partu{\beta_u-1}\partv{\beta_v-1}\partu{}\partv{}(r\phi)} \lesssim r^{-2}(1+\abs{u})^{-(\beta_u-1)}v^{-\beta_v} \leq r^{-1}(1+\abs{u})^{-\beta_u}v^{-\beta_v}, \end{dmath*}
so we have our required decay.

The bounds \eqref{ind2hypulambda}, and \eqref{ind2hypvnu} follow similarly directly from \eqref{SSESF} and our hypothesis.

Now for \eqref{ind2hypm} we consider the terms
\begin{gather*} \partu{\alpha_u}\partv{\alpha_v}m = \partu{\alpha_u}\partv{\alpha_v-1} \Lx \frac{1}{2}(1-\mu)r^2(\partv{}\phi)^2 \Rx\\ \partu{\alpha_u-1}\partv{\alpha_v}\partu{}m = \partu{\alpha_u-1}\partv{\alpha_v}\Lx \frac{1}{2}(1-\mu)r^2(\partu{}\phi)^2 \Rx. \end{gather*}
In each case, the term with minimal decay is when all derivatives act on some copy of $\partv{}\phi$ or $\partu{}\phi$ respectively, and by \eqref{ind2hypuvphi} this is bounded by
\begin{dmath*} \abs{\partu{\alpha_u}\partv{\alpha_v}m} \lesssim (1+\abs{u})^{-\alpha_u}v^{-(\alpha_v+1)}, \end{dmath*}
and
\begin{dmath*} \abs{\partu{\alpha_u-1}\partv{\alpha_v}\partu{}m} \lesssim (1+\abs{u})^{-(\alpha_u+1)}v^{-\alpha_v}. \end{dmath*}
But by \cref{reorder1} the difference between these two terms is of order
\begin{dmath*} (1+\abs{u})^{-(\alpha_u+1)}v^{-(\alpha_v+1)}, \end{dmath*}
so in fact each satisfies both bounds, and we can safely take the minimal value of these two bounds. Thus we obtain \eqref{ind2hypm}. From this, \eqref{ind2hypmr} follows immediately by splitting the derivatives over $m$ and $r^{-k}$.

Next we can proceed to \eqref{ind2hypvlambda}, and \eqref{ind2hypunu}. The approach and bounds are completely symmetric by exchanging $u$ for $v$, so we will only prove \eqref{ind2hypvlambda} in detail. So we write:
\begin{dmath*} \partv{\abs{\alpha}}\lambda(u,v) = \int_1^u \partial_u\partv{\abs{\alpha}}\lambda(u',v) du'. \end{dmath*}
By the same procedure as used in the proof of \cref{reorder1}, this integrand can be bounded by
\begin{dmath*} \abs{\partial_u\partv{\abs{\alpha}}\lambda} \lesssim \abs{\partv{\abs{\alpha}}\partu{}\lambda} + \sum_{i = 1}^{\abs{\alpha}} \abs{\partv{i-1}(\nu^{-1}\partv{}\nu\partu{} - \lambda^{-1}\partu{}\lambda \partv{}) \partv{\abs{\alpha}-i}}. \end{dmath*}
By our hypothesized bounds and the results above, all of the terms in the latter sum are bounded by $r^{-1}v^{-\abs{\alpha}}u^{-1}$, or by $r^{-2}\partv{\abs{\alpha}}\lambda$, and so can be controlled by Gronwall's inequality. In either case we have the necessary decay (since the former term integrates to $v^{-(\abs{\alpha} + 1)}$), so we can safely ignore these. Thus all that remains is
\begin{dmath*} \partv{\abs{\alpha}}\partu{}\lambda = \partv{\abs{\alpha}}\frac{m\lambda}{(1-\mu)r^2}. \end{dmath*}
The term in which all derivatives act on $\lambda$ can be controlled by Gronwall, since $r^{-2}$ in globally integrable, and the remaining terms all satisfy decay like $r^{-2}v^{-\abs{\alpha}}$ with $\partv{\abs{\alpha}}r^{-2}$ being the term with the lowest power of decay. Thus, integrating, we obtain the bound
\begin{dmath*} \abs{\partv{\abs{\alpha}}\lambda} \lesssim v^{-(\abs{\alpha}+1)}, \end{dmath*}
as desired.

Finally we move to \eqref{ind2hypvrphi}, \eqref{ind2hypurphi}. As with \eqref{ind2hypvlambda}, and \eqref{ind2hypunu} above, the proof of \eqref{ind2hypurphi} is the same as that for \eqref{ind2hypvrphi} with the roles of $u$ and $v$ interchanged, and so is not repeated. Thus, we have
\begin{dmath*} \partv{\abs{\beta}}(r\phi)(u,v) = \partv{\abs{\beta}}(r\phi)(1,v) + \int_1^u \partial_u\partv{\abs{\beta}}(r\phi)(u',v)du'. \end{dmath*}
By our constraints on the data, $\partv{\abs{\beta}}(r\phi)(1,v)$ is already good enough so we need to control the integral term. We can bound this integrand by
\begin{equation*}\begin{split} &\abs{\partial_u\partv{\abs{\beta}}(r\phi)}\\&\lesssim \abs{\partv{\abs{\beta}-1}\Lx\frac{m\nu}{(1-\mu)r^2}\Lx \phi - (\partv{}(r\phi)) \Rx\Rx} + r^{-1}(1+\abs{u})^{-1}v^{-(\abs{\beta}+1)}. \end{split}\end{equation*}
This second term already satisfies the necessary decay, so we are concerned only with the first. By our hypothesis and the bounds already checked above we have that (expanding) each term in this derivative verifies decay bounded by at least
\begin{dmath*} r^{-2}v^{-\abs{\beta}}. \end{dmath*}
Thus our integral is bounded overall by $v^{-(\abs{\beta} + 1)}$, so we have
\begin{dmath*} \abs{\partv{\abs\beta}(r\phi)}(u,v) \lesssim v^{-(\abs{\beta} + 1)}. \end{dmath*}
The $u$ case is similar.
\end{proof}

This completes the proof of \cref{ind2} and thus \cref{Obounds}. In particular we have our first main result: \cref{main1}.

\section{Stability to Non-Symmetric Perturbations}
\label{stabproof}
It remains now to prove \cref{main2}. We have established control in spherical symmetry through the proof of \cref{main1} above, but it remains to check that the lifts of these solutions to $(3+1)$ dimensional solutions to \eqref{ESF} are in fact dispersive of order $(C,\gamma_0,N)$ for $\gamma_0,N$ sufficiently large. To do this we must $(1)$ check that these bounds transfer nicely to the full $(3+1)$ dimensional solution, and $(2)$ construct a suitable gauge and coordinate system.

In this section, we address this second issue. In particular we construct a gauge and coordinate system and check that in this setting (given a resolution for the first concern above) the conditions for our solution $(\mathcal{M},\textbf{g},\tilde\phi)$ to be dispersive of order $(C,\gamma_0,N)$ hold for the solutions considered in the hypotheses of \cref{main2}.

\subsection{Coordinates and Gauge}
We must construct a set of coordinates and prescribe a gauge for our solutions to \eqref{ESF} on $\mathcal{M}$ before we can check the conditions of \ref{dispsoln}. Since we take our solution to project to a solution of \eqref{SSESF} of the type considered in \cref{main1} we have immediately that our solution is spherically symmetric, and admits a double-null-ruling by coordinates $(\hat{u},\hat{v},\theta,\phi)$ in which the metric takes the form
\[ \textbf{g} = -\Omega^2d\hat{u}d\hat{v} + r^2(\hat{u},\hat{v})ds^2_{S^2} \]
as in \cref{secprelim}. Recall also that these coordinates present our solution as a lift of the solution to \eqref{SSESF} via the projection $(\hat{u},\hat{v},\theta,\varphi) \mapsto (\hat{u},\hat{v})$. Thus as in the $(1+1)$ dimensional reduced case, these $\hat{u},\hat{v}$ are free up to a choice of gauge. In what follows we will impose the gauge condition \ref{gauge2}.

Then define $t = \hat{u} + \hat{v}$ and $\hat{r} = \hat{v} - \hat{u}$. From here we obtain a coordinate system $(t,x^1,x^2,x^3)$ defined as
\[ t = \hat{u} + \hat{v} \quad x^1 = \hat{r}\cos\theta\sin\varphi \quad x^2 = \hat{r}\sin\theta\sin\varphi \quad x^3 = \hat{r}\cos\varphi. \]
The remainder of this section will be devoted to showing that lifts of the solutions considered in \cref{main1} represented in this coordinate system and gauge satisfy the conditions of \ref{dispsoln}. We do this in two parts: First we check that changing from \ref{gauge1} to \ref{gauge2} preserves the decay properties shown in \cref{main1}. Then we check the remaining conditions of \ref{dispsoln} using the results of \cref{axisregsec}.

\subsection{Changing Gauge}
Here we check that changing gauge from \ref{gauge1} to \ref{gauge2} (at least) preserves the decay found in \cref{main1}. Observe that it suffices to do so in the $(1+1)$ dimensional setting, since the null coordinates here induce an equivalent choice of null coordinates on $\mathcal{M}$ by construction.

The result is the following:
\begin{lemma}\label{gaugeexc}
Let $(r,\phi,m)$ a solution to \eqref{SSESF} verifying hypothesis of \cref{main1} in \ref{gauge1}. Then the estimates \cref{main11,main12,main13,main14,main15,main16,main17,main18,main19,main110} also hold exchanging the $u,v$ of \ref{gauge1} for $\tilde{u},\tilde{v}$ null coordinates for the gauge \ref{gauge2}.
\end{lemma}

\begin{proof}
Recall that we obtain the coordinates $\tilde{u},\tilde{v}$ from $u,v$ by the transformation
\[ \tilde{u}(u,v) = -\int_1^{u} 2\bar\nu(u')du' \qquad \tilde{v}(u,v) = -2\int_1^v\bar\nu(v')dv', \]
where $\bar\nu(x) = \lim_{v \rightarrow \infty}\nu(x,v)$. It follows that
\[ \partial_{\tilde{u}} = \frac{\partial u}{\partial \tilde{u}}\partial_{u} \qquad \partial_{\tilde{v}} = \frac{\partial v}{\partial\tilde{v}}\partial_v. \]
By the above
\[ \frac{\partial u}{\partial \tilde{u}} = -\frac{1}{2\bar\nu(u)} \qquad \frac{\partial v}{\partial \tilde{v}} = -\frac{1}{2\bar\nu(v)}. \]
Then by the bounds \cref{main11,main13}, it follows immediately from \cref{main1} that any solution satisfying the hypotheses of \cref{main1} in \ref{gauge1} verifies the same decay estimates in \ref{gauge2} as well.
\end{proof}

\begin{remark}
Essentially the same argument allows us to pass between \ref{gauge3} and \ref{gauge1} without concern as well, since our transformation is nearly identical.
\end{remark}

Note that in \ref{gauge2} we can augment the bounds of \cref{main1} slightly, as we obtain some new control of $\partial_u^l \nu$. In particular we have the following:
\begin{lemma}\label{mainimprove}
Let $(r,\phi,m)$ be as in \cref{main1}, but presented in the gauge \ref{gauge2}. Then we have
\[ \abs{\partu{l}\nu} \lesssim \min\set*{(1+\abs{u})^{-(l+1)},(1+\abs{u})^{-l}v^{-1}}. \]
\end{lemma}

\begin{proof}
In \ref{gauge2} we have $\lim_{v \rightarrow \infty} \nu(u,v) \equiv -\frac{1}{2}$. In particular we have $\lim_{v\rightarrow \infty}\partial_u^l \nu(u,v) = 0$ for any $u$, $l \leq k$ (since the convergence to $-\frac{1}{2}$ is uniform by construction). Thus we can write
\[ \partu{l}\nu(u,v) = -\int_v^{\infty}\partial_v\partu{l}\nu(u,v') dv'. \]
The desired bound then follows immediately from \eqref{main14}.
\end{proof}

Finally, \ref{gauge2} gives us control of the limiting values of $\lambda$ at $\mathcal{I}^-$:
\begin{proposition}\label{lambdalimit}
In \ref{gauge2} we have $\lim_{v \rightarrow \infty}\lambda(u,v) = \frac{1}{2}$.
\end{proposition}

\begin{proof}
We have by \ref{gauge2} that $\lambda(u,u) = -\nu(u,u)$ for all $u \geq 1$. Moreover, by \cref{main1} we have that $\abs{\partial_u\lambda} \lesssim \min\set*{(1+\abs{u})^{-2},(1+\abs{u})^{-1}(1+v)^{-1}}$.

Integrating the intermediate bound $(1+\abs{u})^{-1}(1+v)^{-1}$ we conclude
\[ \abs{\lambda(u,v)} \lesssim -\nu(v,v) + \abs{\frac{\log(1+\abs{v})}{1+v}}. \]
Then we have $(\partial_u + \partial_v) \nu \lesssim (1+ \abs{u})^{-2}$, and $\lim_{t \rightarrow \infty}\nu(t,\hat{r}) = -\frac{1}{2}$ for any $\hat{r}$. Thus
\[ \abs{\nu(u,u) + \frac{1}{2}} \lesssim \frac{1}{1+\abs{u}}, \]
so we conclude that
\[ \abs{\lambda(u,v) - \frac{1}{2}} \lesssim \frac{1}{1+\abs{v}} + \frac{\log(1+\abs{v})}{1+v}, \]
and thus our limit holds.
\end{proof}

From this bound we have immediately the following corollary:

\begin{corollary}\label{sumlimitcor}
\[ \abs{\lambda+\nu} \lesssim \frac{\log(2+\abs{v})}{1+v}. \]
\end{corollary}

\subsection{Checking Dispersiveness}
Finally we are ready to check the conditions of \ref{dispsoln}. Note that conditions \ref{disp1},\ref{disp7} follow immediately from \cref{main1} and our choice of coordinates. It remains to carefully check the remaining conditions of \ref{dispsoln}. We will do this in two parts, first when $\abs{I} = 0$, then addressing separately the case $\Gamma^I$ acts on the term of interest.
\begin{proposition}
The bounds \ref{disp2}--\ref{disp8} hold for $\abs{I} = 0$.
\end{proposition}

\begin{proof}
We begin with the components of the metric $h_B$. Recall that in null coordinates the metric has the form:
\[ \frac{-\Omega^2}{2}(du\otimes dv + dv\otimes du) + r^2d\gamma_{S^2}. \]
Thus in our $(t,\mathbf{x})$ coordinates described above the metric has the following components:
\[ \mathbf{g}_{tt} = -\Omega^2 \qquad \mathbf{g}_{it} = 0 \qquad \mathbf{g}_{ij} = \delta_{ij}\frac{r^2}{\hat{r}^2} + \frac{1}{\hat{r}^2}\Lx \Omega^2 - \frac{r^2}{\hat{r}^2} \Rx x_ix_j, \]
where $\Omega^2 = \frac{-4\lambda\nu}{1-\mu}$, and $\delta_{ij}$ is the Kronecker $\delta$.

So the components of the background-subtracted metric $h_B$ are:
\[ (h_B)_{tt} = -(\Omega^2-1) \qquad (h_B)_{it} = 0 \qquad (h_B)_{ij} = \delta_{ij}\Lx\frac{r^2}{\hat{r}^2}-1\Rx + \frac{1}{\hat{r}^2}\Lx \Omega^2 - \frac{r^2}{\hat{r}^2} \Rx x_ix_j. \]
We'll begin with the necessary estimates near the axis ($\hat{r} \leq 1$).

There are three terms we must control:
\begin{gather}
	\label{hb1} \Omega^2-1,\\
	\label{hb2} \frac{r^2}{\hat{r}^2}-1,\\
	\label{hb3} \frac{1}{\hat{r}^2}\Lx \Omega^2 - \frac{r^2}{\hat{r}^2} \Rx.\\
\end{gather}

We start with \eqref{hb1}. Recall that $\Omega^2 = \frac{-4\lambda\nu}{1-\mu}$. By \eqref{mapriori} we can write this as
\[ \Omega^2 = -4\lambda\nu \sum_{n = 0}^{\infty} \mu^{n}, \]
and thus
\[ \Omega^2 - 1 = (-4\lambda\nu - 1) - 4\lambda\nu\sum_{n = 1}^{\infty} \mu^n. \]
Then by \eqref{mapriori} $\abs{\mu^n} \lesssim r^{3n-1}(1+\abs{u})^{-6n}$, so this latter term immediately satisfies the required decay for $r \leq 1$, since $\lambda\nu$ is bounded, and $t \leq 2(\abs{u}+1)$ in this region. It thus remains to control $-4\lambda\nu - 1$. Observe that by our construction of the gauge \ref{gauge2} we have that
\[ \lim_{t \rightarrow \infty} \lambda(t,0) = -\lim_{t\rightarrow \infty} \nu(t,0) = \frac{1}{2}. \]
Thus $\lim_{t\rightarrow \infty}\lambda\nu(t,0) = -\frac{1}{4}$. So we can write
\[ (-4\lambda\nu - 1)(t,0) = 4\int_{t}^{\infty} \partial_t(\lambda\nu) dt. \]
Observe that since we are near $\Gamma$ we are WLOG in the region $\mathcal{Q}$ upon projection. Thus by the results of \cref{Qproof} we have that $\partial_t(\lambda\nu) \lesssim (1+\abs{u})^{-3} \lesssim t^{-3}$. We thus obtain
\[ \abs{-4\lambda\nu - 1}(t,0) \lesssim \frac{1}{(1+v)}, \]
since $t \sim v$ in the region $\hat{r} \leq 1$. Moreover, by the results of \cref{axisregsec} and \cref{Qproof} we have that
\[ \abs{\partial_i\lambda\nu} \lesssim (1+\abs{u})^{-3}, \]
so we conclude that for all $\abs{x} \leq 1$,
\[ \abs{-4\lambda\nu - 1}(t,x) \lesssim \frac{1}{(1+v)^2} \]
as well. Thus we have the necessary control of \eqref{hb1}.

Next we consider the term \eqref{hb2}. Observe that we can write (reducing to 2-dimensions by spherical symmetry):
\[ r(t,x) = \int_0^{\abs{x}} (\lambda - \nu)(t,\hat{r}) d\hat{r}. \]
Let $\lambda_0(t,x) = \lambda(t,0)$, and $\nu_0(t,x) = \nu(t,0)$. Then we can rewrite the above as:
\[ r(t,x) = (\lambda_0 - \nu_0)\hat{r} + \int_0^{\abs{x}} (\lambda - \nu)(t,\hat{r}) - (\lambda_0 - \nu_0)(t,\hat{r})d\hat{r}, \]
and so we have:
\begin{dmath*} \frac{r^2}{\hat{r}^2}(t,x) = (\lambda_0 - \nu_0)^2 + 2(\lambda_0 - \nu_0)\frac{1}{\hat{r}}\int_0^{\abs{x}} (\lambda - \nu)(t,\hat{r}) - (\lambda_0 - \nu_0)(t,\hat{r})d\hat{r} + \Lx \frac{1}{r}\int_0^{\abs{x}} (\lambda - \nu)(t,\hat{r}) - (\lambda_0 - \nu_0)(t,\hat{r})d\hat{r} \Rx^2. \end{dmath*}
Thus similar to the above we must control the terms $(\lambda_0 - \nu_0)^2 - 1$ and $\frac{1}{\hat{r}}\int_0^{\abs{x}} (\lambda - \nu)(t,\hat{r}) - (\lambda_0 - \nu_0)(t,\hat{r})d\hat{r}$.

We begin with the former. By the above, we have that $\lim_{t\rightarrow \infty} \lambda_0 = lim_{t\rightarrow \infty} -\nu_0 = \frac{1}{2}$. Thus, we have
\[ \abs{(\lambda_0 -\nu_0)^2 - 1}(t,x) = \int_{t}^{\infty} \partial_t(\lambda - \nu)(t',0) dt'. \]
By \cref{Qproof} as above $\abs{\partial_t(\lambda - \nu)(t,0)} \lesssim (1 + \abs{u})^{-3}$, and thus
\[ \abs{(\lambda_0 -\nu_0)^2 - 1}(t,x) \lesssim \frac{1}{(1+v)^2} \]
as required. We can thus move to our other term.

In this case we have
\[ \abs{\frac{1}{\hat{r}}\int_0^{\abs{x}} (\lambda - \nu)(t,\hat{r}) - (\lambda_0 - \nu_0)(t,\hat{r})d\hat{r}} = \abs{\frac{1}{\hat{r}}\int_0^{\abs{x}}\int_0^{\hat{r}} \partial_{\hat{r}}(\lambda - \nu)(t,\hat{r}')d\hat{r}'d\hat{r}} \]
by \cref{main1}. By arguments of \cref{axisregsec} we in fact have that $\partial_{\hat{r}}(\lambda - \nu)\mid_{\Gamma} = 0$, thus we can write
\[ \partial_{\hat{r}}(\lambda-\nu)(t,x) = \int_0^{\abs{x}} \partial_{\hat{r}}^2(\lambda - \nu)(t,\hat{r}) d\hat{r}. \]
Substituting this in above, and using the bounds of \cref{main1} we have
\[ \abs{\frac{1}{\hat{r}}\int_0^{\abs{x}} (\lambda - \nu)(t,\hat{r}) - (\lambda_0 - \nu_0)(t,\hat{r})d\hat{r}} \lesssim \hat{r}^2(1+v)^{-3} \]
satisfying the required bound for $\abs{x} \leq 1$. This bounds \eqref{hb2}.

Finally, we move to \eqref{hb3}. We start with just $\Omega^2 - \frac{r^2}{\hat{r}^2}$. Now we can write:
\begin{dmath*} \Lx\Omega^2 - \frac{r^2}{\hat{r}^2}\Rx = -4\lambda_0\nu_0 - 4\int_0^{\abs{x}} \partial_{\hat{r}}(\lambda\nu)(t,\hat{r}) d\hat{r} - 4\lambda\nu\sum_{n = 1}^{\infty} \mu^n - (\lambda_0 - \nu_0)^2 - 2(\lambda_0 - \nu_0)\frac{1}{\hat{r}}\int_0^{\abs{x}} (\lambda - \nu)(t,\hat{r}) - (\lambda_0 - \nu_0)(t,\hat{r})d\hat{r} - \Lx \frac{1}{r}\int_0^{\abs{x}} (\lambda - \nu)(t,\hat{r}) - (\lambda_0 - \nu_0)(t,\hat{r})d\hat{r} \Rx^2. \end{dmath*}
Note that since $\lambda = -\nu$ on $\Gamma$ we have that $-4\lambda_0\nu_0 = 4\lambda_0^2 = (2\lambda_0)^2 = (\lambda_0-\nu_0)^2$ so these constant terms vanish. By the above, it remains only to estimate
\[ \int_0^{\abs{x}}\partial_{\hat{r}}(\lambda\nu)(t,\hat{r})d\hat{r}. \]
By the arguments presented in \cref{axisregsec} it follows that
\[ \partial_{\hat{r}}(\lambda\nu)(t,\hat{r}) = \int_0^{\hat{r}} \partial_{\hat{r}}^2 (\lambda\nu)(t,\hat{r}')d\hat{r}', \]
so we have
\[ \abs{\int_0^{\abs{x}}\partial_{\hat{r}}(\lambda\nu)(t,\hat{r})d\hat{r}} \lesssim \hat{r}^2(1+v)^{-3}. \]

Putting this together, we see that
\[ \abs{\frac{1}{\hat{r}^2}\Lx \Omega^2 - \frac{r^2}{\hat{r}^2} \Rx} \lesssim (1+v)^{-2},  \]
and thus satisfies the required bound.

This establishes the condition \ref{disp2} close to the axis, for $\abs{I} = 0$. In fact, with the strength of the estimates above, this also establishes \ref{disp5} under the same extra conditions. It remains to address \ref{disp3} and \ref{disp4}.

In this case, since $u \sim v$ in the finite $\hat{r}$ region we can address each term at the same time. As above, there are three distinct terms to deal with:
\begin{gather}
\Omega^2 - 1,\\
\frac{r^2}{\hat{r}^2} - 1,\\
\label{hb32} \frac{1}{\hat{r}^2}\Lx \Omega^2 - \frac{r^2}{\hat{r}^2} \Rx x_ix_j.
\end{gather}
It suffices to control $\partial_{\alpha}(h_B)_{\mu\nu}$ (as the angular terms in $\bar\partial$ can only act on the $x_ix_j$ which does not affect the decay at all).
In the first two cases, the result follows directly from the analysis of \cref{XboundCor}. In the case \eqref{hb32} this instead follows from \cref{Xboundavgcor}, and the above work to show \ref{disp2}.

With this completed, it remains to establish the required estimates in the region $\hat{r} > 1$. In this region we need not worry about potential singularities in our metric components, since everything is uniformally bounded and smooth away from the axis. Once more we have three types of terms that we must bound:
\begin{gather}
\Omega^2 - 1,\\
\frac{r^2}{\hat{r}^2} - 1,\\
\Lx \Omega^2 - \frac{r^2}{\hat{r}^2} \Rx x_ix_j = \Lx \Omega^2 - \frac{r^2}{\hat{r}^2} \Rx S_i(\theta,\phi)S_j(\theta,\phi),
\end{gather}
where $S_j(\theta,\phi) = \frac{x_i}{\hat{r}}$ is the completely angular part of $x_i$ written in $(t,\hat{r},\theta,\phi)$ coordinates. Note that
\[ \partial_iS_j = \frac{\delta_{ij}}{\hat{r}} - \frac{S_iS_j}{\hat{r}} \sim \frac{1}{\hat{r}}. \]

We first establish \ref{disp2} for each of these terms then move on to each of \ref{disp3}--\ref{disp5}

We'll begin with $\Omega^2 - 1$. Here we have
\[ \Omega^2(u,v) - 1 = \frac{-4\lambda\nu}{1-\mu} - 1. \]
By \cref{lambdalimit}, and \cref{main1} we can write
\[ \lambda(u,v) = \frac{1}{2} + E_{\lambda}(u,v) \qquad \nu(u,v) = -\frac{1}{2} + E_{\nu}(u,v), \]
where $\abs{E_\lambda} \lesssim \frac{1}{1+v} + \frac{\log(1+\abs{u})}{1+v}$, and $\abs{E_\nu} \lesssim \frac{1}{1+v}$. Thus we can write:
\[ \Omega^2(u,v) = \Lx 1 + E_{\lambda} + E_{\nu} + E_{\lambda}E_{\nu} \Rx \sum_{n = 0}^{\infty} \mu^n. \]
By \cref{main1}, terms with $n \geq 1$ immediately satisfy the bounds of \ref{disp2} since the prefactor is uniformly bounded. The remaining term is
\[ 1 + E_{\lambda} + E_{\nu} + E_{\lambda}E_{\nu} - 1 = E_{\lambda} + E_{\nu} + E_{\lambda}E_{\nu}, \]
but then as above, each of these also satisfy the bounds of \ref{disp2}, since $\abs{u} \lesssim \abs{v}$ in $\mathbb{I}$.

Next we move on to $\frac{r^2}{\hat{r}^2} - 1$, so we begin by showing that this, in fact, vanishes for large $\hat{r}$. Note that we have
\[ r(t,\hat{r}) = \int_0^{\hat{r}} \lambda - \nu d\hat{r}'. \]
Then as above we can write
\[ \lambda - \nu = 1 + E_{\lambda} - E_{\nu}, \]
so we have
\[ r - \hat{r} = \int_0^{\hat{r}} E_\lambda - E_\nu d\hat{r}'. \]
By \cref{main1} and the arguments of \cref{lambdalimit}, we have that
\[ \abs{E_{\lambda}} \lesssim \max\set*{(1+\abs{u})^{-1},(1+v)^{-1}} \qquad \abs{E_{\nu}} \lesssim (1+v)^{-1}, \]
and moreover, in $\mathbb{I}$ we have $\hat{r} \leq v, \abs{u} \leq \hat{r}$, so this is bounded by
\[ \abs{r - \hat{r}} \lesssim \log(2+\hat{r}), \]
so we can write
\[ r = \hat{r} + E_{r} \]
where $\abs{E_r} \lesssim \log(2+\hat{r})$. Thus we conclude:
\[ \abs{\frac{r^2}{\hat{r}^2} - 1} \lesssim \frac{\log(2+\hat{r})}{\hat{r}}. \]
Above, we already showed that this term is bounded by $\frac{\hat{r}^2}{(1+v)^3}$, so, applying this new bound in the region $\mathbb{I} \setminus \mathcal{Q}$ we have the required control of this term.

Finally, we consider $\Lx \Omega^2 - \frac{r^2}{\hat{r}^2} \Rx S_iS_j$, but this is immediately controlled by combining our work above for $\Omega^2 - 1$ and $\frac{r^2}{\hat{r}^2} - 1$, and the boundedness of the $S_i$'s.

Now we address \ref{disp3}. Thus we must establish the bound
\[ \abs{\partial_t h_{\mu\nu}} + \abs{\sum \partial_ih_{\mu\nu}} \lesssim \frac{1}{(1+v)(1+\abs{u})^{\gamma_0}} \]
for some $\gamma_0 > 0$. In fact we obtain this for any $\gamma_0 < 1$. In the case of $\partial_t$ all of the required bounds hold immediately by \cref{main1}, and \cref{sumlimitcor}. For $\partial_i$, if we allow this to act on a $\Omega^2$ then this again immediately satisfies our bound by \cref{main1}. If instead we allow this to act on some $S_iS_j$ term, then this simply multiplies $\Lx \Omega^2 - \frac{r^2}{\hat{r}^2} \Rx$
by a term proportional to $\frac{1}{\hat{r}}$. By our above work to prove \ref{disp2} this gives us the required bound as well.

Finally, we must deal with the case $\partial_i\frac{r^2}{\hat{r}^2}$. In this case we can expand:
\[ \partial_i\frac{r^2}{\hat{r}^2} = \frac{2rS_i(\lambda - \nu)}{\hat{r}^2} - \frac{2S_ir^2}{\hat{r}^3} = 2\frac{S_i}{\hat{r}}\Lx \frac{\hat{r}r(\lambda - \nu) - r^2}{\hat{r}^2} \Rx. \]
It thus suffices to establish an estimate for $\hat{r}(\lambda - \nu) - r$. Recalling what we've done above we can write this as
\[ \hat{r}(\lambda - \nu) - r = \hat{r} + (E_{\lambda} - E_{\nu})\hat{r} - \hat{r} + \int_0^{\hat{r}} E_{\lambda} - E_{\nu}d\hat{r}', \]
so we are left with:
\[ (E_{\lambda} - E_{\nu})\hat{r} + \int_0^{\hat{r}} E_{\lambda} - E_{\nu}d\hat{r}'. \]
Thus the overall term can be written as
\[ 2\frac{S_i}{\hat{r}}\frac{r}{\hat{r}}\Lx (E_{\lambda} - E_{\nu}) + \frac{1}{\hat{r}}\int_0^{\hat{r}} E_{\lambda} - E_{\nu}d\hat{r}' \Rx. \]
By the results of \cref{oproof} in the region $\mathcal{Q}$ $E_{\lambda}$ and $E_{\nu}$ are controlled by $\frac{1}{t^2}$, so our term is immediately bounded by
\[ \frac{1}{t^2\hat{r}}. \]
On the other hand, using the unversal bounds of \cref{main1} as we did above, we obtain the bound
\[ \frac{\log(2+\hat{r})}{\hat{r}^2}. \]
Together, these give us the required bound in the region $\hat{r} > 1$, so we conclude that \ref{disp3} holds.

Next we consider \ref{disp4}. Here we must bound $\abs{\partial_v h} + \abs{\slashed{\nabla} h}$, we will treat these terms separately, beginning with $\partial_v h = (\partial_t + \partial_r)h$. We have by \cref{main1} that $\partial_v \Omega^2$ immediately satisfies the required estimates. In the case of $\frac{r^2}{\hat{r}^2}$ we have
\[ \partial_v\frac{r^2}{\hat{r^2}} = \frac{2r\hat{r}\lambda - r^2}{\hat{r}^3}. \]
Now similar to the above we must bound
\[ 2\hat{r}\lambda - r = \hat{r} + 2\hat{r}E_{\lambda} - \hat{r} + \int_{0}^{\hat{r}} (E_{\lambda} - E_{\nu})d\hat{r}',  \]
and thus the resulting term is
\[ \frac{r}{\hat{r}}\Lx \frac{2}{\hat{r}}E_{\lambda} - \frac{1}{\hat{r}^2}\int_{0}^{\hat{r}} (E_{\lambda} - E_{\nu})d\hat{r}' \Rx. \]
As above, this is bounded by $\frac{1}{t^2\hat{r}}$, and $\frac{\log(2+\hat{r})}{\hat{r}^2}$ which gives us the overall required bound.

Finally, $\partial_v$ acts as 0 on the $S_i$ as these are totally angular.

Thus we are left to deal with $\slashed{\nabla} h$. $\slashed{\nabla}$ acts non-trivially only on non-spherically symmetric terms, so this can only affect the terms $x_ix_j$ in our metric. On these terms the operator acts as:
\[ \slashed{\nabla} x_i = S_j - S_k \]
for $j,k \neq i$. In particular, this serves to multiply by an additional factor of $\frac{1}{\hat{r}}$. Thus we can apply our work above for \ref{disp2}, and immediately conclude \ref{disp4}.

Finally, we must deal with \ref{disp5}. All terms without a $\Gamma$ are identically 0 in this case thanks to simple algebraic cancellations.

Now we move to $\phi$ and \ref{disp6}. Again we are intersted only in the case where $\abs{I} = 0$. We begin near the axis. For both $\abs{\partial \phi}$ and $\abs{\bar\partial\phi}$, the required bounds of \ref{disp6} hold immediately by \cref{XboundCor}, since $u \sim v$. Away from the axis, we can directly apply the results of \cref{main1}. In particular, $\abs{\partial \phi} \lesssim \abs{\partial_u \phi}$, and $\abs{\bar{\partial \phi}} = \abs{\partial_v \phi}$, since the $\phi$ is spherically symmetric, so the angular portion of $\bar\partial$ vanishes. Thus the desired estimates follow directly from \cref{main1}.

Finally, we must deal with \ref{disp8}. We must begin by computing $\Box_{\textbf{g}}$ in the first place. Using our computed metric components above, we see through a bit of algebra that
\[ \sqrt{-\det \textbf{g}} = \Omega^2\frac{r^2}{\hat{r}^2} \]
We can also compute the components of the inverse metric $(g^{-1})^{\alpha\beta}$:
\begin{gather*}
(\textbf{g}^{-1})^{00} = -\Omega^{-2},\\
(\textbf{g}^{-1})^{0i} = 0,\\
(\textbf{g}^{-1})^{ij} = \delta_{ij}\frac{\hat{r}^2}{r^2} + \frac{1}{\hat{r}^2}\Lx \Omega^{-2} - \frac{\hat{r}^2}{r^2} \Rx x_ix_j.
\end{gather*}
Thus we have
\[ \Box_{\textbf{g}} = \Omega^{-2}\frac{\hat{r}^2}{r^2}\partial_{\alpha}\Lx (g^{-1})^{\alpha\beta} \Omega^{2}\frac{r^2}{\hat{r}^2} \partial_{\beta} \Rx, \]
and there are four terms we must address:
\begin{gather}
\label{wavet0} \Box_{\textbf{g}} t = -\Omega^{-2}\frac{\hat{r}^2}{r^2}\partial_t \frac{r^2}{\hat{r}^2},\\
\label{wavexi} \Box_{\textbf{g}} x^j = \Omega^{-2}\frac{\hat{r}^2}{r^2}\partial_i\Lx\delta_{ij}\Omega^2 + \frac{1}{\hat{r}^2}\Lx \frac{r^2}{\hat{r}^2} - \Omega^2 \Rx x_ix_j\Rx,
\end{gather}
for $j = 1,2,3$. As above, we separate the near axis and large $\hat{r}$ estimates.

We begin with \eqref{wavet0}, near the axis ($\hat{r} < t/2$). We can rewrite this as:
\[ \Box_{\textbf{g}} t = -\Omega^{-2}\frac{r(\lambda + \nu)}{r^2} = -\Omega^{-2}\frac{\lambda+\nu}{r}. \]
$\Omega^{-2}$ is uniformly bounded, so we we must control $\frac{\lambda+\nu}{r}$. By \ref{gauge2} we have that $\lambda + \nu\mid_{\Gamma} = 0$, we can write by our averaging operator:
\[ \frac{\lambda + \nu}{r}(t,\hat{r}) = \frac{\hat{r}}{r}\frac{1}{\hat{r}}\int_0^{\hat{r}} \partial_{\hat{r}}(\lambda + \nu)(t,\hat{r}')d\hat{r}'. \]
We show in \cref{axisregsec} that $\frac{\hat{r}}{r}$ is bounded, so it suffices to control the averaged term. Then by \cref{main1} we have
\[ \abs{\partial_{\hat{r}}(\lambda+\nu)} \lesssim (1 + \abs{u})^{-2}. \]
Thus, integrating, $\Box_{\textbf{g}}t$ satisfies the required bound.

Away from the axis ($\hat{r} \geq t/2$), we must control $\partial_t\frac{r^2}{\hat{r}^2}$, since once again the pre-factor is only uniformly bounded. This is:
\[ \frac{2r(\lambda + \nu)}{\hat{r}^2}. \]
By \cref{sumlimitcor} we have that $\abs{\lambda + \nu} \lesssim \frac{\log(2+\abs{u})}{(1+v)}$, so we immediately obtain an overall bound of
\[ \abs{\Box_{\textbf{g}}t} \lesssim \frac{\log(2+\abs{u})}{(1+v)\hat{r}}. \]
This is good enough outside of $\mathcal{Q}$, and moreover, using the better bounds which hold in $\mathcal{Q}$ by \cref{main1} in this region, we in fact have
\[ \abs{\Box_{\textbf{g}}t} \lesssim \frac{\log(2+\abs{u})}{(1+v)^2\hat{r}}, \]
which gives us the required bound everywhere.

Finally, we turn to \eqref{wavexi}, beginning near the axis. As before $\Omega^{-2}\frac{\hat{r}^2}{r^2}$ is uniformly bounded, so we are concerned only with the differential terms. In each case we are thus left with
\[ \partial_i \Omega^2 + \sum_{j = 1,2,3}x_ix_j\partial_i\Lx \frac{1}{\hat{r}^2}\Lx \frac{r^2}{\hat{r}^2} - \Omega^2 \Rx  \Rx + \frac{1}{\hat{r}^2}\Lx \frac{r^2}{\hat{r}^2} - \Omega^2 \Rx x_j. \]
We have already obtained sufficient control of $\frac{1}{\hat{r}^2}\Lx \frac{r^2}{\hat{r}^2} - \Omega^2 \Rx x_j$ above, so we need only be concerned with $\partial_i \Omega^2$ and $\partial_i\Lx \frac{1}{\hat{r}^2}\Lx \frac{r^2}{\hat{r}^2} - \Omega^2 \Rx \Rx x_ix_j$. Each of these quantities are adequately controlled by \cref{main1} and \cref{axisregsec}, so again we can move to the case away from the axis.

Observe that in fact
\[ \partial_i(\delta_{ij}\Omega^2 - \Omega^2S_iS_j) = -\Omega^2\partial_i(S_iS_j). \]
Thus we are left to control the terms:
\[ \Lx\partial_i \frac{r^2}{\hat{r}^2}\Rx S_iS_j = 2\frac{r\hat{r}(\lambda-\nu) - r^2}{\hat{r}^3}S_j, \]
and
\[ \Lx\frac{r^2}{\hat{r}^2} - \Omega^2\Rx\partial_i(S_iS_j). \]
But these two terms are already adequately bounded by our work to conclude \ref{disp2} and \ref{disp3} above, so there is nothing left to do.
\end{proof}

\begin{proposition}
The bounds \ref{disp2}--\ref{disp8} hold for $\abs{I} \leq k$.
\end{proposition}

\begin{proof}
The idea for each term, both near and far from the axis, is to be able to count the total number of derivatives acting, as well as the total powers of the accompanying weights, and then apply the results of \cref{main1} in order to obtain the required decay. To this end we will write the most general operator which may act on a given term, and then examine this count of derivatives and weights.

We begin with the components of the metric $h_B$ near the axis.

Here we have a general term:
\[ \partial^{A}(x^i\partial_j - x^j \partial_i)^B(t\partial_i + x^i\partial_t)^C (t \partial_t + \hat{r}\partial_{\hat{r}})^l (h_B)_{\mu\nu}, \]
where $A,B,C$ are multi-indices specifying the particular $\Gamma$ applied. We will work right to left to write the general term we must bound in a nicer way. To do this we determine the form of the general term of the above operator acting on a spherically symmetric function $f$. Once we have done this we will go back to modify our expression to account for the non-spherically symmetric parts of $h_B$.

So we begin with $S = t\partial_t + \hat{r}\partial_{\hat{r}}$. Observe that we can rewrite this in terms of the operator $\hat{X} := \frac{1}{\hat{r}}\partial_{\hat{r}}$:
\[ S = t\partial_t + \hat{r}^2\hat{X}. \]
Observe that $\hat{X}t = 0$, and $\hat{X}\hat{r}^2 = 2$, so $t\partial_t$ and $\hat{r}^2\hat{X}$ commute, and we can write:
\[ S^l g_{\mu\nu} = \sum_{n = 1}^l \binom{l}{n}(\hat{r}^2\hat{X})^n(t\partial_t)^{l-n}f. \]
We can write a generic term in this sum (up to a multiplicative constant) as:
\[ \hat{r}^{2(n-n_1)}t^{l-n-n_2}\hat{X}^{n-n_1}\partial_t^{l-n-n_2}f, \]
for $n_1 < n,n_2 < l-n$.

The next term that can act on our sum above is:
\[ (t\partial_i + x^i\partial_t)^C. \]
The multi-index $C$ need not be ordered since these all commute with each other (up to differential terms with no weight, thus not contributing negatively to our final power counting). Now note that $\partial_t \hat{r}^2 = 0, \partial_it = 0,\partial_i\hat{r}^2 = 2x^i$, and, since $\hat{X}^l\partial_t^sf$ is spherically symmetric for any $l,s$ we have that
\[ \partial_i\hat{X}^l\partial_t^sf = x^i\hat{X}^{l+1}\partial_t^sf. \]
Finally, $\partial_t$ and $\partial_i$ commute for all $i$, so we need not be concerned with ordering these either. The generic operator this contributes is then:
\[ t^{\abs{m} - p}\partial_t^{\abs{C} - \abs{m} - p}\prod_{i = 1,2,3} (x^i)^{C_i - m_i - q_i}\partial_i^{m_i-q_i}, \]
For $p < \min\set*{\abs{m},\abs{C}}$ the number of time derivatives acting on factors of $t$, and $q_i < \min\set*{m_i,C_i}$ similar. 
Finally, letting this act on our general term above, we have
\[ t^{l+\abs{m}-n-n_2 - p - a}\hat{r}^{2(n-n_1 - \abs{b})}\Lx\prod (x^i)^{C_i - 2q_i}\Rx \hat{X}^{\abs{m} + n - n_1 - \abs{q} - \abs{b}}\partial_t^{l + \abs{C} - \abs{m} - n - n_2 - p - a}f. \]
Now we have the term $(x^i\partial_j - x^j\partial_i)^B$. Observe that this operator is completely angular, and thus acts as 0 on $t,\hat{r}^2$ and derivatives of $f$, and only has the effect of exchanging a copy of $x^i$ for one of $x^j$ (possibly with a change of sign). We can ignore the sign here, since this is just a global multiplicative factor, so letting the multi-index $B = (B_{12} + B_{21},B_{23} + B_{32},B_{31} + B_{13})$ we have the general term:
\[ t^{l+\abs{m}-n-n_2 - p - a}\hat{r}^{2(n-n_1 - \abs{b})}\Lx\prod (x^i)^{C_i - 2q_i+\sum_{j\neq i}(B_{ji} - B_{ij})}\Rx \hat{X}^{\abs{m} + n - n_1 - \abs{q} - \abs{b}}\partial_t^{l + \abs{C} - \abs{m} - n - n_2 - p - a}f, \]
for $j \neq i$, and $B_{ij} \leq C_i + B_{ji} - 2q_i$.

Finally, we have our derivatives $\partial^{A}$ for a multi-index $A = (A_t,A_1,A_2,A_3)$. The resulting term letting $\partial^A$ act on the above expression is:
\begin{align*} &t^{l+\abs{m}-n-n_2 - p - a-N_t}\hat{r}^{2(n-n_1 - \abs{b}-\abs{K})}\Lx\prod (x^i)^{A_i + C_i - 2q_i - 2N_i + \sum_{j\neq i}(B_{ji} - B_{ij})}\Rx\\&\cdot\hat{X}^{\abs{m} + n - n_1 - \abs{q} - \abs{b} + \sum_{i=1,2,3} (A_i - N_i - K_i)}\partial_t^{l + \abs{C} - \abs{m} - n - n_2 - p - a + A_t - N_t}f.
\end{align*}
This gives the general term we will consider for our spherically symmetric functions near the axis.

In the case of the background-subtracted metric $h_B$ the only non-spherically symmetric terms we must consider are of the form
\[ \frac{1}{\hat{r}^2}\Lx \Omega^2 - \frac{r^2}{\hat{r}^2} \Rx x_ix_j \]
composed of a spherically symmetric term
\[ \frac{1}{\hat{r}^2}\Lx \Omega^2 - \frac{r^2}{\hat{r}^2} \Rx, \]
and some non-symmetric terms $x_ix_j$. $t$ derivatives act trivially on the $x_i$, if a $\hat{r}^2\hat{X}$ operator acts on the non-spherically symmetric part we have
\[ \hat{r}^2\hat{X}(x_ix_j) = 2x_ix_j, \]
and finally we have
\[ \partial_i x_j = \delta_{ij}. \]
In particular an indefinite number of $\hat{r}^2\hat{X}$ terms can be absorbed by the non-symmetric terms, and only at most two spatial directional derivatives (without weighting) can act on non-symmetric terms before they vanish, as with any spatial coordinate terms the angular derivatives can exchange an $x_i$ for and $x_j$. Thus our generic operator above may be modified by reducing the powers of $\hat{r}^2$ and $\hat{X}$ in proportion to one another (i.e. reducing $l$ and associated quantities) -- recalling the order in which these terms appear -- dropping exchanging pairs (i.e. decreasing the total order of $\abs{B}$) and dropping at most two isolated spatial directional derivatives (reducing the $A_i$ or the $m_i$ contributing to differential terms).

Thus the terms (acting on spherically symmetric parts) we must consider are:
\begin{align*} &t^{l+\abs{m}-n-n_2 - p - a-N_t + 2}\hat{r}^{2(n-n_1 - \abs{b}-\abs{K})}\Lx\prod (x^i)^{A_i + C_i - 2q_i - 2N_i + \sum_{j\neq i}(B_{ji} - B_{ij})}\Rx\\&\cdot\hat{X}^{\abs{m} + n - n_1 - \abs{q} - \abs{b} + \sum_{i=1,2,3} (A_i - N_i - K_i)}\partial_t^{l + \abs{C} - \abs{m} - n - n_2 - p - a + A_t - N_t}f,
\end{align*}
as the number of $\hat{r}^2$ terms and $x^i$ terms remain relative to the number of $\hat{X}$ factors, but up to two $t\partial_i$ terms may lose their derivatives, leaving (at worst) an extra two factors of $t$. Then \cref{axisregkey} allows us to write these terms as acting as:
\begin{align*} &t^{l+\abs{m}-n-n_2 - p - a-N_t + s}\hat{r}^{2(n-n_1 - \abs{b}-\abs{K})}\Lx\prod (x^i)^{A_i + C_i - 2q_i - 2N_i + \sum_{j\neq i}(B_{ji} - B_{ij})}\Rx\\&\cdot\partial_{\hat{r}}^{2(\abs{m} + n - n_1 - \abs{q} - \abs{b} + \sum_{i=1,2,3} (A_i - N_i - K_i))}\partial_t^{l + \abs{C} - \abs{m} - n - n_2 - p - a + A_t - N_t}f,
\end{align*}
where $s \leq 2, s \leq \abs{m}$. From here the bounds (and corresponding order-on-order increases in decay) of \cref{main1} reduce this to a question of counting derivatives and the corresponding powers of weights. In this generic term, we find that we have
\[ \abs{m} + n + 2\abs{A} - 2\abs{N} - 2\abs{K} + l + \abs{C} + A_t - 2n_1 - n_2 - 2\abs{q} - 2\abs{b} - p - a - N_t \]
derivatives acting on a given term (we do not distinguish between $\hat{r}$ and $t$ derivatives since we will work with bounds in terms of $u$ and $v$, and each of these operators mix the two),
\[ l + \abs{m} - n - n_2 - p - a - N_t + 2 \]
powers of $t$ and
\[ 2n - 2n_1 - 2\abs{b} - 2\abs{K} + \abs{A} + \abs{C} - 2\abs{q} - 2\abs{N}  \]
powers of $\hat{r}$. Since $\hat{r}$ is at worst comparable to $t$ in the near axis region, we can combine these weights to get a total weight of
\[ n + l + \abs{m} + \abs{A} + \abs{C} - 2n_1 - n_2 - p - a - 2\abs{b} - 2\abs{K} - 2\abs{q} - 2\abs{N} - N_t + s, \]
and so the net difference between the total number of derivatives and total power of weights is $-2,-\abs{m} \leq \abs{A} + A_t - s$. Observe also that this $s$ only enters on our terms with some non-spherically symmetric part, and in this case, \cref{axisregextra} also gives us two extra differential orders, so our difference is actually just $\abs{A} + A_t$, as we might expect. The required bounds now follow immediately by the work done in the previous proposition, and the $u,v$ bounds given by \cref{main1}.

Away from the axis ($\hat{r} > t/2$) we must take a slightly different approach in order to obtain our proper decay. In particular we observe that we can write our differential operators (when acting on spherically symmetric terms) in terms of $u$ and $v$ as follows:
\[ t\partial_t + \hat{r}\partial_{\hat{r}} = \frac{(v+u)(\partial_v + \partial_u) + (v-u)(\partial_v - \partial_u)}{2} = v\partial_v + u \partial_u, \]
\[ t\partial_i + x^i\partial_t = \frac{S^i}{2}((v+u)(\partial_v - \partial_u) + (v-u)(\partial_v + \partial_u)) = S^i(v\partial_v - u\partial_u), \]
and we also have that $(x^i\partial_j - x^j \partial_i)S^i = -S^j$. Since the only non-spherically symmetric terms we will work with are products of the $x^i$ or $S^i$ these relations, along with the definition of the $\partial_i,\partial_t$ allows us to deal with all of our differential terms. In particular we see that any auxiliary powers of $v$ come along with an additional $v$ derivative, and likewise for $u$. Thus the bounds of \cref{main1}, and \cref{mainimprove} yield the required decay for \ref{disp2}, \ref{disp3}, and \ref{disp4}.

The case is similar for $\phi$, since this is spherically symmetric and satisfies similar (also sufficient) bounds by \cref{main1}.

Finally we must check \ref{disp8}. This is much the same as what we have done above for \ref{disp2} but requires a bit more care.

We begin in the near axis region, $\hat{r} < t/2$, with our $\Box_{\textbf{g}} t$ term. This is completely spherically symmetric, so we are again in the case of our generic near axis operator above acting on
\[ \Box_{\textbf{g}} t = \Omega^{-2} \frac{\hat{r}^2}{r^2} \partial_t \frac{r^2}{\hat{r}^2}. \]
In the proof of \cref{XboundCor} below we establish the bounds
\[ \abs{\hat{X}^l\partial_t^s\frac{\hat{r}}{r}}(u,v), \abs{\hat{X}^{l}\partial_t^s\frac{r}{\hat{r}}}(u,v) \lesssim \sup_{S_{u+v}(v-u)}\abs{\partial_{\hat{r}}^{2l}\partial_t^{s}(\lambda-\nu)}, \]
and control these derivatives of $\Omega^{-2}$ in terms of bounds for derivatives of $\lambda\nu$ and $\mu$. Moreover, in the previous proposition we have established behavior for non-differentiated terms, so this is simply a matter of applying our general form for the operator acting on spherically symmetric terms found above and counting worst-case order of decay. Doing this we find that we have exactly one extra derivative in comparison to our counting above, so in particular our decay improves by a power of $1+\abs{u}$, which is comparable to $1+v$ in the small $\hat{r}$ region we consider, and thus our bounds will hold in this case.

Away from the axis again things are less subtle, and we can apply our simpler general operator found above, differentiating na\"ively throughout and again count our decay. In this case we act on the expression $\Omega^{-2}\frac{\lambda+\nu}{r} = \frac{(1-\mu)(\lambda+\nu)}{4r\lambda\nu}$. As previously, we have already controlled all these terms in their undifferentiated state in the previous proposition, so, by our power counting above we need only check that a $\partial_u$ (resp. $\partial_v$) derivative acting on each term results in an improvement in decay of one power of $u$ (resp. $v$). Comparing our bounds of the previous proposition with those of \cref{main1} and the improvements of \cref{mainimprove} we find that this is the case, and thus the required estimates hold here as well.

Next we must address our spatial coordinate terms, and again we must deal with some non-spherically symmetric pieces. Recall that we have
\[ \Box_{\textbf{g}} x^j = \Omega^{-2}\frac{\hat{r}^2}{r^2}\partial_i \Lx \delta_{ij}\Omega^2 + \frac{1}{\hat{r}^2}\Lx \frac{r^2}{\hat{r}^2} - \Omega^{2} \Rx x^ix^j \Rx. \]
We can split this into three terms which we will deal with indiviually:
\begin{gather}
\label{term1} \Omega^{-2} \frac{\hat{r}^2}{r^2}\partial_i \Omega^{2},\\
\label{term2} \Omega^{-2}\frac{\hat{r}^2}{r^2}x^ix^j \partial_i \Lx \frac{1}{\hat{r}^2}\Lx \frac{r^2}{\hat{r}^2} - \Omega^{2} \Rx \Rx,\\
\label{term3} \Omega^{-2}\frac{\hat{r}^2}{r^2}x^j \Lx \frac{1}{\hat{r}^2}\Lx \frac{r^2}{\hat{r}^2} - \Omega^{2} \Rx \Rx.
\end{gather}

We begin with the near axis case. Here, \eqref{term1} is immediately controlled sufficiently by the power counting above and \cref{main1,mainimprove}, since we gain an extra order of derivative immediately ($\partial_i$ commutes with our other operators up to more strongly decaying terms), and this provides the extra order of decay required for \ref{disp8}. Similarly, writing $\partial_i = x^i\hat{X}$ we see that \eqref{term2} gains the required decay by our power counting (losing up to three spatial derivatives now), and \cref{Xboundavgcor}, which gives us an extra two derivatives acting, and thus an extra two powers of decay, for a total of one additional power as required. \eqref{term3} follows in the same manner (we lose our extra derivative, and one potential lost differential order from \eqref{term2}).

Finally, away from the axis we see that the required bound on \eqref{term1} follows directly from our power counting and \cref{main1}. Here it is easier to combine \eqref{term2} and \eqref{term3} and write them instead as
\[ \Omega^{-2}\frac{\hat{r}^2}{r^2}\partial_i\Lx S^iS^j \Lx \frac{r^2}{\hat{r}^2} - \Omega^{2} \Rx \Rx. \]
Then using $\hat{r} \sim v$ in this region, and the estimates already considered above the required bound is direct from power counting and \cref{main1},\cref{mainimprove}.

With this we establish all the conditions for \ref{dispsoln}, and so we obtain stability for our class of solutions.
\end{proof}

\section{Regularity Near the Axis}
\label{axisregsec}
In this section we prove several key results used above in the proof of \cref{main2}. In particular, we show that the lift of our reduced spherically symmetric solution to \eqref{SSESF} given above gives rise to a smooth solution to \eqref{ESF} with good decay in $(3+1)$ dimensions. Note that the only issue is due to the singularity of the $(u,v)$ coordinates along the axis of symmetry, thus away from the axis there is already nothing left to do. However, in order to guarantee control of spatial ($x^i$) derivatives near the axis, we must control the differential operator $\tilde{X} := \frac{1}{\hat{r}}\partial_{\hat{r}}$ across the axis (observe that away from the axis this is immediately controlled).

In order to do this we note that it suffices to establish estimates in $(1+1)$ dimensions for the operator $X := \frac{1}{\tilde{r}}\partial_{\tilde{r}}$ near the axis of symmetry. Thus we reduce again to the $(1+1)$ dimensional setting for the remainder of this section.

\subsection{Preparations}
Before beginning we note a few essential facts:
\begin{remark}
Corresponding to $(r,\phi,m)$ a solution to \eqref{SSESF} in $(1+1)$ there is a spherically symmetric solution $(\mathcal{M},\phi,\textbf{g})$ to \eqref{ESF} in $(3+1)$ dimensions which reduces to $(r,\phi,m)$. Moreover if the data is sufficiently smooth ($C^{\infty}$ certainly suffices), then by persistence of regularity and Sobolev embeddding the solution may be taken to be at least $C^k$, so $\phi$ and the components of the metric must be $C^k$ smooth in the smooth structure $\R^{3+1}$. In particular, reducing to $(1+1)$ dimensions, this immediately implies that $X^l\phi$ is well defined on $\R^{(1+1)}_{+}$, and bounded on the axis for $l \leq k$, and thus $\partial_{\hat{r}}^{2l+1}\phi\mid_{\Gamma} = 0$ for all $l \leq \lfloor \frac{k}{2} \rfloor -1$. This final fact will be essential to what follows.
\end{remark}

We will also make use of the following elementary proposition:
\begin{proposition}\label{evenprop1}
Let $\R^{(1+1)}_+ = \set*{(x,t) \in \R^2; x \geq 0}$, and let $f:\R^{(1+1)}_+ \rightarrow \R$ be $C^k$ up to the boundary. Suppose the extension $\bar{f}$ to all of $\R^{(1+1)}$ given by
\[ \bar{f}(x,t) = \begin{cases} f(x,t) & x \geq 0\\ f(-x,t) & x \leq 0 \end{cases} \]
is $C^k$ as well. Then $\lim_{x \rightarrow 0^+}\partial^{2l+1}_{x}f(x,t) = 0$ for all $l \leq \lfloor \frac{k}{2} \rfloor - 1$.
\end{proposition}


\subsection{Estimates For $\phi,\mu,\lambda - \nu$, and $\lambda\nu$}
In order to establish the necessary control of $\phi,\mu, \lambda-\nu$ and $\lambda\nu$ we will make use of the following lemma:
\begin{lemma}\label{axisregkey}
Let $f:\R^2_+ \rightarrow \R$ a $C^{2k}$ function such that $\partial_x^{2l-1}f\mid_{\Gamma} \equiv 0$ for all $l \leq k$. Then for $2l+s \leq 2k$:
\[ \abs{X^l\partial_t^sf}(t,x) \lesssim \sup_{S_{t}(x)}\abs{\partial_x^{2l}\partial_t^sf}. \]
\end{lemma}

\begin{proof}
Observe that since odd order derivatives of $f$ vanish on the axis, we can write
\[ X^l f(t,x) = \frac{1}{x^{2l-1}}\int_{\set*{(t',x');t' = t, x' \in [0,x]}}(X^{l-1}\partial^2_{x}f) (x')^{2l-2} dx'. \]
Since we have
\[ Xf = \frac{1}{x}\int \partial_x^2f dx' = I_x^1[\partial_xf], \]
and so applying applying our differentiation formula for averaging operators, and rewriting the integrand (multiplying and dividing by $x'$ to obtain extra $X$ operators and maintain the averaging operator form) we obtain the above expression.

Thus, integrating using the supremum bound for our $f$ term we have
\[ \abs{X^lf}(t,x) \lesssim \sup_{S_t(x)}\abs{X^{l-1}\partial^2_xf}, \]
so it suffices to bound $X^{l-1}\partial^2_xf$ by $\partial^{2l}_xf$. But $\partial^2_xf$ satisfies the same assumptions as $f$ with $k' = k-1$. Proceeding inductively, we conclude by the same reasoning as above that
\[ \abs{X^{l-s}\partial^{2s}_xf}(t,x) \leq \sup_{S_t(x)}\abs{X^{l-s-1}\partial^{2s+2}_xf}.  \]
Thus it follows that
\[ \abs{X^lf}(t,x) \lesssim \sup_{S_t(x)}\abs{\partial^{2l}_xf} \]
for all $l \leq k$.

Now observe that, since odd order derivatives of $f$ vanish on the axis, $\partial_x,\partial_t$ commute, and the axis is an integral curve of $\partial_t$, we in fact have that
\[ \partial_x^{2l-1}\partial_t^sf\mid_{\Gamma} \equiv 0 \]
for any $2l-1+s\leq k$. But then $\partial_t^sf$ satisfies the assumptions of our lemma, so by the above argument we in fact have
\[ \abs{X^l\partial_t^sf}(t,x) \lesssim \sup_{S_{t}(x)}\abs{\partial_x^{2l}\partial_t^sf}. \]
\end{proof}

\begin{corollary}\label{XboundCor}
Let $(r,\phi,m)$ a locally scattering solution to \eqref{SSESF} in $\mathcal{Q}$ with data asymptotically flat of order $\omega' \geq 2$ in $C^{2k}$ towards $\mathcal{I}^+$. Then the following bounds hold for all $l \leq k$:
\begin{gather}
\label{phiXbound} \abs{X^l\partial_t^s\phi} \lesssim \abs{\partial_{\hat{r}}^{2l}\partial_t^s \phi},\\
\label{muXbound} \abs{X^l\partial_t^s\mu} \lesssim \abs{\partial_{\hat{r}}^{2l}\partial_t^s \mu},\\
\label{lambdanuXbound} \abs{X^l\partial_t^s\lambda\nu} \lesssim \abs{\partial_{\hat{r}}^{2l}\partial_t^s \lambda\nu},\\
\label{lambdanusubXbound} \abs{X^l\partial_t^s(\lambda - \nu)} \lesssim \abs{\partial_{\hat{r}}^{2l}\partial_t^s (\lambda - \nu)}.
\end{gather}
\end{corollary}

\begin{proof}
Observe that \eqref{phiXbound} holds immediately by \cref{axisregkey} and \cref{main1}, since $\phi$ lifts to a smooth function on $\R^{(3+1)}$, and thus has $X^l\phi$ bounded for all $l \leq 2k$, thus immediately verifying the hypotheses of \cref{axisregkey}.

We turn next to \eqref{lambdanuXbound}, and \eqref{lambdanusubXbound}. Observe that the function $\bar{r}(u,v) := \begin{cases}r(u,v) & v \geq u\\ -r(v,u) & v \leq u \end{cases}$ is a smooth extension of $r$ to all of $\R^{(1+1)}$ (one checks directly that the derivatives in $u$ and $v$ match up across the axis so long as they are well defined in $\R^{(1+1)}_{+}$).
Moreover this function is odd in the $\hat{r}$ coordinate (this is exactly the condition in $u,v$ translated to these other coordinates), and thus it follows immediately that $\partial_{\tilde{r}}\bar{r} = \frac{1}{2}(\bar{\lambda} - \bar{\nu})$ is even. Thus $\lambda - \nu$ immediately satisfies the hypotheses of \cref{axisregkey}, and thus by \cref{main1} verifies \eqref{lambdanusubXbound}.

In the case of $\lambda\nu$, observe that $\bar{\lambda}\bar{\nu}$ is a smooth function on $\R^{(1+1)}$, and we have
\[ \bar{\lambda}(u,v) = -\bar{\nu}(v,u) \]
by construction of $\bar{r}$. Thus we have
\[ \bar{\lambda}\bar{\nu}(u,v) = (-\bar{\nu})(-\bar\lambda)(v,u) = \bar{\lambda}\bar{\nu}(v,u). \]
So $\bar{\lambda}\bar{\nu}$ is an even function, $C^{2k}$ extension of $\lambda\nu$. As above it follows by \cref{axisregkey} and \cref{main1} that $\lambda\nu$ satisfies \eqref{lambdanuXbound}.

We are left to deal with \eqref{muXbound}, which is rather more involved. Observe first that we can write:
\[ \mu(u,v) = \frac{\tilde{r}}{r}\frac{1}{\tilde{r}}\int_{S_{u+v}(v-u)}\partial_{\tilde{r}}md\tilde{r}'. \]
We will deal with the terms $\frac{\tilde{r}}{r}$ and $\frac{1}{\tilde{r}}\int_{S_{u+v}(v-u)}\partial_{\tilde{r}}md\tilde{r}'$ separately, beginning with the latter.

Employing our averaging operators we have
\[ \Lx\partial^l_{\tilde{r}}\frac{1}{\tilde{r}}\int_{S_{u+v}(v-u)}\partial_{\tilde{r}}md\tilde{r}'\Rx(u,v) = \frac{1}{\tilde{r}^{l+1}(u,v)}\int_{S_{u+v}(v-u)}\partial_{\tilde{r}}^{l+1}m(\tilde{r}')^{l}d\tilde{r}'. \]
We also have
\[ \partial_{\tilde{r}}m = \frac{(1-\mu)r^2}{4}\Lx \frac{1}{\lambda}(\partial_v\phi)^2 - \frac{1}{\nu}(\partial_u\phi)^2 \Rx. \]
Recall that $\partial_{\tilde{r}}^{2l-1}\phi$ vanishes on the axis. It follows that the extension $\bar{\phi}(u,v):= \begin{cases} \phi(u,v)& v\geq u\\ \phi(v,u) & u \geq v \end{cases}$ is a smooth, even extension of $\phi$ to all of $\R^{(1+1)}$. Moreover, we have that
\[ \partial_{u}\bar\phi(u,v) = \partial_v\bar{\phi}(v,u). \]
In particular the function
\[ \frac{1}{\bar\lambda}(\partial_v\bar\phi)^2 - \frac{1}{\bar\nu}(\partial_u\bar\phi)^2 \]
is a smooth extension of $\frac{1}{\lambda}(\partial_v\phi)^2 - \frac{1}{\nu}(\partial_u\phi)^2$, which is even, and thus has vanishing odd order $\partial_{\tilde{r}}$ derivatives on the axis. Since $r$ admits an odd extension, $r^2$ is a smooth even extension and thus also has vanishing odd order derivatives. Thus we conclude that, so long as $\partial_{\tilde{r}}^{l-1}(1-\mu) = \partial_{\tilde{r}}\mu$ vanishes along $\Gamma$, so does
\[ \frac{1}{\tilde{r}^{l+1}(u,v)}\int_{S_{u+v}(v-u)}\partial_{\tilde{r}}^{l+1}m(\tilde{r}')^{l}d\tilde{r}'. \]
In particular, so long as $\frac{\tilde{r}}{r}$ is well behaved we can conclude by induction that $\mu$ has vanishing odd order $\tilde{r}$ derivatives on the axis.

Now we must deal with $\frac{\tilde{r}}{r}$, as above we wish to show that the odd order $\tilde{r}$ derivatives vanish on the axis. Observe that since $\frac{\tilde{r}}{r}$ is bounded away from 0 (since $\lambda,\nu$ are bounded away from 0), it suffices to work with $\frac{r}{\tilde{r}}$. This is advantageous, as we can write
\[ \frac{r}{\tilde{r}}(u,v) = \frac{1}{\tilde{r}(u,v)}\int_{S_{u+v}(v-u)} \frac{1}{2}\Lx\lambda - \nu\Rx d\tilde{r}'. \]
Thus we have
\[ \abs{\partial_{\tilde{r}}^{l}\frac{r}{\tilde{r}}(u,v)} \lesssim \sup_{S_{u+v}(v-u)}\abs{\partial_{\tilde{r}}^{l}(\lambda-\nu)}. \]
So by our analysis of $\lambda-\nu$ above we conclude that every odd order derivative of $\frac{\tilde{r}}{r}$ vanishes on the axis. Thus the same holds for $\frac{r}{\tilde{r}}$. Moreover, we conclude by \cref{main1} that
\[ \abs{\partial_{\tilde{r}}^{l}\partial_t^s\frac{\tilde{r}}{r}}(u,v) \lesssim \sup_{S_{u+v}(v-u)}\abs{\partial_{\tilde{r}}^{l}\partial_t^{s}(\lambda-\nu)}, \]
counting powers in the bounds \eqref{main12}, \eqref{main14}.

Combining this with the above, we conclude inductively that $\partial_{\tilde{r}}^{2l-1}\mu$ vanishes on $\Gamma$ (since an odd number of derivatives must always act on one of the terms in our expression for $\mu$). Thus \eqref{muXbound} holds as well.
\end{proof}

Finally we would like to make use of the above bounds for the extended cases $\frac{\mu}{\tilde{r}^2}, \frac{(\lambda - \nu)-(\lambda_0 - \nu_0)}{\tilde{r}^2}$, and $\frac{\lambda\nu - \lambda_0\nu_0}{\tilde{r}^2}$ (where $\lambda_0(t,\tilde{r}) = \lambda(t,0)$, $\nu_0$ likewise). To this end we have the following lemma:
\begin{lemma}\label{axisregextra}
Suppose $f$ is as in \cref{axisregkey}, and moreover satisfies $f(t,0) = 0$. Then $\tilde{f} := \frac{f}{\tilde{r}^2}$ also satisfies the assumptions of \cref{axisregkey} and the bound:
\[ \abs{X^l\partial_t^s \tilde{f}} \lesssim \sup_{S_{t}(x)}\abs{\partial_x^{2l+2}\partial_t^sf}. \]
\end{lemma}

\begin{proof}
We check that the even extension of $\tilde{f}$ (which we also denote by $\tilde{f}$) is in fact differentiable across the axis by inductively controlling its derivatives in terms of those of $f$. First, we can write
\[ \tilde{f}(t,R) = \frac{1}{R^2}\int_0^R \partial_{\tilde{r}}f(t,r') dr' = \frac{1}{R^2}\int_0^R\int_0^{r'} \partial_{\tilde{r}}^2f(t,r'')dr''dr' \]
since $f$ and $\partial_{\tilde{r}}f$ each vanish on the axis. We thus conclude that
\[ \abs{\tilde{f}(t,R)} \leq \sup_{r < R}\abs{\partial_{\tilde{r}}^2f}, \]
and moreover (via the same argument with the value on the axis subtracted) if $f$ is $C^k$ for $k \geq 2$ (resp. $>2$) then $\tilde{f}$ is differentiable (resp. continuously differentiable) across the axis.

Suppose now that we have that $\abs{\partial_{\tilde{r}}^{n-1}\tilde{f}(t,R)} \leq \sup_{r < R}\abs{\partial_{\tilde{r}}^{n+1}f}$ for all $0 \leq n-1 < k-2$. We show that such a bound holds for at order $n$ as well. So we can write:
\[ \partial_{\tilde{r}}^{n}\tilde{f} = \partial_{\tilde{r}}^{n}\frac{f}{\tilde{r}^2} = \sum_{m = 0}^{n} (-1)^m\binom{n}{m}(m+1)! \frac{\partial_{\tilde{r}}^{n-m}}{\tilde{r}^{m+2}} = \frac{1}{\tilde{r}^{n+2}}\sum_{m = 0}^{n} (-1)^m\binom{n}{m}(m+1)! r^{n-m}\partial_{\tilde{r}}^{n-m}. \]
Observe that every term in this sum vanishes along the axis by assumptions on $f$, so we can differentiate each term and integrate to obtain
\begin{multline}
\frac{1}{R^{n+2}}\int_0^R \Bigg( (-1)^n(n+1)! \partial_{\tilde{r}}f(t,r)\\
 + \left.\sum_{m = 0}^{n-1} (-1)^m\binom{n}{m}(m+1)! \Lx (n-m)r^{n-m-1}\partial_{r}^{n-m}f(t,r) + r^{n-n}\partial_{\tilde{r}}^{n-m+1}f(t,r) \Rx\right) dr.
\end{multline}

Combining terms with equal powers of $r$ and differential order on $f$, we obtain
\[ \frac{1}{R^{n+2}} \int_0^R \sum_{m = 0}^n (-1)^m \frac{n!}{(n-m)!}r^{n-m}\partial_{\tilde{r}}^{n-m+1}f(t,r) dr. \]
Once again every term in this sum vanishes on the axis so we can differentiate and integrate once more to obtain
\begin{dmath}
\frac{1}{R^{n+2}}\int_0^R \int_0^r (-1)^n n! \partial_{\tilde{r}}^n f(t,r') + \sum_{m = 0}^{n-1} (-1)^m \frac{n!}{(n-m)!}\Lx (n-m)(r')^{n-m-1}\partial_{\tilde{r}}^{n-m+1}f(t,r') + (r')^{n-m}\partial_{\tilde{r}}^{n-m+2}f(t,r') \Rx dr'dr.
\end{dmath}
Every term in this expression cancels except for the highest differential order and we obtain:
\begin{equation}
\frac{1}{R^{n+2}} \int_0^R\int_0^r (r')^n\partial_{\tilde{r}}^{n+2}f(t,r')dr'dr.
\end{equation}
From this expression we immediately have the bound
\begin{equation}
\abs{\partial_{\tilde{r}}^{n}\tilde{f}(t,R)} \lesssim \sup_{r < R} \abs{\partial_{\tilde{r}}^{n+2}f(t,r)},
\end{equation}
and similar to the above we also conclude that $f$ is differentiable at order $n+1$ as well. Moreover, since odd order derivatives of $f$ vanish along the axis, the same is true of those of $\tilde{f}$, and thus we are in the situation of \cref{axisregkey}, and we obtain the required bounds.
\end{proof}

\begin{corollary}\label{Xboundavgcor}
\begin{gather}
\abs{X^l\partial_t^s\frac{\mu}{\tilde{r}^2}} \lesssim \abs{\partial_{\hat{r}}^{2l + 2}\partial_t^s \mu},\\
\abs{X^l\partial_t^s\frac{(\lambda-\nu) - (\lambda_0 - \nu_0)}{\tilde{r}^2}} \lesssim \abs{\partial_{\hat{r}}^{2l + 2}\partial_t^s (\lambda - \nu)},\\
\abs{X^l\partial_t^s\frac{\lambda\nu - \lambda_0\nu_0}{\tilde{r}^2}} \lesssim \abs{\partial_{\hat{r}}^{2l + 2}\partial_t^s \lambda\nu},\\
\abs{X^l\partial_t^s\Lx \frac{1}{\tilde{r}^2}\Lx \Omega^2 - \frac{r^2}{\tilde{r}^2} \Rx \Rx} \lesssim \abs{\partial_{\hat{r}}^{2l + 2}\partial_t^s \Lx \Omega^2 - \frac{r^2}{\tilde{r}^2} \Rx}.
\end{gather}
\end{corollary}

\section{Acknowledgements}
Many thanks to Jonathan Luk for guidance, many helpful discussions, and providing numerous references. I would also like to thank the SURIM program for helping to support this work as part of my undergraduate thesis.

\bibliography{./nonlin_stab.bib}{}

\providecommand{\bysame}{\leavevmode\hbox to3em{\hrulefill}\thinspace}
\providecommand{\MR}{\relax\ifhmode\unskip\space\fi MR }
\providecommand{\MRhref}[2]{%
  \href{http://www.ams.org/mathscinet-getitem?mr=#1}{#2}
}
\providecommand{\href}[2]{#2}
\begin{thebibliography}{10}

\bibitem{Chris86}
Demetrios Christodoulou, \emph{{The problem of a self-gravitating scalar
  field}}, Comm. Math. Phys. \textbf{105} (1986), no.~3, 337--361.

\bibitem{Chris87}
\bysame, \emph{{A mathematical theory of gravitational collapse}}, Comm. Math.
  Phys. \textbf{109} (1987), no.~4, 613--637.

\bibitem{Chris91}
\bysame, \emph{{The formation of black holes and singularities in spherically
  symmetric gravitational collapse}}, Comm. Pure Appl. Math. \textbf{44}
  (1991), no.~3, 339--373.

\bibitem{Chris93}
\bysame, \emph{{Bounded variation solutions of the spherically symmetric
  Einstein-scalar field equations}}, Comm. Pure Appl. Math. \textbf{46} (1993),
  no.~8, 1131--1220.

\bibitem{Chris94}
\bysame, \emph{{Examples of naked singularities in the gravitational collapse
  of a scalar field}}, Annals of Math (2) \textbf{140} (1994), no.~3, 604--653.

\bibitem{Chris99}
\bysame, \emph{{The instability of naked singularity formation in the
  gravitational collapse of a scalar field}}, Annals of Math (2) \textbf{149}
  (1999), no.~1, 183--217.

\bibitem{ChrisKlain}
Demetrios Christodoulou and Sergiu Klainerman, \emph{{The global nonlinear
  stability of the Minkowski space}}, Princeton Mathematical Series \textbf{41}
  (1993).

\bibitem{DafRod}
Mihailis Dafermos and Igor Rodnianski, \emph{{A proof of Price's law for the
  collapse of a self-gravitating scalar field}}, Invent. Math \textbf{162}
  (2005), no.~2, 381--457.

\bibitem{Dafermos2003}
Mihalis Dafermos, \emph{Black hole formation from a complete regular past},
  Communications in Mathematical Physics \textbf{289} (2009), no.~2, 579–596.

\bibitem{DafHolzRodTay}
Mihalis Dafermos, Gustav Holzegel, Igor Rodnianski, and Martin Taylor,
  \emph{{The non-linear stability of the Schwarzschild family of black holes}},
  arXiv: 2103.08222 (2021).

\bibitem{FriedStab}
Helmut Friedrich, \emph{{On the existence of n-geodesically complete or future
  complete solutions of Einstein’s field equations with smooth asymptotic
  structure}}, Comm. Math. Phys \textbf{107} (1986), no.~4, 587--609.

\bibitem{HintzVasy}
Peter Hintz and András Vasy, \emph{{The global non-linear stability of the
  Kerr–de Sitter family of black holes}}, Acta Mathematica \textbf{220}
  (2018), no.~1, 1–206.

\bibitem{KlainSzef}
Sergiu Klainerman and Jeremie Szeftel, \emph{{Global nonlinear stability of
  Schwarzschild spacetime under polarized perturbations}},
  arXiv:gr-qc/1711.07597 (2018).

\bibitem{LindRod}
Hans Lindblad and Igor Rodnianski, \emph{{The global stability of Minkowski
  space-time in harmonic gauge}}, Annals of Math (2) \textbf{171} (2010),
  no.~3, 1401--1477.

\bibitem{Luk20}
Jonathan Luk and Sung~Jin Oh, \emph{{Global nonlinear stability of large
  dispersive solutions to the Einstein equations}}, Currently Unpublished,
  Preprint.

\bibitem{Luk2015}
\bysame, \emph{{Quantitative decay rates for dispersive solutions to the
  Einstein-scalar field system in spherical symmetry}}, Analysis and PDE
  \textbf{8} (2015), no.~7, 1603--1674.

\bibitem{luk2019strong}
Jonathan Luk and Sung-Jin Oh, \emph{{Strong cosmic censorship in spherical
  symmetry for two-ended asymptotically flat initial data I. The interior of
  the black hole region}}, arXiv:gr-qc/1702.05715 (2019).

\bibitem{Luk2018}
Jonathan Luk, Sung-Jin Oh, and Shiwu Yang, \emph{{Solutions to the
  Einstein-scalar-field system in spherical symmetry with large bounded
  variation norms}}, Annals of PDE \textbf{4} (2018), no.~1.

\bibitem{Price}
Richard~H. Price, \emph{{Nonspherical perturbations of relativistic
  gravitational collapse. I. Scalar and gravitational perturbations}}, Phys.
  Rev. D \textbf{5} (1972), 2419--2438.

\end{thebibliography}
\bibliographystyle{amsplain}
\end{document}